\documentclass[a4paper]{lipics-v2021margins}
\usepackage{microtype}
\usepackage{amsthm}
\usepackage{xypic}
\usepackage{nicefrac}
\usepackage{microtype,comment}
\usepackage{enumerate}
\usepackage{boxedminipage}
\usepackage{hyperref}
\usepackage{amssymb,amsmath}
\usepackage{slashbox}
\usepackage{comment}
\usepackage[T1]{fontenc}
\usepackage{graphicx}
\usepackage{comment}
\usepackage{twoopt}

\usepackage{tikz}
\usetikzlibrary{arrows}
\usetikzlibrary{shapes}
\usetikzlibrary{shapes}
\usetikzlibrary{decorations.pathreplacing}
\usetikzlibrary{hobby, positioning, shadows} 
\tikzset{
  vert/.style={circle, draw=black!100,fill=black!100,thick, inner sep=0pt, minimum size=.8mm},
  empty/.style={draw=none, fill=none, minimum size=0mm, inner sep=0pt}} 
\tikzset{
  vert/.style={circle, draw=black!100,fill=black!100,thick, inner sep=0pt, minimum size=2mm}, 
  smallvert/.style={circle, draw=black!100,fill=black!100,thick, inner sep=0pt, minimum size=1mm},
  bigvert/.style={circle, draw=black!100,fill=black!100,thick, inner sep=0pt, minimum size=4mm},
  bigring/.style={fill = none, thick, minimum size = 8mm},
  fillring/.style={fill = red, thick, minimum size = 8mm, opacity = .5},
  labvert/.style={circle, draw=black!100,fill=none,thick, inner sep=2pt, minimum size=2mm}, 
  empty/.style={draw=none, fill=none, minimum size=0mm, inner sep=0pt},
  arc/.style = {->,> = latex'}
}

\newtheorem{fact}{Fact}
\newtheorem*{subclaim*}{Subclaim}
\newcounter{ctrclaim}[theorem]
\newcounter{ctrcase}[theorem]

 \newcommand{\Verts}[2][,]{\foreach \i/\x/\y in {#2}{\draw (\x#1\y) node (\i){};}}
 \newcommand{\Edges}[2][]{\foreach \i/\j in {#2}{\draw (\i) edge[#1](\j);}}

\newcommand{\NP}{{\sf NP}}
\newcommand{\cP}{{\sf P}}

\renewcommand{\H}{\ensuremath{\mathbb{H}}}
\newcommand{\app}{$\spadesuit$}

\newcommand{\LabVerts}[2][]{\foreach \i/\x/\y in {#2}{\draw (\x,\y) node[labvert] (#1\i){$\i$};}}

\newcommand{\Vlabel}[4][.3cm]{\draw node[empty, #3 of = #2, node distance = #1] () {$#4$};}
\newcommand{\Vlabels}[2][.3cm]{\foreach \ver/\pos/\lab in {#2}{\Vlabel[#1]{\ver}{\pos}{\lab};}}

\def\romannum{\begingroup
  \def\theenumi{\textup{(\roman{enumi})}}
  \def\p@enumi{}
  \def\labelenumi{\theenumi}
  \enumerate}

\title{Complexity Framework for Forbidden Subgraphs II:\\ Edge Subdivision and the ``H''-graphs} 
 
\titlerunning{Complexity Framework For Forbidden Subgraphs II}

\author{Vadim Lozin}{University of Warwick, Coventry, United Kingdom}{v.lozin@warwick.ac.uk}{}{}
\author{Barnaby Martin}{Durham University, Durham, United Kingdom}{barnaby.d.martin@durham.ac.uk}{}{}
\author{Sukanya Pandey}{Utrecht University, Utrecht, The Netherlands}{s.pandey1@uu.nl}{}{}
\author{Dani\"el Paulusma}{Durham University, Durham, United Kingdom}{daniel.paulusma@durham.ac.uk}{0000-0001-5945-9287}{}
\author{Mark Siggers}{Kyungpook National University, South Korea}{mhsiggers@knu.ac.kr}{}{}
\author{Siani Smith}{University of Bristol and Heilbronn Institute for Mathematical Research, Bristol, United Kingdom}{siani.smith@bristol.ac.uk}{}{}
\author{Erik Jan van Leeuwen}{Utrecht University, Utrecht, The Netherlands}{e.j.vanleeuwen@uu.nl}{0000-0001-5240-7257}{}

\authorrunning{V. Lozin et al.}



\nolinenumbers
\hideLIPIcs  

\begin{document}
\maketitle

\begin{abstract}
For a fixed set ${\cal H}$ of graphs, a graph $G$ is ${\cal H}$-subgraph-free if $G$ does not contain any $H \in {\cal H}$ as a (not necessarily induced) subgraph. A recently proposed framework gives a complete classification on ${\cal H}$-subgraph-free graphs (for finite sets ${\cal H}$) for problems that are solvable in polynomial time on graph classes of bounded treewidth, \NP-complete on subcubic graphs, and whose \NP-hardness is preserved under edge subdivision. While a lot of problems satisfy these conditions, there are also many problems that do not satisfy all three conditions and for which the complexity in ${\cal H}$-subgraph-free graphs is unknown.
We study problems for which only the first two conditions of the framework hold (they are solvable in polynomial time on classes of bounded treewidth and \NP-complete on subcubic graphs, but \NP-hardness is not preserved under edge subdivision). In particular, we make inroads into the classification of the complexity of four such problems: {\sc Hamilton Cycle}, {\sc $k$-Induced Disjoint Paths}, {\sc $C_5$-Colouring} and {\sc Star $3$-Colouring}. Although we do not complete the classifications, we show that the boundary between polynomial time and \NP-complete differs among our problems and also from problems that do satisfy all three conditions of the framework, in particular when
we forbid certain subdivisions of the ``H''-graph (the graph that looks like the letter ``H''). Hence, we exhibit a rich complexity landscape among problems for ${\cal H}$-subgraph-free graph classes.
\end{abstract}

\newpage

\section{Introduction}\label{s-intro}

Graph containment relations, such as the (topological) minor and induced subgraph relations, have been extensively studied both from a graph-structural and algorithmic point of view. In this paper, we focus on the {\it subgraph relation}.
If a graph $H$ can be obtained from a graph $G$ by a sequence of vertex deletions and edge deletions, then $G$ contains $H$ as a {\it subgraph}; otherwise, $G$ is \emph{$H$-subgraph-free}.
For a  set of graphs ${\cal H}$, a graph $G$ is {\it ${\cal H}$-subgraph-free} if $G$ is $H$-subgraph-free for every $H\in {\cal H}$; if ${\cal H}=\{H_1,\ldots,H_p\}$, then we also write that $G$ is $(H_1,\ldots,H_p)$-subgraph-free. 
Graph classes closed under deletion of edge and vertices are called {\it monotone}~\cite{ABKL07,BL02}, and every monotone graph class ${\cal G}$ can be characterized by a unique (and possibly infinite) set of forbidden induced subgraph ${\cal H}_{\cal G}$. We determine the complexity of two connectivity problems {\sc Hamilton Cycle} and {\sc $k$-Induced Disjoint Paths}, and two colouring problems {\sc $C_5$-Colouring} and {\sc Star $3$-Colouring} on ${\cal H}$-subgraph-free graphs for various families ${\cal H}$. We focus on families~${\cal H}$ consisting of certain {\it subdivided ``H''-graphs} ${\mathbb H}_i$, where $\mathbb{H}_1$ looks like the letter ``H'' (see Fig.~\ref{f-st} for the definition and illustration of the graphs ${\mathbb H}_i$). At first sight, these four problems appear to have not much in common, and the graphs ${\mathbb H}_i$ might also seem arbitrary. However,  these problems turn out to be well suited for a combined study, as they fit in a more general framework, in which the graphs $\mathbb{H}_i$ play a crucial role.

\smallskip
\noindent
{\bf Context.} If a graph problem is computationally hard, it is natural to restrict the input to some special graph class.
Ideally we would like to know exactly which properties $P$ such a graph class ${\cal G}$ must have such that any hard graph problem that satisfies some conditions $C$ becomes easy on graphs from ${\cal G}$ (here, the distinction between ``easy'' and ``hard'' could for example mean $\cP$ versus \NP-complete, or almost-linear versus at-least-quadratic). We first discuss some natural conditions $C$ a graph problem $\Pi$ might satisfy.

A graph is {\it subcubic} if every vertex has degree at most~$3$. For $p\geq 1$, the {\it $p$-subdivision} of an edge $e=uv$ of a graph~$G$ replaces $e$ by a path of $p+1$ edges with endpoints $u$ and~$v$. 
The {\it $p$-subdivision} of a graph~$G$ is the graph obtained from~$G$ after $p$-subdividing each edge; see also Fig.~\ref{f-st}.
For a graph class ${\cal G}$ and an integer~$p$, we let ${\cal G}^p$ be the class consisting of the $p$-subdivisions of the graphs in ${\cal G}$.
 A graph problem $\Pi$ is hard {\it under edge subdivision of subcubic graphs} if for every $j \geq 1$ there is an~$\ell \geq j$ such that:
if $\Pi$ is hard for the class ${\cal G}$ of subcubic graphs, then $\Pi$ is hard for ${\cal G}^{\ell}$.
We can now say that a graph problem~$\Pi$ has property:\\[-9pt]
\begin{itemize}
\item {C1} if $\Pi$ is easy for every graph class of bounded tree-width;\\[-12pt]
\item C2 if $\Pi$ is hard for subcubic graphs ($K_{1,4}$-subgraph-free graphs, where $K_{1,4}$ is the $5$-vertex star);\\[-12pt]
\item {C3} if $\Pi$ is  hard under edge subdivision of subcubic graphs;\\[-12pt]
\item {C4} if $\Pi$ is hard for planar graphs;\\[-12pt]
\item {C5} if $\Pi$ is hard for planar subcubic graphs.\\[-9pt]
\end{itemize}
We say that $\Pi$ is a C123-problem if it satisfies C1, C2 and C3, while $\Pi$ is a C14-problem if it satisfies C1 and C4, and so on. Classical results of Robertson and Seymour~\cite{RS86} yield the following two meta-classifications. 
For all sets~${\cal H}$, if ${\cal H}$ contains a planar graph, then every C14-problem $\Pi$ is easy on ${\cal H}$-minor-free graphs, or else $\Pi$ is hard. For all sets~${\cal H}$, if ${\cal H}$ contains a planar subcubic graph, then
every C15-problem~$\Pi$ is easy on ${\cal H}$-topological-minor-free graphs, or else $\Pi$ is hard. No meta-classification for the induced subgraph relation exists (apart from a limited one~\cite{JMOPPSV} that is a direct consequence of the treewidth dichotomy~\cite{LR22}). However, for the subgraph relation, known results on  {\sc Independent Set}~\cite{AK92}, {\sc Dominating Set}~\cite{AK92}, {\sc Long Path}~\cite{AK92}, {\sc Max-Cut}~\cite{Ka12} and {\sc List Colouring}~\cite{GP14}
for monotone graph classes that are {\it finitely defined} (so, where the associated set of forbidden subgraphs ${\cal H}$ is finite) were recently unified and extended in~\cite{JMOPPSV}. This led to a new meta-classification, where the set ${\cal S}$  consists of all graphs, in which every connected component is either a path or a subcubic tree with exactly one vertex of degree~$3$ (see Fig.~\ref{f-st}).

\begin{theorem}[\cite{JMOPPSV}]\label{t-dicho}
For any \emph{finite} set of graphs ${\cal H}$, a \emph{C123}-problem $\Pi$  is easy on ${\cal H}$-subgraph-free graphs if ${\cal H}$ contains a graph from ${\cal S}$, or else it is hard.
\end{theorem}

\noindent
In~\cite{JMOPPSV} a list of 25 C123-problems was given that include, apart from the five problems above, other well-known partitioning, covering, packing, network design, width parameter and distance metric problems.

\noindent
{\bf Our Focus.} Many graph problems  are not C123. See~\cite{BHKKOOZ20} and~\cite{BJMOPPSV} for partial complexity classifications of the C\hspace*{-1mm}$\not 1$23-problems
{\sc Subgraph Isomorphism} and {\sc Steiner Forest}, respectively, for $H$-subgraph-free graphs and~\cite{JMPPSV23} for partial complexity classifications of the C1\hspace*{-1mm}$\not 2$3-problems 
{\sc (Independent) Feedback Vertex Set}, {\sc Connected Vertex Cover}, {\sc Colouring} (see also~\cite{GPR15}) and {\sc Matching Cut} for $H$-subgraph-free graphs
(note that if a problem does not satisfy C2, then C3 is implied). Here, we consider the question:

\smallskip
\noindent
{\it Can we classify the complexity of C12\hspace*{-0.7mm}$\not 3$-problems (so that do not satisfy C3) on monotone graph classes?}

\begin{figure}[t]
\begin{minipage}[c]{0.48\textwidth}
\centering
\includegraphics[scale=0.7]{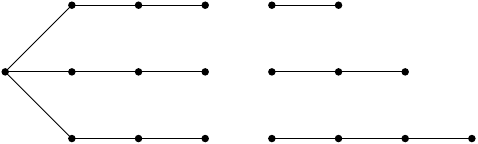}
\end{minipage}
\hspace{1.2cm}
\begin{minipage}[c]{0.42\textwidth}
\centering
\includegraphics[scale=0.4]{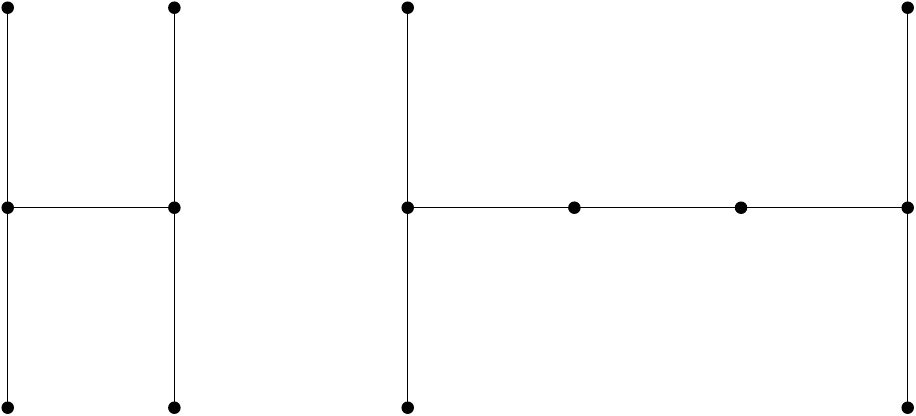}
\end{minipage}
\caption{\cite{BJMOPPSV} Left: A graph in ${\cal S}$: the graph $S_{3,3,3}+P_2+P_3+P_4$; 
note that $S_{3,3,3}$ is the $2$-subdivision of the claw $K_{1,3}$.
Right: the graphs $\mathbb{H}_1$ and $\mathbb{H}_3$; here, $\mathbb{H}_1$ is the ``H''-graph, formed by an edge (the {\it middle edge}) joining the middle vertices of two $P_3$s, and $\mathbb{H}_i$ ($i\geq 2$) is obtained from $\mathbb{H}_1$ by $(i-1)$-subdividing the middle edge.}\label{f-st}
\vspace*{-0.25cm}
\end{figure} 

\medskip
\noindent
{\bf Why the Graphs ${\mathbf{\H_i}}$.} 
All C1-problems are easy on ${\cal H}$-subgraph-free graphs if ${\cal H}$ has a graph from ${\cal S}$~\cite{RS84}.  
The infinite set ${\cal M}=\{C_3,C_4,\ldots,K_{1,4},\H_1,\H_2,\ldots\}$ of minimal graphs not in~${\cal S}$ is a maximal antichain in the poset of connected graphs under the subgraph relation.
Conditions C2 and C3 ensure that for every finite set ${\cal M}'$, C123-problems are hard on ${\cal M}'$-subgraph-free graphs  if ${\cal M}'\subseteq {\cal M}$.
If C3 is not satisfied, this is no longer guaranteed. Therefore, a natural starting point to answer our research question is to determine for which finite subsets ${\cal M}' \subseteq {\cal M}$, C12-problems are still easy on ${\cal M}'$-subgraph-free graphs.
So consider a C12-problem $\Pi$ that is not C3.  Let ${\cal M}'$ be a finite subset of ${\cal M}$. If ${\cal M}'=\{K_{1,4}\}$, then $\Pi$ is hard for  ${\cal M}'$-subgraph-free~graphs due to C2. Hence, ${\cal M}'$ must contain at least one $C_s$ or $\H_i$. Let the {\it girth} of a graph that is not a forest be the length of a shortest induced cycle in it. We say that $\Pi$ has property:\\[-9pt]
\begin{itemize}
\item C2' if for all $g\geq 4$, $\Pi$ is hard for subcubic graphs of girth $g$ ($(K_{1,4},C_3,\ldots,C_{g-1})$-subgraph-free graphs).\\[-9pt]
\end{itemize}
So if $\Pi$ is not only a C12-problem but even a C12'-problem, then $\Pi$ is hard on ${\cal M}'$-subgraph-free graphs unless ${\cal M'}$ contains some $\H_i$. This makes studying the graphs $\mathbb{H}_i$ even more pressing.

\medskip
\noindent
{\bf Our Testbed Problems.}
To address our research question, we take, as mentioned, four testbed problems: 

\smallskip
\noindent
{\sc Hamilton Cycle}. This  problem is to decide if a graph~$G$ has a {\it Hamiltonian cycle}, i.e., a cycle containing all vertices of $G$.
The problem is polynomial-time solvable for graphs of bounded treewidth~\cite{AP89} and \NP-complete for bipartite subcubic graphs of girth $g$, for every $g\geq 3$~\cite{ABKL07}.
Hence, it is even a C12'-problem.

\smallskip
\noindent
{\sc $k$-Induced Disjoint Paths}.  Given a graph $G$ and pairwise disjoint vertex pairs $(s_1, t_1), (s_2,t_2), \ldots (s_k, t_k)$ for some $k\geq 2$, this problem is to decide if $G$ has $k$ {\it mutually induced} $s_i$-$t_i$-paths $P^i$, i.e., $P^1,\ldots, P^k$ are pairwise vertex-disjoint and there are no edges between vertices from different $P^i$ and $P^j$.
The problem is polynomial-time solvable for graphs of bounded treewidth by Courcelle's Theorem~\cite{Co90} and \NP-complete for subcubic graphs for all $k\geq 2$~\cite{LLMT09}.
Hence, it is a C12-problem for all $k\geq 2$.

\smallskip
\noindent
 {\sc $C_5$-Colouring}. This problem is to decide if a graph $G$ has a {\it homomorphism} to the $5$-cycle $C_5$ (also called a {\it $C_5$-colouring}) which is a mapping $f:V(G)\to V(C_5)$ such that for every two vertices $u,v$ it holds that $f(u)f(v)\in E(C_5)$ whenever  $uv\in E(G)$. The problem is polynomial-time solvable for graphs of bounded treewidth~\cite{DP89} and \NP-complete for subcubic graphs~\cite{GHN00}. Hence, it is a C12-problem.
 
\smallskip
\noindent
{\sc Star $3$-Colouring}. This problem is to decide if a graph $G$ has a {\it star $3$-colouring}, which is a mapping $f:V(G)\to \{1,2,3\}$ such that for every $i$, the set $U_i$ of vertices of $G$ mapped to $i$ is independent (so, $f$ is a {\it $3$-colouring}) and $U_1\cup U_2$, $U_1\cup U_3$, $U_2\cup U_3$ all induce a disjoint union of stars. The problem is polynomial-time solvable for graphs of bounded treewidth due to Courcelle's Theorem~\cite{Co90} and \NP-complete for bipartite planar subcubic graphs of girth at least $g$, for every $g\geq 3$~\cite{SA22}. 
Hence, it is even a C12'-problem.

\medskip
\noindent
We do not know if {\sc $k$-Induced Disjoint Paths} and  {\sc $C_5$-Colouring} are C12', 
even though {\sc $C_5$-Colouring} is \NP-complete for graphs of maximum degree $6\cdot 5^{13}$ and girth at least $g$, for all $g\geq 3$ (see Appendix~\ref{a-degree}).

All four problems violate C3.
For $p\geq 3$,  {\sc $C_5$-Colouring} and {\sc Star $3$-Colouring} become 
 true 
under $p$-subdivision, while
 {\sc Hamilton Cycle} becomes 
 false 
 (unless we started with a cycle), and {\sc $k$-Induced Disjoint Paths} reduces to the polynomial-time solvable problem {\sc $k$-Disjoint Paths}~\cite{RS95,Sh80},
 which only requires the paths in a solution to be pairwise vertex-disjoint. See Appendix~\ref{a-c3}. We also note the following. First, when~$k$ is part of the input, {\sc Disjoint Paths} and {\sc Induced Disjoint Paths} are C123-problems \cite{JMOPPSV}.
Second, instead of $C_5$-{\sc Colouring} we could have considered $C_{2i+1}$-{\sc Colouring}, which is a C12-problem for all $i\geq 2$~\cite{DP89,GHN00}. Third, {\sc Star-$k$-Colouring} does not satisfy C2 for large $k$; see Appendix~\ref{a-star}.

\begin{figure}[t]
\vspace*{-0.6cm}
	\begin{center}
	\scalebox{.75}{
		\begin{picture}(150,40)
			\put(-10,5){\circle*{4}} 
			\put(20,5){\circle*{4}}
			\put(50,25){\circle*{4}}
			\put(50,5){\circle*{4}} 
			\put(80,5){\circle*{4}} 
			\put(20,25){\circle*{4}}
			\put(110,5){\circle*{4}}
			\put(140,5){\circle*{4}}
			\put(110,25){\circle*{4}}
			\put(110,5){\line(1,0){30}} 		
			\put(-8,5){\line(1,0){26}} 
			\put(22,5){\line(1,0){26}}
			\put(50,5){\line(1,0){30}} 
			\put(80,5){\line(1,0){30}}
			\put(50,5){\line(0,1){20}} 
			\put(20,5){\line(0,1){20}} 
			\put(110,5){\line(0,1){20}} 
		\end{picture}}		
	\end{center}
	\vspace*{-0.8cm}
	\caption{The tree $T$.}\label{fig:H12}
	\vspace*{-0.45cm}
\end{figure}

\smallskip
\noindent
{\bf Our Results.}
We show that the complexity of our four problems differ from each other and also from C123-problems, when we forbid certain graphs $\H_i$.
We first show that C1-problems, and thus C12-problems, are easy on $(\H_{\ell},\H_{\ell+1},\ldots)$-subgraph-free graphs for every $\ell\geq 1$ and on $(\H_i,\H_{2i},\H_{3i},\ldots)$-subgraph-free graphs for every  $i\geq 1$ (so, in particular if we forbid all even $\H_i$), as all these graph classes have bounded treewidth (see Appendix~\ref{a-h}).
In contrast, any hard problem for bipartite graphs in which one partition class has maximum degree~$2$ is hard on $(\H_1,\H_3,\ldots)$-subgraph-free graphs (so, if we forbid all odd $\H_i$): every path between vertices of degree at least~$3$ has even length. The \NP-hardness reduction in~\cite{ACKKR04} shows that  {\sc Star $3$-Colouring} is such a problem (see Appendix~\ref{a-al}).
We complement all these results as follows.

In Section~\ref{s-hc} we show that {\sc Hamilton Cycle} is polynomial-time solvable for $\H_\ell$-subgraph-free graphs for $\ell\in \{1,2,3\}$ by doing this for the superclass of $T$-subgraph-free graphs ($T$ is the tree shown in Figure~\ref{fig:H12}).

In Section~\ref{s-k} we prove that for all $k\geq 2$, {\sc $k$-Induced Disjoint Paths} is polynomial-time solvable~for $\H_\ell$-subgraph-free graphs for $\ell\in\{1,2\}$, but \NP-complete for subcubic $(\H_4,\ldots,\H_\ell)$-subgraph-free graphs for all~$\ell\geq 4$.
For the first result, we first apply the algorithm for {\sc $k$-Disjoint Paths}~\cite{RS95}. If this yields a solution that is not mutually induced, we apply a reduction rule and repeat the process on a smaller instance.
For the second result, we carefully adapt the proof of~\cite{LLMT09}  that shows that the problem of deciding if a subcubic graph contains an induced cycle between two given degree~$2$-vertices is~\NP-complete.

In Section~\ref{s-c5} we determine all  {\it $C_5$-critical} $\H_3$-subgraph-free graphs, which are not $C_5$-colourable unlike every proper subgraph of them. We show that this leads to a polynomial-time algorithm for $\H_3$-subgraph-free graphs that is even {\it certifying}. In contrast, the problem is \NP-complete for
the ``complementary'' class of $(\H_1,\H_2,\H_4,\H_5,\H_7,\H_8,\ldots)$-subgraph-free graphs (see Appendix~\ref{a-c5easy}).

In Section~\ref{s-star} we 
give a linear-time {\it certifying} algorithm for {\sc Star $3$-Colouring} on $(\H_2,\H_4,\H_6\ldots)$-subgraph-free graphs, after determining all star-$3$-colourable bipartite $(\H_2,\H_4,\H_6\ldots)$-subgraph-free~graphs.

Combining the above results yields the following four state-of-the-art summaries.\\[-5mm] 

\begin{theorem}\label{thm:main-hamilton}
{\sc Hamilton Cycle} is polynomial-time solvable for $(\H_{\ell},\H_{\ell+1},\ldots)$-subgraph-free graphs $(\ell\geq 1)$, for
$(\H_i,\H_{2i},\H_{3i},\ldots)$-subgraph-free graphs $(i\geq 1)$ and for $T$-subgraph-free graphs, and thus for $\H_\ell$-subgraph-free graphs $(\ell\in \{1,2,3\})$.\\[-5mm] 
\end{theorem}

\begin{theorem}\label{thm:main-kIDP}
For all $k\geq 2$, {\sc $k$-Induced Disjoint Paths} is polynomial-time solvable for $\H_\ell$-subgraph-free graphs  $(\ell\in\{1,2\})$, for $(\H_{\ell},\H_{\ell+1},\ldots)$-subgraph-free graphs $(\ell\geq 1)$ and for $(\H_i,\H_{2i},\H_{3i},\ldots)$-subgraph-free graphs $(i\geq 1)$, but \NP-complete for
subcubic 
 $(\H_4,\ldots,\H_\ell)$-subgraph-free graphs $(\ell\geq 4)$.\\[-5mm] 
\end{theorem}

\begin{theorem}\label{thm:main-C5-col}
{\sc $C_5$-Colouring} is polynomial-time solvable for $\H_3$-subgraph-free graphs, for $(\H_{\ell},\H_{\ell+1},\ldots)$-subgraph-free graphs $(\ell\geq 1)$ and for $(\H_i,\H_{2i},\H_{3i},\ldots)$-subgraph-free graphs $(i\geq 1)$,
but \NP-complete for $(\H_1,\H_2,\H_4,\H_5,\H_7,\H_8,\ldots)$-subgraph-free graphs.\\[-5mm] 
\end{theorem}

\begin{theorem}\label{thm:main-star}
{\sc Star $3$-Colouring} is polynomial-time solvable for $(\H_{\ell},\H_{\ell+1},\ldots)$-subgraph-free graphs $(\ell\geq 1)$ and $(\H_i,\H_{2i},\H_{3i},\ldots)$-subgraph-free graphs $(i\geq 1)$, but \NP-complete for  $(\H_1,\H_3,\H_5\ldots)$-subgraph-free~graphs.\\[-5mm] 
\end{theorem}

\noindent
We note that the complexity classifications above indeed differ except perhaps for {\sc Hamilton Cycle} and {\sc $k$-Induced Disjoint Paths}. Hence, Theorems~\ref{thm:main-hamilton}--\ref{thm:main-star} give clear evidence of a rich landscape for C12-problems on ${\cal H}$-subgraph-free graphs.
In Section~\ref{s-con} we discuss open problems resulting from our study.

\section{Hamilton Cycle}\label{s-hc}

Before proving Theorem~\ref{t-hhh} we note that there are trees $T^*$ for which {\sc Hamilton Cycle} is \NP-complete over $T^*$-subgraph-free graphs. For example, take a complete binary tree with depth $3$ (four layers), or any graph that is not a caterpillar with hairs of arbitrary length (see Section 3.2 in \cite{KLMT11}). Proofs of missing claims in the proof of Theorem~\ref{t-hhh} are in Appendix~\ref{a-hhh} (throughout our paper, proofs of claims, lemmas and theorems marked with an \app\ are in the appendix).

\begin{theorem}\label{t-hhh}
	{\sc Hamilton Cycle} is polynomial-time solvable for $T$-subgraph-free graphs.
\end{theorem}

\begin{proof}
	Let $G$ be a $T$-subgraph-free graph. We call vertices of degree $2$ in $G$ {\it white}
	and vertices of degree at least $3$ {\it black}. The {\it black graph} is a subgraph of $G$ induced by black
	vertices and a {\it black component} is a connected component in the black graph.
	
	We first describe some helpful rules to solve the problem and 
	a set of reductions simplifying the input graph, i.e. reductions transforming $G$ into 
	a graph $G'$ that has fewer edges and/or vertices and that has a Hamiltonian cycle if and only if $G$ has. 
	We emphasize that by deleting an edge or a vertex from an $H$-subgraph-free graph, we obtain 
	an $H$-subgraph-free graph again.
	We start with some obvious rules:\\[-10pt]
	\begin{itemize}
		\item[(R1)] if the graph has vertices of degree $0$ or $1$, then stop: $G$ has no Hamiltonian cycle.
		\item[(R2)] if the graph contains a vertex adjacent to more than two white vertices, 
		then stop: $G$ has no Hamiltonian cycle.
		\item[(R3)] if the graph is disconnected, then stop: $G$ has no Hamiltonian cycle.
		\item[(R4)] if the graph contains a vertex $v$ adjacent to exactly two white vertices,
		then delete the edges connecting $v$ to all other its neighbours (if there are any).\\[-10pt]
	\end{itemize} 	
Now we introduce a reduction applicable to a graph $G$ containing an induced 
subgraph shown on the left in Figure~\ref{fig:diamond}, in which vertices $a,b,c$
have degree $3$ in $G$. The reduction depends on the degree of $x$. If 
the degree of $x$ is also $3$, the reduction consists in deleting
the edges $ab$ and $xc$. Otherwise, it transforms the graph as shown
in Figure~\ref{fig:diamond}. We refer to this reduction as 
the {\it diamond reduction} and denote it by~(R5).  
	
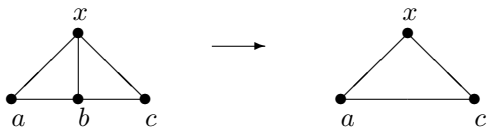
\begin{figure}[ht]
	\begin{center}
		\begin{picture}(120,50)
			\put(25,5){\circle*{4}}
			\put(50,30){\circle*{4}}
			\put(50,5){\circle*{4}} 
			\put(75,5){\circle*{4}}
			
			\put(25,5){\line(1,0){25}}
			\put(50,5){\line(1,0){25}}
			
			\put(50,30){\line(0,-1){25}}
			\put(50,30){\line(-1,-1){25}}
			\put(50,30){\line(1,-1){25}}
			
			\put(100,25){\vector(1,0){20}} 
			\put(25,-5){$a$}
			\put(50,-5){$b$}
			\put(75,-5){$c$}
			\put(48,35){$x$}
		\end{picture}
		\begin{picture}(100,50)
			\put(25,5){\circle*{4}}
			\put(50,30){\circle*{4}}
			\put(75,5){\circle*{4}}
			
			\put(25,5){\line(1,0){25}}
			\put(50,5){\line(1,0){25}}
			
			\put(50,30){\line(-1,-1){25}}
			\put(50,30){\line(1,-1){25}}
			
			\put(25,-5){$a$}
			\put(75,-5){$c$}
			\put(48,35){$x$}\end{picture}
	\end{center}
	\caption{The diamond reduction: it is applicable to a graph $G$ containing 
		an induced subgraph shown on the left, in which vertices $a,b,c$ have degree $3$ 
		in $G$. If the degree of $x$ is also $3$, the reduction consists in deleting
		the edges $ab$ and $xc$. Otherwise, the reduction consists in deleting vertex $b$
		and introducing the edge $ac$.}
	\label{fig:diamond}
\end{figure}

\begin{claim}
	Let $G'$ be a graph obtained from $G$ by the diamond reduction. Then $G$ has a Hamiltonian cycle
	if and only if $G'$ has a Hamiltonian cycle. Moreover, if $G$ is $T$-subgraph-free, then so is $G'$.
\end{claim}

\begin{proof}
Assume first that the degree of $x$ is $3$ in $G$, in which case the reduction 
consists in deleting the edges $ab$ and $xc$. Clearly, if $G'$ has a Hamiltonian
cycle, then the same set of edges forms a Hamiltonian cycle in $G$. Conversely, 
assume $G$ has a Hamiltonian cycle $C$. It is not difficult to see that either 
the edges $ax$, $xb$, $bc$ belong to $C$ or the edges $ab$, $xb$, $xc$ belong 
to $C$. Since vertices $x$ and $b$ do not have neighbours different from $a$ and $c$,
any of the two combinations can be chosen to create a Hamiltonian cycle. Therefore, 
by deleting the edges $ab$ and $xc$ we transform $G$ into a graph with 
a Hamiltonian cycle. 

Now assume that the degree of $x$ is more than $3$, in which case the reduction 
is illustrated in Figure~\ref{fig:diamond}.
Suppose first that $G$ has a Hamiltonian cycle $C$. Since the degree of $b$ is $3$,
at least one of the edges $ab$ and $bc$ belong to $C$. 
If both of them belong to $C$, then by replacing 
these edges with $ac$ we obtain a Hamiltonian cycle in $G'$. 
Otherwise, without loss of generality, $ab\in C$ and $bc\not\in C$. 
Then $bx\in C$ and $cx\in C$ (since the degree of $b$ and $c$ is $3$),
and hence by replacing the edges $ab$ and $bx$ with the edge $ax$ we 
obtain a Hamiltonian cycle in $G'$.

Now let $G'$ have a Hamiltonian cycle $C$. If the edge $ac$ belongs to $C$, 
then by replacing this edge with the edges $ab$ and $bc$ we obtain 
a Hamiltonian cycle in $G$. Otherwise, both edges incident to 
$a$ and both edges incident to $c$ belong to $C$. Therefore, 
by replacing $ax$ with the edges $ab$ and $bx$
we obtain a Hamiltonian cycle in $G$.

Now we show that the transformation represented in Figure~\ref{fig:diamond}
preserves $T$-subgraph-freeness. Assume to the contrary that $G'$ contains 
a copy of $T$ as a subgraph. Then this copy contains the edge $ac$, since otherwise
$G$ contains $T$ as a subgraph. If $a$ (or $c$) is a leaf of $T$, then by replacing
it with $b$ we obtain a copy of $T$ in $G$. Therefore, without loss of generality,
$a$ has degree $3$ in $T$ and $c$ has degree $2$ in $T$. If $a$ represents 
the rightmost vertex of degree $3$ in Figure~\ref{fig:H12}, then by replacing in $T$ 
the three edges incident to $a$ with the three edges incident to $b$ we obtain 
a copy of $T$ in $G$. 

Now assume that $a$ represents the middle vertex of degree~$3$ in
Figure~\ref{fig:H12}. We denote by $u_0$ the neighbour of $a$ of degree~$3$ in $T$
and by $u_1,u_2$ the neighbours of $u_0$ of degree~$1$. Also, let $v_0$ be 
the neighbour of $c$ of degree $3$ in $T$ (different from $a$) and let $v_1,v_2$ be
the neighbours of $v_0$ of degree $1$. Since the degree of $x$ in $G$ is more than
$3$, we may consider a vertex $y\not\in\{a,b,c\}$ adjacent to $x$.
If $y=u_0$, then the edges $u_0u_1$, $u_0u_2$, $u_0x$, $xc$, $cb$, $cv_0$, $v_0v_1$,
$v_0v_2$ form a forbidden subgraph in $G$. If $y=u_1$, then the edges  
$au_0$, $ab$, $ax$,  $xu_1$, $xc$, $cv_0$, $v_0v_1$,
$v_0v_2$ form a forbidden subgraph in $G$. By symmetry, we also conclude that 
$y\ne u_2,v_0,v_1,v_2$. But then the edges $u_0u_1$, $u_0u_2$, $u_0a$, $ax$, $xy$,
$xc$, $cb$, $cv_0$ form a forbidden subgraph in $G$. 
This contradiction completes the proof.
\end{proof}

One more reduction is illustrated in Figure~\ref{fig:butterfly}. 
We will refer to this reduction as the {\it butterfly reduction} and 
will denote it by (R6). 

\begin{figure}[ht]
	\begin{center}
		\begin{picture}(100,50)
			\setlength{\unitlength}{0.5mm}
			\put(5,5){\circle*{4}}
			\put(25,5){\circle*{4}}
			\put(15,15){\circle*{4}}
			\put(5,25){\circle*{4}} 
			\put(25,25){\circle*{4}}
			
			\put(25,5){\line(-1,1){10}}
			\put(25,5){\line(0,1){20}}

			\put(5,5){\line(0,1){20}}
			\put(5,5){\line(1,1){10}}
			
			\put(15,15){\line(-1,1){10}}
			\put(15,15){\line(1,1){10}}
			\put(-3,22){$a$}
			\put(-3,3){$b$}
			\put(28,3){$c$}
			\put(28,22){$d$}
			\put(13,7){$x$}
			
			\put(40,15){\vector(1,0){20}} 
		\end{picture}
		\begin{picture}(50,50)
			\setlength{\unitlength}{0.5mm}
			\put(5,5){\circle*{4}}
			\put(25,5){\circle*{4}}
			\put(15,15){\circle*{4}}
			\put(5,25){\circle*{4}} 
			\put(25,25){\circle*{4}}
			
			\put(25,5){\line(-1,1){10}}
			\put(25,5){\line(0,1){20}}

			\put(5,5){\line(0,1){20}}
			
			\put(15,15){\line(1,1){10}}
			
			\put(-3,22){$a$}
			\put(-3,3){$b$}
			\put(28,3){$c$}
			\put(28,22){$d$}
			\put(13,7){$x$}
		\end{picture}
	\end{center}
	\caption{The butterfly reduction: it is applicable to a graph $G$ containing 
		an induced subgraph shown on the left, in which vertices $a,b,c$ have degree 
		$3$ in $G$, and moreover, $a$ and $b$ have white neighbours.}
	\label{fig:butterfly}
\end{figure}
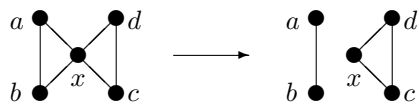

\begin{claim}[\app]\label{c-butterfly}
	Let $G'$ be a graph obtained from $G$ by the butterfly reduction. 
	Then $G$ has a Hamiltonian cycle if and only if $G'$ has a Hamiltonian cycle. 
\end{claim}

\noindent
In our algorithm to solve the problem we implement the above rules and reductions
whenever they are applicable. 
We now develop more reductions allowing us to bound the number of vertices 
in black components. We assume that none of the above rules and reductions is applicable to $G$.

\begin{claim}[\app]\label{claim:ld}
	Let $x$ be a vertex of degree at least $13$. If the neighbourhood of $x$ does not 
	contain two adjacent vertices of degree $3$, then $G$ has no Hamiltonian cycle. 
	Otherwise, $G$ has a Hamiltonian cycle if and only if $G-x$ has.
\end{claim}

\noindent
Application of Claim~\ref{claim:ld} to vertices of large degree either shows that $G$ has no Hamiltonian 
cycle or reduces the input graph to a graph of maximum degree $12$. We will refer to this reduction as 
the {\it large degree reduction} and will denote it by (R7).

\begin{claim}\label{claim:2}
	The black graph has no induced paths of length $8$.
\end{claim}

\begin{proof}
Assume by contradiction that $P=(u_0,u_1,\ldots,u_{8})$ is an induced path consisting of black vertices.
By definition, every vertex of $P$ has a neighbour outside of $P$. 
\begin{itemize}
	\item[($\star$)] If a vertex $v\notin P$ is adjacent to vertices $u_i$ and $u_j$ $(1\le i<j\le 7)$,
	then $j-i\le 2$. 
	
	To show this, assume without loss of generality, that $i=1$ and $j=4$ (for 
	other values of $i$ and $j$ the proof is similar). Then any neighbour $w\notin P$ of $u_5$ 
	must coincide with $v$, since otherwise the edges $u_0u_1$, $u_1u_2$,  $u_3u_4$, $u_4u_5$,
	$u_5u_6$, $u_1v$, $u_4v$, $u_5w$ form a forbidden subgraph. This contradiction shows that $v$ is the only neighbour of $u_5$ outside of $P$, and similarly, $v$ is the only neighbour of $u_6$ and $u_7$ outside of $P$. But then the diamond reduction is applicable to the subgraph induced by $v,u_5,u_6,u_7$.  
\end{itemize}
From ($\star$) we conclude that $u_1,u_4,u_7$ have pairwise different neighbours outside of $P$. We
denote these neighbours by $v_1,v_4,v_7$, respectively. 
\begin{itemize}
	\item[($\star$$\star$)] Vertex $u_2$ has exactly one neighbour outside of $P$, which is either $v_1$ or $v_4$.
	Moreover, if $u_2$ is adjacent to $v_1$ (respectively, $v_4$), then $v_1$ (respectively, $v_4$) 
	is the unique neighbour of $u_1$ (respectively, $u_4$) outside of $P$.
	
	To prove this, assume a neighbour $w\notin P$ of $u_2$ is different from $v_1$ and $v_4$. 
	Then the edges $u_0u_1$, $u_1u_2$, $u_2u_3$, $u_3u_4$, $u_4u_5$,
	$u_1v_1$, $u_4v_4$, $u_2w$ form a forbidden subgraph. This proves, in particular, 
	the second sentence of ($\star$$\star$).
	
	Also, if $u_2$ is adjacent both to $v_1$ and $v_4$, then a neighbour $w\notin P$ of $u_5$ must be different 
	from $v_1$ and $v_4$ (by ($\star$)), in which case the edges  $u_1u_2$, $u_3u_4$, $u_4u_5$, $u_5u_6$,
	$u_2v_1$, $u_2v_4$, $u_4v_4$, $u_5w$ form a forbidden subgraph. 
	
	By symmetry, similar conclusions apply to vertices $u_3$, $u_5$, $u_6$.
\end{itemize}

Assume $u_2$ is adjacent to $v_4$. Then, according  ($\star$),
$u_5$ and $u_6$ are not adjacent to $v_4$ and hence by ($\star$$\star$) 
$v_7$ is the only neighbour of $u_5$, $u_6$ and $u_7$
outside of $P$. Therefore, the diamond reduction is applicable to 
the subgraph induced by $v_7,u_5,u_6,u_7$. 
This contradiction shows that $v_1$ is the only neighbour of $u_2$ outside of $P$. 
Similarly, $v_7$ is the only neighbour of $u_6$ outside of $P$. 

Now the diamond reduction is applicable either to the subgraph induced by $v_1,u_1,u_2,u_3$ 
(if $u_3$ is adjacent to $v_1$) or to the subgraph induced by $v_7,u_5,u_6,u_7$ (if $u_5$ is adjacent
to $v_7$) or to the subgraph induced by $v_4,u_3,u_4,u_5$ (if $u_3$ is adjacent
to $v_4$ and $u_5$ is adjacent to $v_4$).        
\end{proof}

\noindent
Since graphs of diameter $D$ and maximum degree $\Delta$ have fewer than $\frac{\Delta}{\Delta-2}(\Delta -1)^D$
vertices, we conclude that after eliminating vertices of large degree, every black component has fewer than
$\frac{12}{10}11^7$ vertices.

To develop more rules and reductions, assume $G$ has a Hamiltonian cycle $C$.
We can further assume that not all vertices of the graph are black, since otherwise
the graph contains fewer than $\frac{12}{10}11^7$ vertices, in which case
we can solve the problem by brute-force. 
A sequence of consecutive vertices of $C$ surrounded by white vertices will be 
called a {\it black interval}. Observe that 
each black interval consists of at least $2$ vertices (according to (R4)).

Let $K$ be a black component of $G$. We will call the vertices of $K$ 
that have white neighbours the {\it contact vertices}. Note that $K$ 
may consists of one or more intervals. Each interval gives rise to exactly two contact vertices,
and hence the number of contact vertices in $K$ is even. 

Let us show that for $T$-subgraph-free graphs, then number of intervals is at most $2$.

\begin{claim}\label{l-11}
	Any black component consists of at most two intervals.   
\end{claim}

\begin{proof}
	 Assume that a black component $K$ has more than two intervals. 
	Then $K$ must contain an interval $I_1$ the vertices of which have neighbours 
	in two other intervals, say in $I_2$ and $I_3$. 
	
	We denote $I_1=(a_1,\ldots,a_{k_1})$, $I_2=(b_1,\ldots,b_{k_2})$,
	$I_3=(c_1,\ldots,c_{k_3})$ ($k_1,k_2,k_3\ge 2$)
	and let $a_ib_k$ and $a_jc_{\ell}$ be two chords. Without loss of generality we 
	assume that $i\le j$ and that $j-i$ is as small as possible. Also, we assume 
	that if we move along $C$ from $I_1$ in the anticlockwise direction, then
	first we meet $I_2$ and then $I_3$. Observe that $a_{i}$ and $b_k$ are separated 
	by at least two vertices of $C$, since otherwise the only vertex of $C$ 
	separating $a_i$ and $b_k$ must be white, in which case the edge $a_ib_k$ 
	cannot belong to any Hamiltonian cycle and hence can be removed from the graph.
	
	\medskip
	{\it Case} 1: $j=i+1$. If $b_k$ and $c_{\ell}$ are separated by 
	at least two vertices of $C$, then the two edges of $C$ incident to $b_k$,  
	the two edges of $C$ incident to $c_{\ell}$, the two edges of $C$ incident to $a_{i}$
	together with the chords $a_ib_k$ and $a_jc_{\ell}$ form a forbidden subgraph. 
	
	Now assume that $b_k$ and $c_{\ell}$ are separated by a single (white) vertex 
	$x$ of $C$, say $k=1$ and $\ell=1$ and $x$ is adjacent to both $b_1$ and $c_1$.  
	If a chord incident to $b_2$
	connects it to a vertex not in $\{a_{i},a_{i+1},c_1,c_2\}$, then
	the two edges of $C$ incident to $b_2$,  
	the two edges of $C$ incident to $x$ together with the edges $c_1c_2$, $a_ib_1$, 
	$a_{i+1}c_1$ and the chord incident to $b_2$ form a forbidden subgraph.
	If any chord incident to $b_2$ connects it to a vertex in $\{a_{i},a_{i+1},c_1,c_2\}$,
	and by symmetry any chord incident to $c_2$ connects it to a vertex in
	$\{a_{i},a_{i+1},b_1,b_2\}$, then we are in the conditions of 
	the first paragraph of Case 1.
	
	\medskip
	{\it Case} 2: $i=j$. Assume without loss of generality that $b_{k+1}$ is black.
	To avoid Case~1, $b_{k+1}$ has no neighbours in $I_3$ and 
	is not adjacent to the neighbours of $a_i$ on $C$. If additionally $b_{k+1}$ is not
	adjacent to $a_i$, then the two edges of $C$ incident to $b_{k+1}$, 
	the  two edges of $C$ incident to $c_{\ell}$ together with the edges $a_ib_k$,
	$a_ic_{\ell}$, an edge of $C$ incident to $a_i$ and a chord incident to $b_{k+1}$
	form a forbidden subgraph (if $b_{k+1}$ and $c_{\ell}$ have a common (white)
	neighbour in $C$, then replace the edge $a_ib_k$ with the second edge of $C$ 
	incident to $a_i$). This contradiction leads us to the conclusion that 
	$a_ib_{k+1}$ is the only chord incident to $b_{k+1}$, and by symmetry $a_ib_{k}$ 
	is the only chord incident to $b_{k}$, i.e.
	$b_k$ and $b_{k+1}$ have degree $3$ in $G$. If $I_2$ has more than $2$ vertices,
	say $b_{k+2}$ is black, then for the same reason $a_ib_{k+2}$ is the only chord
	incident to $b_{k+2}$, in which case the diamond reduction is applicable to 
	the subgraph induced by $a_i,b_k,b_{k+1},b_{k+2}$. Thus, we conclude that 
	$I_2$ consists of exactly two vertices both of which are adjacent to $a_i$. 
	By symmetry $I_3$ consists of exactly two vertices both of which are adjacent 
	to $a_i$. But now the butterfly reduction is applicable to $G$.
	
	\medskip
	{\it Case} 3: $i+1<j$. Since $j-i$ is minimal possible, 
	the vertices of $I_1$ between $a_i$ and $a_j$ have no neighbours in $I_2$ or $I_3$.
	Also, to avoid Case~2, $a_i$ has no neighbours in $I_3$, 
	while $a_j$ has no neighbours in $I_2$.
	As before, without loss of generality we assume that $b_{k+1}$ is black. 
	
	\medskip
	{\it Case} 3.1: there is a chord $b_{k+1}x$ with $x\ne a_i$. 
	We also consider a chord $a_{i+1}y$ and a chord $a_{i+2}z$. 
	Note that $y$ belongs to $I_1$ (to avoid Case~1).
	
	First, we observe that $z\in \{a_{i-1}, a_{i}\}$, since otherwise 
	the four edges of $C$ incident to $a_i$, $a_{i+1}$, $a_{i+2}$, 
	the two edges of $C$ incident to $b_k$ together with the chords $a_ib_k$ 
	and $a_{i+2}z$ form a forbidden subgraph. 
		
	Next, we note that $x\in \{b_{k-1}, a_{i-1}\}$. To show this, assume 
	$x\not\in \{b_{k-1}, a_{i-1}\}$. If $x\ne y$,
	then the three edges of $C$ incident to $b_k$ 
	and $b_{k+1}$, the two edges of $C$ incident to $a_{i+1}$ together with the chords
	$a_ib_k$, $b_{k+1}x$, $a_{i+1}y$ form a forbidden subgraph.
	If $x=y$, then $x\in I_1-\{a_i,a_{i+1},\ldots,a_j\}$ (according to the above arguments).
	Assume $x$ belongs to the path of $C$ connecting $a_i$ to $b_k$ and avoiding $I_3$
	(the other case is similar). Also let $u$ be the neighbour of $x$ on $C$
	between $x$ and $b_k$ (such a vertex different from $b_k, b_{k+1}$ must exist,
	because $x$ and $b_k$ belong to different intervals; moreover, $ub_{k+1}$ is not
	an edge of $C$, since otherwise $u$ is white, in which case
	the edge $b_{k+1}x$ cannot belong to any Hamiltonian cycle and hence can be removed 
	from graph). Then the two edges of $C$
	incident to $b_{k+1}$, the two edges of $C$ incident to $a_{i+2}$, together 
	with the edges $xu$, $a_{i+1}x$, $b_{k+1}x$ and $a_{i+2}z$ form a forbidden subgraph
	(remember that $z\in \{a_{i-1}, a_{i}\}$ and $x\ne a_{i-1}$).

	Finally, we conclude that $y=a_{i-1}$, since otherwise the two edges of $C$ 
	incident to $a_{i+1}$, the two edges of $C$ incident to $b_{k+1}$ together with 
	the chords $a_ib_k$, $a_{i+1}y$, $a_{b+1}x$ and the edge $a_{i-1}a_i$ 
	(if $x=b_{k-1}$) or the edge $b_{k-1}b_k$ (if $x=a_{i-1}$) form a forbidden subgraph.

	Assume $x=b_{k-1}$.  If $z=a_i$, then
	the two edges of $C$ incident to $a_{i-1}$, the two edges of $C$ incident 
	to $b_{k+1}$ (or to $b_{k-1}$ if $b_{k+1}$ is closer to $I_1$ than $b_{k-1}$)
	together with the chords $a_{i+1}a_{i-1}$, $a_{i+2}a_{i}$, $b_{k+1}b_{k-1}$, $a_ib_k$
	form a forbidden subgraph. If $z=a_{i-1}$, then  
	the two edges of $C$ incident to $a_{i-1}$, the two edges of $C$ incident to
	$a_{i+2}$, the two edges of $C$ incident to $b_k$ together with the chords $a_ib_k$
	and $a_{i+2}a_{i-1}$ form a forbidden subgraph (observe that since $a_{i-1}$,  
	$b_{k-1}$, $b_{k+1}$ are black, there is at least one white vertex separating 
	$a_{i-1}$ from $b_{k-1}$, $b_{k+1}$). 
	
	Suppose at last that $x=a_{i-1}$. 
	Let $u$ be the neighbour of $a_{i-1}$ on $C$ different from $a_i$, 
	and let $v$ be the neighbour of $b_{k+1}$ on $C$ different from $b_k$.
	Note that $u\ne v$, since otherwise $u=v$ is white, in which case 
	the edge $b_{k+1}a_{i-1}$ does not belong to any Hamiltonian cycle and 
	hence can be removed from the graph. If $z=a_i$, then the two edges of $C$ 
	incident to $a_{i+2}$, the two edges of $C$ incident to $b_{k+1}$
	together with the chords $a_{i+1}a_{i-1}$, $a_{i+2}a_{i}$, $a_{i-1}b_{k+1}$ and 
	the edge $a_{i-1}u$ form a forbidden subgraph. We are left with the case
	when $a_{i+2}a_{i-1}$ is the only chord incident to $a_{i+2}$. In this case we
	can consider vertex $a_{i+3}$ and a chord $f$ incident to $a_{i+3}$. 
	If $f\ne a_{i+3}a_{i-1}$, then the three edges of $C$ incident to $a_{i+2}$ and
	$a_{i+3}$, the two edges of $C$ incident to $b_{k+1}$ together with the chords 
	$a_{i-1}b_{k+1}$, $a_{i-1}a_{i+2}$ and $f$ form a forbidden subgraph. If 
	$a_{i+3}a_{i-1}$ is the only chord incident to $a_{i+3}$, then the diamond reduction
	is applicable to the subgraph of $G$ induced by $a_{i-1},a_{i+1},a_{i+2},a_{i+3}$.

	\medskip
	{\it Case} 3.2: the chord $b_{k+1}a_i$ is the only chord incident to $b_{k+1}$.
	By symmetry, $b_{k}a_i$ is the only chord incident to $b_{k}$. To avoid 
	the diamond reduction and Case~3.1, we conclude that $I_2$ consists of exactly
	two vertices both of which are adjacent to $a_i$. 
	
	If $a_{i+1}$ is incident to a chord $a_{i+1}y$ with $y\ne a_{i-1}$, 
	then the two edges of $C$ incident to $y$, the two edges of $C$ incident to $b_{k+1}$
	together with the edges $a_ia_{i+1}$, $a_ib_{k+1}$, $a_{i+1}y$ and one of the edges
	$a_{i+1}a_{i+2}$ or $a_{i-1}a_{i}$ (depending on whether $y$ appears in $I_1$ 
	before or after $a_i$) form a forbidden subgraph. 
	
	If $a_{i+1}a_{i-1}$ is the only chord incident to $a_{i+1}$,
	then the butterfly reduction is applicable to the subgraph of $G$ induced 
	by the two vertices of $I_2$ and the three vertices $a_{i-1},a_i,a_{i+1}$ of $I_1$.
\end{proof}

\noindent
By Claim~\ref{l-11}, if $G$ has a Hamiltonian cycle, then every black component 
has two or four contact vertices.

\begin{itemize}
	\item[(R8)]  If a black component $K$ has exactly two contact vertices, check
	if $K$ has a Hamiltonian path connecting the contact vertices. If such a path does not exist, 
	then stop: the input graph has no Hamiltonian cycle. Otherwise, choose arbitrarily a Hamiltonian path connecting the contact vertices, include the edges of the path in the solution and delete all other edges from $K$. 
\end{itemize}

\noindent
Rule (R8) destroys black components with two contact vertices, i.e. after 
its implementation all vertices in such components become white. 

Now we discuss the case where each black component has exactly four contact vertices.
Let $K$ be such a component with contact vertices $v_1,v_2,v_3,v_4$. If $G$ has a Hamiltonian cycle, 
then the vertices of $K$ can be partitioned into two parts each of which forms a path connecting 
a pair of contact vertices. We will call such a partition a {\it pairing} (of contact vertices) 
and will refer to a pairing as the set of edges in the two paths. Also, we will say that two 
pairings are of the same type, if they pair the contact vertices in the same way. Clearly, 
if all possible pairings in $K$ have the same type, then it is irrelevant which one to choose,
since non-contact vertices of $K$ have no neighbours outside of $K$.

The above discussion justifies the following two rules.

\begin{itemize}
	\item[(R9)] If a black component $K$ with four contact vertices does not admit any pairing, 
	then stop: the input graph has no Hamiltonian cycle.
	\item[(R10)] If in a black component $K$ with four contact vertices all possible pairings have the same 
	type, then choose arbitrarily any such pairing and delete all other edges from $K$.
	If this procedure disconnects the graph, then stop: the input graph has no Hamiltonian cycle.
\end{itemize}

Finally, we analyse the situation when each black component of $G$ admits pairings of at least two
different types. 

\begin{claim}\label{lem:key}
	If 
	each black component of the (connected) graph $G$ admits pairings of at least two
	different types, then $G$ has a Hamiltonian cycle.
\end{claim}

\begin{proof}
	Let $K$ be a black component with contact vertices $v_1,v_2,v_3,v_4$ and 
	let $B$ and $R$ be 
	two pairings of different types, say $B$ pairs $v_1$ with $v_2$ and $v_3$ with $v_4$,
	while $R$ pairs $v_1$ with $v_3$ and $v_2$ with $v_4$. Assume that 
	\begin{itemize}
		\item[(1)] the deletion of all edges of $K$ except for the edges of $B$ disconnects the graph into two 
		components $C_{12}$ (containing vertices $v_1$ and $v_2$) and $C_{34}$ (containing vertices $v_3$ and $v_4$),
		and
		\item[(2)] the deletion of all edges of $K$ except for the edges of $R$ disconnects the graph into two 
		components $C_{13}$ (containing vertices $v_1$ and $v_3$) and $C_{24}$ (containing vertices $v_2$ and $v_4$).
	\end{itemize}
	Note that (1) separates $v_1$ from $v_3$ and $v_4$, while (2) separates $v_1$ from $v_4$. Therefore, 
	after the deletion of {\it all} edges of $K$ vertex $v_1$ is separated from all other contact vertices.
	In other words, after the deletion of all edges of $K$, vertices $v_1,v_2,v_3,v_4$ belong to pairwise 
	different connected components, say $V_1,V_2,V_3,V_4$, respectively. 
	
	We observe that in each connected component $V_i$ vertex $v_i$ has degree $1$ 
	(it is adjacent to a white vertex only).
	Any other vertex of odd degree in $V_i$ (if there is any) is black, i.e. appears in some black component $K'$.
	In the graph $G[K']$ the number of odd vertices is even (by the Handshake lemma). 
	Attaching to $G[K']$ four white neighbours changes the parity of exactly four vertices of $K'$ and 
	hence leaves the number of vertices of $K'$ with odd degrees in the graph $G$ even. 
	Since all vertices of $K'$ belong to only one of the components $V_i$, we conclude 
	that in each component $V_i$ the number of vertices of odd degree is odd. This is not possible by 
	the Handshake lemma and hence either (1) or (2) is not valid, 
	i.e. we can keep one of the pairings  and delete all other edges of $K$ without disconnecting $G$.
	This operation destroys $K$, i.e. makes all vertices of $K$ white. 
	
	Applying the above arguments to all black components, one by one, we transform $G$ into a connected graph
	in which all vertices are white, i.e. to a Hamiltonian cycle. 
\end{proof}

We summarize the discussion in the following algorithm to solve the problem.

\begin{itemize}
	\item[1.] Apply rules and reductions (R1) -- (R7)
	as long as they are applicable. 
	
	\item[2.] If the algorithm did not stop at Step 1 and the graph has fewer than $\frac{12}{10}11^7$ vertices, then solve the problem by brute-force. Otherwise, 
	check the number of contact vertices in black components. If there is a black component with the number of contact vertices different from $2$ or $4$, then stop: 
	$G$ has no Hamiltonian cycle.
	
	\item[3.] If the algorithm did not stop at Step 2,  then apply (R8) to 
	black components with two contact vertices, and (R9) and (R10) to black components with four contact vertices.
	\item[4.] If the algorithm did not stop at Step 3, then find a Hamiltonian cycle 
	according to Claim~\ref{lem:key}.
\end{itemize}

Reductions (R8), (R9), (R10) can be implemented in constant time, because the number of 
vertices in each black component is bounded by a constant. It is also obvious that all other rules, and hence all steps of the algorithm can be implemented in polynomial time. The correctness of the algorithm follows from the proofs of the claims.
\end{proof}

\section{$\mathbf{k}$-Induced Disjoint Paths}\label{s-k}

Our first result essentially follows from the observation that solutions of {\sc $k$-Induced Disjoint Paths} with long paths are solutions of {\sc $k$-Disjoint Paths}, which we recall is polynomial-time solvable~\cite{RS95}. 

\begin{theorem}[\app]\label{t-kip}
For all~$k\geq 2$, {\sc $k$-Induced Disjoint Paths} is polynomial-time solvable for $\H_1$-subgraph-free graphs.
\end{theorem}

\noindent
We sketch our proof for the case where $H=\H_2$.

\begin{theorem}[\app]\label{t-kip2}
For all~$k\geq 2$, {\sc $k$-Induced Disjoint Paths} is polynomial-time solvable for $\H_2$-subgraph-free graphs.
\end{theorem}

\noindent{\bf \sffamily\bfseries\textcolor{lipicsGray}{Proof sketch.}}
First, branch on all $O(2^k n^{3k})$ options (so a polynomial number, as $k$ is fixed) of solution paths that have at most three internal vertices. For each branch, we remove the guessed solution paths and the neighbours of the vertices on these paths. Let $k$ still be the number of terminal pairs. We now only look for solution paths with at least four vertices. Branch on all $O(n^{4k})$ options of choosing the first two vertices $a_z,b_z$ on the solution path from every terminal $z\in \{s_i,t_i\}$ for $i\in \{1,\ldots,k\}$. In each branch, we remove all other neighbours of $z,a_z$ from the graph, so every terminal $z$ now has degree~$1$, while $a_z$ has degree~$2$. We discard the branch if $(\dagger)$ $\{a_z,b_z\} \cap \{a_{z'},b_{z'}\} \not= \emptyset$ for some terminals $z,z'$ or one of $a_z,b_z$ is the same or neighbours one of $a_{z'},b_{z'}$ for some terminals $z,z'$ not from the same terminal pair.

We now start a recursive procedure. We first preprocess the input. If $b_{s_i}$ and $b_{t_i}$ are adjacent for some $i \in \{1,\ldots,k\}$, then we remove the solution path $s_i,a_{s_i}, b_{s_i}, b_{t_i},a_{t_i},t_i$ and their neighbours from the graph. If $b_{z}$ and $b_z'$ are adjacent for some terminals $z,z'$ that do not form a terminal pair, we discard the branch. 

We run the polynomial-time algorithm for {\sc $k$-Disjoint Paths} from~\cite{RS95} on the remaining terminal pairs. If this results in a no-answer, then we discard the branch. Else, we found a solution $P_1,\ldots,P_k$. We may assume that each path $P_i$ is induced, or we may shortcut it. If $P_1,\ldots,P_k$ is also a solution of {\sc $k$-Induced Disjoint Paths}, then we return a yes-answer. Otherwise, there is (say) an edge $(x_1,x_2)$ between paths $x_1 \in P_1$ and $x_2 \in P_2$. We pick $x_1$ such that it is closest to $t_1$ on $P_1$ and under that condition we pick $x_2$ such that it is closest to $t_2$ on $P_2$.

Let $z_1,z_3$ be the two neighbours of $x_1$ on $P_1$ and $z_2,z_4$ the two neighbours of $x_2$ on $P_2$.
We let $S=\{z_1,x_1,z_3,z_2,x_2,z_4\}$. Observe that $S$ does not contain any terminal by $\dagger$ and the preprocessing.
Suppose any $z \in \{z_1,z_2,z_3,z_4\}$ has two neighbours outside of $S$. Then $G$ has a $\H_2$ as a subgraph. Thus we may assume $(\ddagger)$ that each $z \in \{z_1,z_2,z_3,z_4\}$ has at most one neighbour outside of $S$.

By the choice of $(x_1,x_2)$ and the fact that $P_1$ is induced, $z_3$ has no neighbours in $S$ except $x_1$. Suppose the edge $(z_1,z_2)$ exists and one of $\{x_1,x_2\}$ has a neighbour outside of $S$. Then there is an $\H_2$ with middle path $x_1,z_1,z_2$ since $s_2 \not\in S$. Suppose the edge $(z_1,z_4)$ exists and one of $\{x_1,x_2\}$ has a neighbour outside of $S$. Then there is an $\H_2$ with middle path $x_1,z_1,z_4$ since $s_2 \not\in S$. 

Now either the edges $(z_1,z_2)$ and $(z_1,z_4)$ do not exist (see Figure~\ref{fig:rule-1-bis}) or at least one of them exists and $x_1,x_2$ have no neighbours outside $S$ (see Figure~\ref{fig:rule-2-bis}). In the former case, we will apply Rule 1, while in the latter case, we will apply Rule 2 (see Figure~\ref{fig:rule-1-bis} and~\ref{fig:rule-2-bis} for the description of these two rules). Using $\ddagger$, we can prove that Rules~1 and~2 are safe and that they both preserve $\dagger$ and $H_2$-subgraph-freeness (\app). We can recognize both rules and apply them in polynomial time.
This makes the instance smaller by one vertex and we recurse. Hence, our algorithm will terminate in polynomial time.
\pushQED{\qed}\popQED

\begin{figure}
	\centering
	\includegraphics[width=0.8\textwidth]{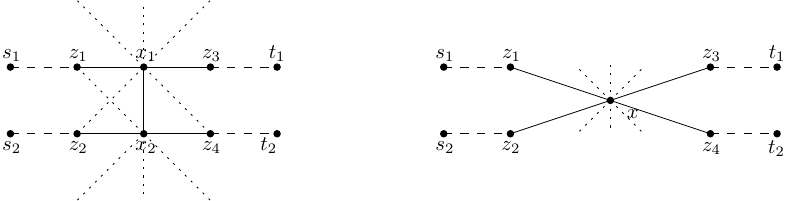}
	\vspace*{-0.2cm}
	\caption{Rule 1. Possible connections in our subgraph (left). What we replace this subgraph with (right). Dotted lines are possible additional edges.}
	\label{fig:rule-1-bis}
\end{figure}

\begin{figure}
	\centering
	\includegraphics[width=0.65\textwidth]{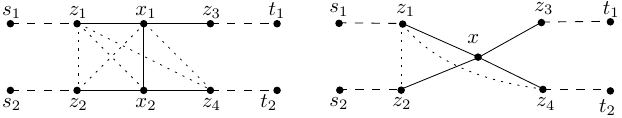}
	\caption{Rule 2. Possible connections in our subgraph (left). What we replace this subgraph with (right).}
	\label{fig:rule-2-bis}
\end{figure}

\medskip
\noindent
For our next result we follow the proof from Section 2.4 in \cite{LLMT09}. We cannot just take the $p$-subdivision~of that construction for some fixed $p$. However, as we will show, {\it some} of the edges may be liberally subdivided.

\begin{theorem}[\app]\label{t-khard}
For all $k\geq 2$, {\sc $k$-Induced Disjoint Paths} is \NP-complete for subcubic $(\H_4,\ldots,\H_\ell)$-subgraph-free graphs for all $\ell\geq 4$.
\end{theorem}

\section{$\mathbf{C_5}$-Colouring}\label{s-c5}

In this section, we give our polynomial-time certifying algorithm for {\sc $C_5$-Colouring} on $\H_3$-subgraph-free graphs.
For this purpose we introduce the following graphs. The {\em $C_5$-flower} $F_n$ is the graph (see Figure~\ref{fig:h3counters}) that we get from $C_{3n}$ (for $n \geq 3$) by adding a new central vertex with an edge to every third vertex of $C_{3n}$. If $n$ is odd, we call $F_n$ an {\em odd $C_5$-flower}, and if it is even we call $F_n$ an  {\em even $C_5$-flower}. We refer to the graphs $E_1, E_2$ and $E_3$ shown in Figure~\ref{fig:h3counters} as {\em exceptional graphs}. 
The following lemma (whose proof is a simple exercise) shows that all these graphs are $C_5$-critical (i.e. they are not $C_5$-colourable but every proper subgraph of them is). 

\begin{lemma}\label{l-5}
     The graph $K_3$, the odd flowers $F_n$ for odd $n \geq 3$ and the exceptional graphs $E_1, E_2$ and $E_3$ 
     are all $\H_3$-subgraph-free and $C_5$-critical. 
\end{lemma}

    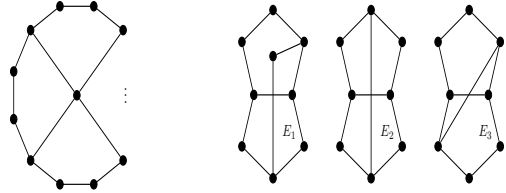
\begin{figure}
    \begin{center}
    \resizebox{6.5cm}{2.5cm}{\begin{tikzpicture}[every node/.style={vert}]
    
    \begin{scope}[xshift = -6cm, yshift = 1.8cm, rotate=0]
     \foreach \i[evaluate=\i as \d using {15+30*\i}] in {1,...,10}
             {\draw (\d:2) node (v\i){};} 
        \draw (0,0) node (m){};
        \draw (v1.center) -- (v2.center) -- (v3.center) -- (v4.center) -- (v5.center) -- (v6.center) -- (v7.center) -- (v8.center) -- (v9.center) -- (v10.center);
       \foreach \i in {1,4,7,10}{\draw (m) edge (v\i);}
       \draw (4:1.5) node[empty] (){{\Large $\vdots$}};
       \end{scope}
    
      \begin{scope}[xshift = 0cm, yshift = 1cm, rotate=90]
      \begin{scope}
       \Verts[:]{v1/36/1, v2/108/1, v3/180/1, v4/-108/1, v5/-36/1}
       \Edges{v1/v2,v2/v3,v3/v4,v4/v5,v5/v1} 
       \draw (0,-.5) node[empty]{$E_1$};
      \end{scope}
      \begin{scope}[xshift = 1.65cm]
        \Verts[:]{u3/0/1,u2/72/1,u4/-72/1, m/-90/0} 
       \end{scope} 
       \Edges{v1/u2, u2/u3, u3/u4, u4/v5, v3/m, m/u4}        
       \end{scope}
       
      \begin{scope}[xshift = 3cm,yshift = 1cm,rotate=90]
      \begin{scope}
       \Verts[:]{v1/36/1, v2/108/1, v3/180/1, v4/-108/1, v5/-36/1}
       \Edges{v1/v2,v2/v3,v3/v4,v4/v5,v5/v1} 
       \draw (0,-.5) node[empty]{$E_2$};
      \end{scope}
      \begin{scope}[xshift = 1.65cm]
        \Verts[:]{u3/0/1,u2/72/1,u4/-72/1} 
       \end{scope} 
       \Edges{v1/u2, u2/u3, u3/u4, u4/v5, v3/u3}
       \end{scope}
        
      \begin{scope}[xshift = 6cm,yshift = 1cm,rotate=90]
      \begin{scope}
       \Verts[:]{v1/36/1, v2/108/1, v3/180/1, v4/-108/1, v5/-36/1}
       \Edges{v1/v2,v2/v3,v3/v4,v4/v5,v5/v1} 
       \draw (0,-.5) node[empty]{$E_3$};
      \end{scope}
      \begin{scope}[xshift = 1.65cm]
        \Verts[:]{u3/0/1,u2/72/1,u4/-72/1} 
       \end{scope} 
       \Edges{v1/u2, u2/u3, u3/u4, u4/v5, v2/u4}
       \end{scope}
    \end{tikzpicture}}
    \caption{A (partial) $C_5$-flower and the three exceptional $\H_3$-subgraph-free $C_5$-critical graphs $E_1$, $E_2$ and $E_3$.}\label{fig:h3counters} 
    \end{center}
    \end{figure}

\noindent
We can now show the following structural result. 

  \begin{theorem}[\app]\label{thm:C5-char}
    The only $\H_3$-subgraph-free $C_5$-critical graphs are $K_3$, 
    odd flowers $F_n$ for $n \geq 3$, and the exceptional graphs $E_1, E_2$ and $E_3$. Equivalently,  the following three statements are true:
  \begin{enumerate}
    \item All $\H_3$-subgraph-free graphs of girth at least $6$ are $C_5$-colourable. 
    \item The only $\H_3$-subgraph-free $C_5$-critical graphs 
    of girth $5$ are $E_1$, $E_2$ and odd $C_5$-flowers $F_n$.
    \item The only $\H_3$-subgraph-free $C_5$-critical graph of
    girth $4$ is $E_3$. 
  \end{enumerate}
  \end{theorem}
 
 \noindent
 We can now prove the following result.
 
 \begin{theorem}
 There exists a polynomial-time certifying algorithm for {\sc $C_5$-Colouring} on $\H_3$-subgraph-free graphs.
 \end{theorem}
   
\begin{proof}
As every graph that does not map to $C_5$ must contain a $C_5$-critical subgraph, we get from Theorem~\ref{thm:C5-char} that $\H_3$-subgraph-free $C_5$-colouring is solved by detecting the non-existence of the graphs $K_3, E_1, E_2, E_3$ and $F_n$ (odd $n\geq 3$) in a $\H_3$-subgraph-free input graph $G$. We can detect the presence of the fixed graphs $K_3$, $E_1, E_2$ and $E_3$ in $G$ in polynomial time by brute force. What we must show is that we can detect the presence of any odd $C_5$-flower $F_n$ in polynomial time.
  To do so,  we simply observe that for a fixed centre vertex, $v_0$ we can make an auxiliary graph on its neighbours putting an edge between two if there is a path on three edges between them in $G$. Now, $G$ contains an odd $C_5$-flower with centre $v_0$ if and only if this auxiliary graph contains an odd cycle. We can check this in polynomial time for each $v_0$, so can find an odd $C_5$-flower in $G$ polynomial time.  
\end{proof}

\section{Star ${\mathbf{3}}$-Colouring}\label{s-star}

In this section, we give our linear-time certifying algorithms for {\sc Star $3$-Colouring} for bipartite $(\H_2,\H_4,\H_6\ldots)$-subgraph-free~graphs and general $(\H_2,\H_4,\H_6\ldots)$-subgraph-free~graphs, starting with the bipartite case.
Let $\mathbb{A}$ and $\overline{2P_4+3P_6}$ be the graphs shown in Figure~\ref{f-two}. We can show the following result.

\begin{figure}[t]
	\centering
	\includegraphics[width=0.65\textwidth]{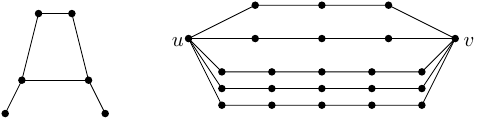}
	\caption{The graph $\mathbb{A}$ (left) and the graph $\overline{2P_4 + 3P_6})$ (right).}\label{f-two}
\end{figure}

\begin{theorem}[\app]\label{t-bipartite}
A bipartite $(\H_2,\H_4,\H_6\ldots)$-subgraph-free~graph is star $3$-colourable if and only if it is $(\mathbb{A},\overline{2P_4 + 3P_6})$-subgraph-free.
\end{theorem} 

\noindent
Theorem~\ref{t-bipartite} leads to a linear-time certifying algorithm, as $(\H_2,\H_4,\H_6\ldots)$-subgraph-free~graphs have bounded treewidth (see Appendix~\ref{a-h}) and we can then use Courcelle's Theorem~\cite{Co90} to check for a subgraph isomorphic to $\mathbb{A}$ or $\overline{2P_4 + 3P_6}$. For the general case, we apply the same arguments but we do not have an explicit list of the forbidden graphs (which we prove has finite size).

\begin{theorem}[\app]\label{t-general}
A $(\H_2,\H_4,\H_6\ldots)$-subgraph-free~graph is star $3$-colourable if and only if it is ${\cal F}$-subgraph-free for some \emph{finite} set of graphs ${\cal F}$.
\end{theorem} 

\section{Conclusions}\label{s-con}

We took four classic problems, {\sc Hamilton Cycle}, {\sc $k$-Induced Disjoint Paths}, {\sc $C_5$-Colouring} and {\sc Star $3$-Colouring}, that are easy on bounded treewidth, but for which we showed that  their hardness on subcubic graphs is not preserved under edge subdivision. For these C12-problems we gave both polynomial and \NP-completeness results for ${\cal H}$-subgraph-free graphs when ${\cal H}$ is some subset of $\{\H_1,\H_2,\ldots\}$.
Of course, we would like to have a classification for all our problems, among all $\mathcal{H}$-subgraph-free classes (even where $|\mathcal{H}|=1$), but it makes sense to understand the $\H_i$ first, and below we pose relevant open problems.

First, is there a graph $\H_\ell$ such that {\sc Hamilton Cycle} is \NP-complete for $\H_\ell$-free graphs? Second, what is the complexity of {\sc $k$-Induced Disjoint Paths} on $\H_3$-subgraph-free graphs?
The answer would give a dichotomy for {\sc $k$-Induced Disjoint Paths} on $\H_i$-subgraph-free graphs. Third, what is the complexity of {\sc $C_5$-Colouring} on $\H_i$-subgraph-free graphs, when $i=0\bmod 3$?
If these are in $\cP$, then we would have a dichotomy for  {\sc $C_5$-Colouring} on $\H_i$-subgraph-free graphs based on $i \bmod 3$. 
Fourth, what is the complexity of {\sc Star $3$-Colouring} on $\H_{2i}$-subgraph-free graphs for $i\geq 1$? 
If these are in $\cP$, then we would have a dichotomy for  {\sc Star $3$-Colouring} on $\H_i$-subgraph-free graphs based on $i \bmod 2$. 
Moreover, even though {\sc Star $k$-Colouring} is not C2 for $k\geq 10$ (see Appendix~\ref{a-star}), this is not known for $k=4$ (in~\cite{SA22}, Shalu and Antony ask about the complexity of {\sc Star $4$-Colouring} on subcubic graphs). 

We also ask what is the complexity of {\sc $k$-Induced Disjoint Paths} and {\sc $C_5$-Colouring} for subcubic graphs of girth $g$ for $g\geq 3$, i.e., are both C12-problems even C12' just like the other two? 
We also do not know the complexity of {\sc $k$-Induced Disjoint Paths}, for $k\geq 2$, on graphs of girth at least~$g$ without a degree bound, whereas the best degree bound for {\sc $C_5$-Colouring} is $6\cdot 5^{13}$ (see Appendix~\ref{a-degree}).

Finally, we note that {\sc Acyclic $3$-Colouring}, we have the same results as for {\sc Star $3$-Colouring} in Theorem~\ref{thm:main-star}; see Appendix~\ref{a-acyclic}. However, in contrast to {\sc Star $3$-Colouring}, we do not know if {\sc Acyclic $3$-Colouring} satisfies C2 (see also~\cite{SA24}), and we also have no certifying algorithms for this problem.



\newpage
\appendix

\section{${\mathbf C_{\mathbf 5}}$-Colouring for Bounded Degree and Large Girth}\label{a-degree}

The {\sc $k$-Colouring} problem is to decide if a graph $G$ has a {\it $k$-colouring}, which is a mapping $c:V(G)\to \{1,\ldots,k\}$ such that $c(u)\neq c(v)$ for any two adjacent vertices $u$ and $v$ of $G$.
We need a result of Emden-Weinert, Hougardy and Kreuter:

 \begin{theorem}[\cite{EHK98}]\label{t-g}
For all $k\geq 3$ and all $g\geq 3$, {\sc $k$-Colouring}  is \NP-complete for graphs with girth at least~$g$ and with maximum degree at most~$6k^{13}$. 
\end{theorem}

\noindent
We now repeat the proof of Chudnovsky et al.~\cite{CHRSZ19}, which comes down to replacing each edge of an input graph $G$ of {\sc $5$-Colouring}, which we may assume has girth at least $g$ and maximum degree at most $6\cdot 5^{13}$ due to Theorem~\ref{t-g}, by a path of length~$3$. This yields a new graph $G'$ of girth at least $g$, such that $G$ and~$G'$ have the same maximum degree. Hence, we derive the following result.

 \begin{proposition}
For every $g\geq 3$, {\sc $C_5$-Colouring}  is \NP-complete for graphs with girth at least~$g$ and with maximum degree at most~$6\cdot 5^{13}$. 
\end{proposition}

\section{The Four Testbed Problems Do Not Satisfy C3}\label{a-c3}

In this section we show that none of our four problems satisfy C3. We use the following notation in this section: for a graph $G$ and an integer $p\geq 1$, let $G_p$ be the $p$-subdivision of $G$ (which we recall is the graph obtained from $G$ after subdividing each edge of $G$ exactly $p$ times).

\begin{proposition}
{\sc Hamilton Cycle} does not satisfy C3.
\end{proposition}

\begin{proof}
We observe that for every graphs $G$ and every $p\geq 1$, $G_p$ is a no-instance of {\sc Hamilton Cycle} unless $G$ was a cycle.
\end{proof}

\begin{proposition}
{\sc $k$-Induced Disjoint Paths} does not satisfy C3.
\end{proposition}

\begin{proof}
Under any kind of subdivision, {\sc $k$-Induced Disjoint Paths} reduces to {\sc $k$-Disjoint Paths} over the same graph, which is in $\cP$ for all $k\geq 2$, as shown in~\cite{Sh80} for $k=2$ and in~\cite{RS95} for every $k\geq 3$.
\end{proof}

\begin{proposition}
{\sc $C_5$-Colouring} does not satisfy C3.
\end{proposition}

\begin{proof}
We first prove that for all $p\geq 4$, and for all $x,y \in V(C_5)$, there is a walk of length $p$ in $C_5$ from $x$ to $y$. 
First let $p=4$. 
To walk a distance of zero: walk two forward then two back. To walk at distance one (without loss of generality) forward: walk four backward. To walk at distance two (without loss of generality) forward: walk one back, one forward, and two forward. 
Now let $p=5$.
To walk a distance of zero: walk five forward. To walk at distance one (without loss of generality) forward: walk two forward, two back and one forward. To walk at distance two (without loss of generality) forward: walk one back, one forward, and three back. 
Finally, let $p\geq 6$. Keep moving one forward then one back until one of the two previous cases applies.

Now let $G$ be a graph. We give each vertex in $G$ a label from $\{1,\ldots,5\}$. From the above it follows that for every $p\geq 3$, we can extend $c$ to a homomorphism from $G_p$ to $C_5$; in other words, $G_p$ is a yes-instance of {\sc $C_5$-Colouring}.
\end{proof}

\begin{proposition}\label{p-sc3}
{\sc Star $3$-Colouring} does not satisfy C3.
\end{proposition}

\begin{proof}
Let $G$ be a graph. We show that for all $p\geq 3$, $G_p$ is a yes-instance of {\sc Star $3$-Colouring}. We do this by giving each vertex in $G$ a label from $\{1,2,3\}$. The resulting labelling $c$ might not be a $3$-colouring, but this is not important: we will show that we can extend $c$ to a star $3$-colouring of $G_p$ as follows.

Consider an edge $e$ in $G$ and let $P$ be the corresponding path (of $p+1$ edges) in $G_p$. 
It suffices to give two star $3$-colourings of this path, so that the first three vertices are distinct colours and the last three vertices are distinct colours: one in which the first and last vertices are the same colour and one in which they are a different colour. Let us identify a $3$-colouring of $P$ by a sequence of length $p+1$ over $\{1,2,3\}$. If $p+1$ is a multiple of three, then use $(123)^{\frac{p+1}{3}}$ for the different colour and $(123)^{\frac{p+1}{3}-1}231$ for the same colour. If $p+1$ is $1 \bmod 3$, then use $(123)^{\frac{p}{3}-1}2132$ for the different colour and $(123)^{\frac{p}{3}}1$ for the same colour.  If $p+1$ is $2 \bmod 3$, then use $(123)^{\frac{p-1}{3}}12$ for the different colour and $(123)^{\frac{p-1}{3}}21$  for the same colour.
\end{proof}

\section{Star ${\mathbf k}$-Colouring on Subcubic Graphs for Large ${\mathbf k}$}\label{a-star}

We cannot generalise our result for {\sc Star $3$-Colouring} to {\sc Star $k$-Colouring} for any $k\geq 3$, as for large $k$ the problem no longer satisfies C2. In fact, we prove even a stronger statement. A $k$-colouring of a graph $G$ is said to be {\it injective} if for every vertex $u\in V(G)$, every neighbour of $u$ is assigned a different colour, or in other words, the union of any two colour-classes induce a disjoint union of isolated vertices and edges. So, any injective $k$-colouring is a star $k$-colouring (but the reverse does not necessarily hold, for instance the $P_3$ is star $2$-colourable but has no injective $2$-colouring).

\begin{proposition}
For $k\geq 10$, all subcubic graphs have an injective $10$-colouring. 
\end{proposition}

\begin{proof}
It suffices to prove the statement for $k=10$. We do this by induction. For the base case, a graph with one vertex is star $10$-colourable. Now take a vertex~$v$ in a graph $G$ and assume $G \setminus \{v\}$ has an injective $10$-colouring. As $G$ is subcubic, $v$ has at most three neighbours, each of which have at most two more neighbours each. Thus there are at most nine vertices whose colour we wish to avoid. As we have ten colours in total, this means that we can safely colour $v$.
\end{proof}

\section{Bounded Treewidth Results}\label{a-h}

A graph $G$ contains $H$ as a {\it minor} if $G$ can be modified to $H$ by a sequence of vertex deletions, edge deletions and edge contractions; if not, then $G$ is {\it $H$-minor-free}. 
We need the following classic result.

\begin{theorem}[\cite{BRST91}]\label{t-bi}
For every forest $F$, all $F$-minor-free graphs have pathwidth, and thus treewidth, at most $|V(F)|-2$.
\end{theorem}

\noindent
We now make the following observation.

\begin{proposition}\label{p-h}
For every $\ell\geq 1$, the class of $(\H_{\ell},\H_{\ell+1},\ldots )$-subgraph-free graphs has bounded treewidth.
\end{proposition}

\begin{proof}
Let $\ell\geq 1$.
As every $(\H_{\ell},\H_{\ell+1},\ldots )$-subgraph-free graph is $H_\ell$-minor-free, we apply Theorem~\ref{t-bi}.
\end{proof}

\noindent
We also prove the following result.

\begin{proposition}\label{p-h2}
For every $n\geq 1$, the class of $(\H_n,\H_{2n},\H_{3n},\ldots)$-subgraph-free graphs has bounded treewidth.
\end{proposition}

\begin{proof}
We proceed by contraposition. If a class of graphs has unbounded treewidth, then every grid appears as a minor in some graph~\cite{RS86}. Let us explain the argument for $n=2$ first. We consider that the $3 \times 3$-grid appears as a minor in some graph $G$ in our class and let $f$ be the minor map from $G$ to the $3 \times 3$-grid. Consider the three vertices in the $3 \times 3$-grid that form the central row as $u,v,w$ (in succession). Choose $u' \in f^{-1}(u), v' \in f^{-1}(v), w' \in f^{-1}(w)$ so that $u',v',w'$ have degree greater than $2$, noting that such vertices must exist. If the distance in $G$ between $u'$ and $v'$ is even, of length $2i$, then there is an $\H_{2i}$ subgraph in $G$ with central path from $u'$ to $v'$. If the distance in $G$ between $v'$ and $w'$ is even, of length $2i$, then there is an $\H_{2i}$ subgraph in $G$ with central path from $v'$ to $w'$. Else, there is a path of even length $4i$ from $u'$ to $w'$ and then there is an $\H_{4i}$ subgraph in $G$ with central path from $u'$ to $w'$.

For $(\H_n,\H_{2n},\ldots)$-subgraph-free graphs, consider the Abelian group $(\mathbb{Z}/n\mathbb{Z})$. The Davenport constant of an Abelian group $G$ is the minimum $d$ so that any sequence of elements of $G$ contains a non-empty consecutive subsequence of zero-sum (that adds to the identity element $0$). It is known that for $(\mathbb{Z}/n\mathbb{Z})$ the Davenport constant is $n$ (see page 24 in \cite{GR09}). Take an $(n+1) \times (n+1)$-grid and consider some row not at the top or bottom of the grid with vertices $w_1,\ldots,w_{n+1}$ in succession. Consider some $w'_1 \in f^{-1}(w_1), \ldots, w'_{n+1} \in f^{-1}(w_{n+1})$ where $f$ is the minor map as before, and the distances $x_i$ between $w'_{i+1}$ and $w'_i$. Using the Davenport constant, there is a subsequence $x_j,\ldots,x_{j'}$ ($j'>j$) such that $x_j+\ldots+x_{j'}=0 \bmod n$. Now choose $w'_j,\ldots,w'_{j'+1}$ as the central path in some $\H_{in}$.
\end{proof}

\section{The Standard \NP-hardness Reduction to Star-$\mathbf{3}$-Colouring}\label{a-al}

For reference, we explain the gadget from Albertson et al.~\cite{ACKKR04} that yields the following result.

\begin{theorem}[\cite{ACKKR04}]\label{t-s3c}
{\sc Star $3$-Colouring} is \NP-complete for planar bipartite graphs in which one partition class has size~$2$.
\end{theorem}

\begin{proof}
Reduce from {\sc $3$-Colourability} which is known to be \NP-complete even for planar graphs~\cite{Da80}.
Let $G$ be a planar graph. Replace each edge $e$ by three new vertices $a_e$, $b_e$, $c_e$ that are made adjacent only to the two end-vertices of $e$ in $G$. Let $G'$ be the resulting graph. Then every vertex of $V(G')\setminus V(G)$ has degree~$2$ in $G$. Moreover, $G'$ is planar and bipartite with partition classes $V(G')\setminus V(G)$ and $V(G)$.
It remains to observe that $G$ has a $3$-colouring if and only if $G'$ has a star $3$-colouring.
\end{proof}

\section{The Standard \NP-hardness Reduction to $\mathbf{C_5}$-Colouring}\label{a-c5easy}

We make the following observation.

\begin{proposition}
{\sc $C_5$-Colouring} is \NP-complete for $(\H_1,\H_2,\H_4,\H_5,\ldots)$-subgraph-free graphs.
\end{proposition}

\begin{proof}
It is well known~\cite{HN90} and easy to see that there is a reduction from $K_5$-{\sc Colouring}, which is to decide if a graph has a $K_5$-colouring, that is, a homomorphism from $G$ to the complete graph $K_5$ on five vertices. This problem is well known to be \NP-complete~\cite{HN90}. Let $G$ be a graph, and let $G'$ be the $2$-subdivision of $G$. 
We note that $G'$ is $(\H_1,\H_2,\H_4,\H_5,\ldots)$-subgraph-free (but may contain many instances of $\H_\ell$ where $\ell=0\mod 3$). Moreover, $G$ has a $K_5$-colouring if and only if $G'$ has a $C_5$-colouring.
\end{proof}

\section{Missing Proofs of Claims in Theorem~\ref{t-hhh}}\label{a-hhh}

Here is the missing proof of Claim~\ref{c-butterfly}.

\medskip
\noindent
{\bf Claim~\ref{c-butterfly} (restated).}
{\it Let $G'$ be a graph obtained from $G$ by the butterfly reduction. 
	Then $G$ has a Hamiltonian cycle if and only if $G'$ has a Hamiltonian cycle.}

\begin{proof}
	Since $G'$ is obtained from $G$ be deleting some edges, any Hamiltonian cycle in
	$G'$ is also a Hamiltonian cycle  in $G$. Conversely, assume $C$ is a Hamiltonian 
	cycle in $G$. If the edge $ab$ belongs to $C$, then the edges $ax$ and $bx$ 
	do not belong to $C$ (remember that $a$ and $b$ have white neighbours), 
	in which case $C$ is also a Hamiltonian cycle in $G'$. 
	
	Suppose now that $ab$ does not belong to $C$. Then  $ax$ and $bx$ belong to $C$
	and hence $cx$ and $dx$ do not belong to $C$, implying that $cd$ belongs to $C$ 
	(remember that the degree of $c$ is $3$ in $G$). But then $G$ has one more 
	Hamiltonian cycle obtained from $C$ be replacing the edges $ax$, $bx$, $cd$
	with the edges  $cx$, $dx$, $ab$. This second cycle is clearly a Hamiltonian cycle
	in $G'$. 
\end{proof}

\noindent
Here is the missing proof of Claim~\ref{claim:ld}.

\medskip
\noindent
{\bf Claim~\ref{claim:ld} (restated).}
{\it Let $x$ be a vertex of degree at least $13$. If the neighbourhood of $x$ does not 
	contain two adjacent vertices of degree $3$, then $G$ has no Hamiltonian cycle. 
	Otherwise, $G$ has a Hamiltonian cycle if and only if $G-x$ has.}

\begin{proof}
	Assume $G$ has a Hamiltonian cycle $C$. We call any edge of $G$ not in $C$ a {\it chord}
	and the endpoints of a chord {\it weak neighbours}.
	
	Let $y$ and $z$ be the neighbours of $x$ on the cycle.  
	Without loss of generality, we let $y$ be black (otherwise (R4) is applicable to $G$) and 
	$v$ be a weak neighbour of $y$. Assume $v\ne z$. Since the degree of $x$ is at least $13$,
	$x$ has at least $5$ weak neighbours either between $y$ and $v$ or between $z$ and $v$ on the cycle $C$.
	Let $u$ be a middle of these $5$ neighbours. In the first case, two edges of $C$ incident to 
	$y$, two edges of $C$ incident to $u$, two edges of $C$ incident to $v$ together 
	with the edges $xu$ and $yv$ form a forbidden subgraph. In the second case, two edges of $C$ incident to 
	$x$, two edges of $C$ incident to $u$, two edges of $C$ incident to $v$ together 
	with the edges $xu$ and $yv$ form a forbidden subgraph. A contradiction in both cases shows that $v=z$,
	i.e. the degree of $y$ is $3$. By symmetry, the degree of $z$ is $3$. This proves that 
	if the neighbourhood of $x$ does not contain two adjacent vertices of degree $3$, 
	then $G$ has no Hamiltonian cycle. 
	
	Now assume that the neighbourhood of $x$ contains two adjacent vertices of degree $3$.
	Suppose first that $G$ has a Hamiltonian cycle. From the first part of the proof, we 
	know that the neighbours of $x$ on the cycle are adjacent. Therefore, by removing $x$ 
	from the graph we are left with a Hamiltonian cycle in $G-x$.
	
	Conversely, let $G-x$ have a Hamiltonian cycle $C$. Denote by $y$ and $z$ two adjacent vertices of degree $3$
	in the neighbourhood of $x$. In the graph $G-x$ the vertices $y$ and $z$
	have degree $2$ and hence the edge $yz$ belongs to $C$. By replacing this edge with the edges
	$xy$ and $xz$ we obtain a Hamiltonian cycle in the graph $G$. 
\end{proof}

\section{The Proof of Theorem~\ref{t-kip}}\label{a-kip}

Here is the missing proof of Theorem~\ref{t-kip}, which we restate below.

\medskip
\noindent
{\bf Theorem~\ref{t-kip}. (restated)}
{\it For all~$k\geq 2$, {\sc $k$-Induced Disjoint Paths} is polynomial-time solvable for $\H_1$-subgraph-free graphs.}

\begin{proof}
We prove the result for $k=2$. The extension to $k\geq 2$ will be straightforward. Let $G$ be an instance of {\sc $2$-Induced Disjoint Paths} together with two terminal pairs $(s_1,t_1)$ and $(s_2,t_2)$. We may assume without loss of generality that there is no edge between $s_1$ and $t_1$ and no edge between $s_2$ and $t_2$. 

We first check if there exists a solution in which one of the paths has length~$2$. We can do this in polynomial time as follows. We first consider all $O(n)$ options of choosing a vertex to be the middle vertex of one of these paths. We then check if the graph obtained from removing the guessed middle vertex and its two neighbouring terminals $s_i$ and $t_i$ as well all the neighbours of these three vertices has a connected component that contains both terminals $s_j$ and $t_j$ of the other pair. This takes polynomial time.

We now check if there exists a solution in which both paths have length at least $3$. We consider all $O(n^4)$ options of choosing the neighbours $s_1'$, $t_1'$, $s_2'$, $t_2'$ of $s_1$, $t_1$, $s_2$, $t_2$, respectively, on the two solution paths (should a solution exist).  We discard a branch if there exists an edge between a vertex of $\{s_1,s_1',t_1,t_1'\}$ and a vertex of $\{s_2,s_2',t_2,t_2'\}$. Suppose this is not the case. We remove $s_1,t_1,s_2,t_2$ and every neighbour of a vertex in $\{s_1,t_1,s_2,t_2\}$ that does not belong to $\{s_1',t_1',s_2',t_2'\}$. Afterwards, it suffices to solve $2$-{\sc Disjoint Paths} on the resulting graph $G'$ with terminal pairs $(s_1',t_1')$ and $(s_2',t_2')$. This can be seen as follows. Any solution of $2$-{\sc Induced Disjoint Paths} is a solution of {\sc $2$-Disjoint Paths}. Now suppose we have a solution $(P_1,P_2)$ of {\sc $2$-Disjoint Paths}. If there exist an edge between a vertex of $P_1$ and a vertex of $P_2$, then we find the forbidden subgraph $\H_1$ (possibly after adding the vertices $s_1,t_1,s_2,t_2$ back).
Since the number of branches is $O(n^4)$ and each created instance of $2$-{\sc Disjoint Paths} can be solved in polynomial time~\cite{RS95,Sh80}, the running time of this case is polynomial as well. 
\end{proof}

\section{The Missing Parts of Theorem~\ref{t-kip2}}\label{a-kip2}

Here we prove that Rules~1 and~2 displayed in Figures~\ref{fig:rule-1-bis} and~\ref{fig:rule-2-bis} are safe, and moreover that they both preserve $\dagger$ and $H_2$-subgraph-freeness (this was all what was left to show).

\medskip
\noindent
Recall that Rule~1 is applied in the situation that $(z_1,z_2)$ and $(z_1,z_4)$ are both not edges of the graph. The rule is to contract the edge $(x_1,x_2)$, removing any parallel edges that may arise; see Figure~\ref{fig:rule-1-bis}. We first show that Rule~1 is safe. 

Suppose that we have a solution to {\sc $k$-Induced Disjoint Paths} in $G$. If this solution uses no vertices in $S$, then it is already a solution to {\sc $k$-Induced Disjoint Paths} in $G'$. Thus, it must use some vertex in $S$. If the solution does not use $x_1$ nor $x_2$, then recall that by $\ddagger$, each of $z_1,z_2,z_3,z_4$ has at most one neighbour outside of $S$, and thus the solution must avoid thus $S$ entirely, a contradiction.
If the solution uses both $x_1$ and $x_2$, then it must use the edge $(x_1,x_2)$. We can substitute the edge $(x_1,x_2)$ in the solution to {\sc $k$-Induced Disjoint Paths} in $G$ with $x$ to obtain a solution to {\sc $k$-Induced Disjoint Paths} in $G'$. 
Hence, without loss of generality, suppose the solution uses $x_1$. We can substitute this for $x$ to obtain a solution to {\sc $k$-Induced Disjoint Paths} in $G'$, unless some other solution path runs through a neighbour $q$ of $x_2$. Note $q$ cannot be a terminal due to our preprocessing. Hence it has two neighbours $p$ and $r$ on this other solution path, and these are outside of $\{z_1,x_1,z_3\}$ because this path must avoid $x_1$ and any of its neighbours. But now $p,q,r$, $q,x_2,x_1$, $z_1,x_1,z_3$ forms an $\H_2$ (with middle path $q,x_2,x_1$), a contradiction.

Suppose we have a solution to {\sc $k$-Induced Disjoint Paths} in $G'$. If this solution does not involve $x$, then it maps to a solution of {\sc $k$-Induced Disjoint Paths} in $G$. Suppose now it does involve $x$. Suppose mapping $x$ to either of $x_1$ or $x_2$ does not produce a solution to {\sc $k$-Induced Disjoint Paths} in $G$. Then mapping $x$ to either the edge $(x_1,x_2)$ (or the symmetric $(x_2,x_1)$) must produce a solution to {\sc $k$-Induced Disjoint Paths} in $G$.

\medskip
\noindent
Recall that Rule~2 is applied if at least one of $(z_1,z_2)$ and $(z_1,z_4)$ is an edge of the graph and $x_1,x_2$ do not have any neighbours outside $S$. The rule is to contract the edge $(x_1,x_2)$, removing any parallel edges that may arise; see Figure~\ref{fig:rule-2-bis}. We first show that Rule~2 is safe. 

Suppose we have a solution to {\sc $k$-Induced Disjoint Paths} in $G$. If it uses no vertices in $S$, then it is already a solution to {\sc $k$-Induced Disjoint Paths} in $G'$. Thus, it must use some vertex in $S$. Suppose the edge $(z_1,z_2)$ exists and the solution uses $(z_1,z_2)$. Then by $\ddagger$, the solution does not use any other vertex from $S$ and we can keep this edge to obtain a solution for {\sc $k$-Induced Disjoint Paths} in $G'$. Suppose the edge $(z_1,z_4)$ exists and the solution uses $(z_1,z_4)$. Then by $\ddagger$, the solution does not use any other vertex from $S$ and we can keep this edge to obtain a solution for {\sc $k$-Induced Disjoint Paths} in $G'$.

If the solution uses both $x_1$ and $x_2$, then it must use the edge $(x_1,x_2)$ and we can substitute $(x_1,x_2)$ in the solution to {\sc $k$-Induced Disjoint Paths} in $G$ with $x$ to obtain a solution to {\sc $k$-Induced Disjoint Paths} in $G'$. Suppose it uses neither of $x_1$ and $x_2$. Then by $\ddagger$ and the fact that $S$ is used, the solution must use either the edge $(z_1,z_4)$ or $(z_1,z_2)$ and we are in a previous case.

Hence, without loss of generality, suppose the solution uses $x_1$. We can substitute this for $x$ to obtain a solution to {\sc $k$-Induced Disjoint Paths} in $G'$. This is safe, as $x_1,x_2$ have no neighbours outside $S$.

Suppose we have a solution to {\sc $k$-Induced Disjoint Paths} in $G'$. If this solution does not involve $x$ then it maps to a solution of {\sc $k$-Induced Disjoint Paths} in $G$. Suppose now it does involve $x$. Suppose mapping $x$ to either of $x_1$ or $x_2$ does not produce a solution to {\sc $k$-Induced Disjoint Paths} in $G$. Then mapping $x$ to either the edge $(x_1,x_2)$ (or the symmetric $(x_2,x_1)$) must produce a solution to {\sc $k$-Induced Disjoint Paths} in $G$. 

\medskip
\noindent
Next, we show that any graph $G'$, obtained after applying Rule~1 or~2 is $\H_2$-subgraph-free as well. Suppose $G'$ has an $\H_2$. Then this $\H_2$ must contain $x$. If $x$ is a leaf vertex in $\H_2$, then it is clear that $G$ already had this $\H_2$ involving either $x_1$ or $x_2$. 

Suppose $x$ is a degree-$3$ vertex in this $\H_2$. If the neighbours of $x$ in the $\H_2$ were all neighbours of $x_1$ or all neighbours of $x_2$ in $G$, then it is clear that $G$ already had this $\H_2$, a contradiction. Let $z'_1$ and $z'_2$ be the leafs of the $\H_2$ adjacent to $x$ in $G'$. 

Suppose that $z'_1$ and $z'_2$ are both adjacent to $x_2$ and both not to $x_1$. Then the middle vertex of the $\H_2$ is only adjacent to $x_1$. Ideally, we would replace $x$ by $x_1$, $z'_1$ by $z_1$ and $z'_2$ by $z_3$. This does not work if (say) $z_1$ is part of the $\H_2$. However, $z'_1$ and $z'_2$ are both not $z_1$, because $z_1$ is adjacent to $x_1$ and we would contradict our assumption on the adjacency of $z'_1$ and $z'_2$. We now consider three cases, depending on where $z_1$ is in the $\H_2$.

Suppose that $z_1$ is a leaf of the $\H_2$. By $\ddagger$ and the inducedness of paths, its neighbouring degree-3 vertex cannot be one of $z_2,z_3,z_4$. Hence, this must be the unique neighbour $p$ of $z_1$ outside $S$. The other neighbours $q,r$ of $p$ on the $\H_2$, where $r$ is the middle vertex, are both not $z_3$ since $P_1$ is induced. Hence, $G$ has a $\H_2$ formed by $q,p,z_1$, $p,r,x_1$, $x_2,x_1,z_3$, a contradiction.

Suppose that $z_1$ is the middle vertex of the $H_2$. By $\ddagger$, the other degree-3 vertex of the $\H_2$ cannot be $z_2$ or $z_4$, so it must be the unique neighbour $p$ of $z_1$ outside $S$. The other neighbours $q,r$ of $p$ on the $\H_2$, which are both leafs of the $H_2$, are both not $z_3$ since $P_1$ is induced. Hence, $G$ has a $\H_2$ formed by $q,p,r$, $p,z_1,x_1$, $x_2,x_1,z_3$, a contradiction.

Suppose that $z_1$ is a degree-3 vertex of the $\H_2$. Let $p$ be the unique neighbour of $z_1$ outside $S$; it is unique by $\ddagger$. Then one of $p,z_2,z_3$ must be the middle vertex of the $\H_2$ and the other two the leafs neighbouring $z_1$. If the middle vertex is $z_2$, then $z'_1,x_2,z'_2$, $x_2,z_2,z_1$, $p,z_1,z_4$ is a $\H_2$ in $G$, a contradiction. The other cases are similar. 

This concludes the argument when $z'_1$ and $z'_2$ are both adjacent to $x_2$ and both not to $x_1$. 

Suppose instead that, say $z'_1$, is adjacent to $x_1$ and the other, $z'_2$, is adjacent to $x_2$. Let $x',x'',z''_1,z''_2$ form the remaining vertices of the $\H_2$ where $x,x',x''$ and $z''_1,x'',z''_2$ are both paths of length $2$ in this $\H_2$. Thus, $z'_1,x,z'_2$, $x,x',x''$ and $z''_1,x'',z''_2$ form the $\H_2$ in $G'$. Without loss of generality, suppose $x'$ was adjacent to $x_1$ in $G$. Now it is clear that $z'_1,x_1,x_2$, $x_1,x',x''$ and $z''_1,x'',z''_2$ formed an $\H_2$ in $G$.

Finally, suppose that $x$ is the degree-$2$ vertex in $\H_2$. Let $z'_1,x',z'_2$, $x',x,x''$, $z''_1,x'',z''_2$ be the paths that form the $\H_2$ in $G'$. Suppose, without loss of generality, that $x'$ was adjacent to $x_1$ in $G$. If $x''$ was also adjacent to $x_1$ in $G$, then $z'_1,x',z'_2$, $x',x_1,x''$, $z''_1,x'',z''_2$ are paths that form an $\H_2$ in $G$. Suppose now that $x''$ was adjacent to $x_2$ but not $x_1$ in $G$ and we may also assume that $x'$ is adjacent to $x_1$ but not $x_2$. Now $z'_1,x',z'_2$, $x',x_1,x_2$, $z_2,x_2,z_4$ are paths that form a $\H_2$ in $G$, unless $\{z_2,z_4\}\cap \{z'_1,z'_2\}\neq \emptyset$. Without loss of generality, suppose $z_2=z'_1$. Note that $z_2\neq s_2$ (recall that $S$ does not contain any terminal). Let $p$ be the next vertex on the path from $t_2$ to $s_2$ after $z_2$. Then $p,z_2,x_2$, $z_2,x',x_1$, $z_1,x_1,z_3$ is an $\H_2$ in $G$ (note that $\{z_1,z_3\} \cap \{x',z_2,p\}=\emptyset$), a contradiction.

\medskip
\noindent
Finally, it remains to show that $\dagger$ is preserved. Note that $x_1$ and $x_2$ cannot be $z$ or $a_z$ for some terminal $z$, as these vertices have degree~$1$ and~$2$ respectively, while $x_1,x_2$ have degree at least~$3$. Moreover, $\{x_1,x_2\} \not= \{b_z,b_z'\}$ for some terminals $z,z'$ by our preprocessing. Hence, $\dagger$ is preserved.

\section{The Proof of Theorem~\ref{t-khard}}\label{a-khard}

In this section, we prove that for all $k\geq 2$, {\sc $k$-Induced Disjoint Paths} is \NP-complete for $(\H_4,\ldots,\H_\ell)$-subgraph-free graphs for all $\ell\geq 4$.
L{\'{e}}v{\^{e}}que et al.~\cite{LLMT09} considered the problem  $2$-{\sc Induced Cycle}. This problem has as input a graph $G$ with two specified vertices $x$ and $y$ that are not adjacent to each other and have degree~$2$.
The question is whether $G$ has an induced cycle containing $x$ and $y$. This problem was shown to be \NP-complete by Bienstock~\cite{Bi91}.
L{\'{e}}v{\^{e}}que et al.~\cite{LLMT09} proved that {\sc $2$-Induced Cycle} is \NP-complete even for subcubic graphs (under various restrictions).

We first prove that {\sc $2$-Induced Cycle} is \NP-complete for subcubic $(\H_4,\ldots,\H_\ell)$-subgraph-free graphs for all $\ell\geq 4$; afterwards, we use a lemma from~\cite{MPSL23} to prove Theorem~\ref{t-khard}.
As mentioned, we follow the argument from Section 2.4 in \cite{LLMT09} by subdividing certain edges a sufficient number of times. 
Indeed, our gadgets are precisely those from \cite{LLMT09} with some edges subdivided $\ell-1$ times. These edges are drawn in dashed lines in our gadgets in Figures~\ref{fig:2IDP-literal},~\ref{fig:2IDP-clause} and~\ref{fig:2IDP-variable}. Thus, the dashed edges represent $\ell$-paths.

\begin{theorem}\label{t-iii} 
{\sc $2$-Induced Cycle} is \NP-complete for subcubic $(\H_4,\ldots,\H_\ell)$-subgraph-free graphs for all $\ell\geq 4$.
\end{theorem}

\begin{proof}
The problem is ready seen to be in \NP. We will now prove \NP-hardness.

Let $\ell\geq 1$ be an integer. Let $\phi$ be an instance of {\sc 3-Satisfiability}, consisting of $m$ clauses $C_1,\ldots,C_m$ on $n$ variables $z_1,\ldots,z_n$. For each clause $C_j$ ($j=1,\ldots,m$), with $C_j = y_{3j-2} \vee y_{3j-1} \vee y_{3j}$, then $y_i$ ($i = 1,\ldots,3m$) is a literal from $\{z_1,\ldots,z_n, \overline{z}_1,\ldots,\overline{z}_n\}$.

Let us build a graph $G^\ell_\phi$ with two specified vertices $x$ and $y$ of degree $2$. There will be a hole containing $x$ and $y$ in $G_\phi$ if and only if there exists a truth assignment satisfying $\phi$.

For each literal $y_j$ ($j = 1,\ldots,3m$), prepare a graph $G(y_j)$ on 20 named vertices \[\alpha,\alpha',\alpha^{1+},\ldots,\alpha^{4+},\alpha^{1-},\ldots,\alpha^{4-},\beta,\beta',\beta^{1+},\ldots,\beta^{4+},\beta^{1-},\ldots,\beta^{4-},\] together with certain paths in between using unnamed vertices, as drawn in Figure~\ref{fig:2IDP-literal}. (We drop the subscript $j$ in the labels of the vertices for clarity.)

For $i = 1, 2, 3$ add paths of length $\ell$ between $\alpha^{i+}$ and $\alpha^{(i+1)+}$; $\alpha^{i-}$ and $\alpha^{(i+1)-}$; $\beta^{i+}$ and $\beta^{(i+1)+}$; and $\beta^{i-}$ and $\beta^{(i+1)-}$. Also add the edges $\alpha^{1+}\beta^{1-}$, $\alpha^{1-}\beta^{1+}$, $\alpha^{4+}\beta^{4-}$, $\alpha^{4-}\beta^{4+}$, $\alpha \alpha^{1+}$,  $\alpha \alpha^{1-}$, $\alpha^{4+} \alpha'$, $\alpha^{4-} \alpha'$, $\beta \beta^{1+}$,  $\beta \beta^{1-}$, $\beta^{4+} \beta'$, $\beta^{4-} \beta'$.

\begin{figure}[t]
\vspace*{-1cm}
\[
\xymatrix{
& \alpha^{1+} \ar@/_1.5pc/@{-}[ddddd] \ar@{--}[r] & \alpha^{3+} \ar@{--}[r]  & \alpha^{3+} \ar@{--}[r] & \alpha^{4+} \ar@/^1.5pc/@{-}[ddddd] \ar@{-}[dr] \\
\alpha \ar@{-}[ur] \ar@{-}[dr] & & & & & \alpha' \\
& \alpha^{1-} \ar@{-}[d] \ar@{--}[r] &  \alpha^{2-} \ar@{--}[r] & \alpha^{3-} \ar@{--}[r]  & \alpha^{4-} \ar@{-}[d] \ar@{-}[ur] \\
& \beta^{1+} \ar@{--}[r] & \beta^{2+} \ar@{--}[r] & \beta^{3+} \ar@{--}[r]  & \beta^{4+} \ar@{-}[dr] \\
\beta \ar@{-}[ur] \ar@{-}[dr] & & & & &  \beta' \\
& \beta^{1-} \ar@{--}[r] &  \beta^{2-} \ar@{--}[r] & \beta^{3-} \ar@{--}[r]  & \beta^{4-} \ar@{-}[ur] \\
}
\]
\vspace*{-0.3cm}
\caption{The literal gadget 
(dashed lines indicate paths of length $\ell$).}
\label{fig:2IDP-literal}
\vspace*{-0.0cm}
\end{figure}
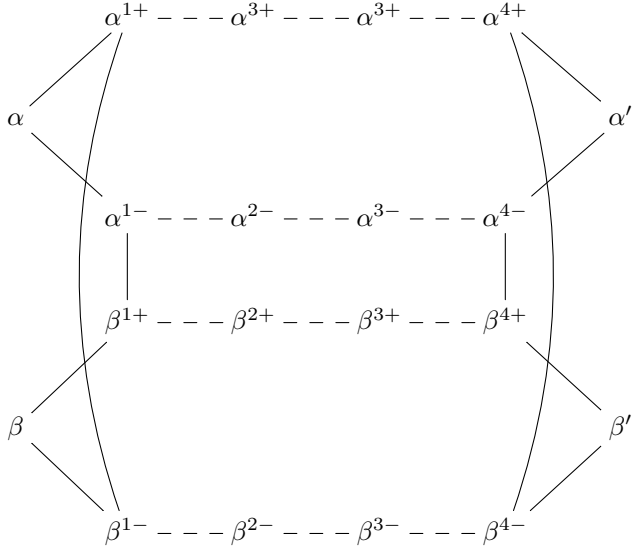
For each clause $C_j$ ($j = 1,\ldots,m$), prepare a graph $G(C_j)$ with 10 named vertices \[c^{1+}, c^{2+}, c^{3+}, c^{1-}, c^{2-}, c^{3-}, c^{0+}, c^{12+}, c^{0-}, c^{12-},\] together with certain paths in between using unnamed vertices, as drawn in Figure~\ref{fig:2IDP-clause}. (We drop the subscript j in the labels of the vertices for clarity.)
Add paths of length $\ell$ between the following pairs of vertices: $c^{12+}$ and $c^{1+}$; $c^{12+}$ and $c^{2+}$; $c^{12-}$ and $c^{1-}$; $c^{12-}$ and $c^{2-}$; $c^{0+}$ and $c^{12+}$; $c^{0+}$ and $c^{3+}$; $c^{0-}$ and $c^{12-}$; $c^{0-}$ and $c^{3-}$.
\begin{figure}
\[
\xymatrix{
& \alpha^{1-} \ar@{--}[r]  &  \alpha^{2-} \ar@{--}[r] & \alpha^{3-} \ar@{--}[r]  &  \alpha^{4-}  \\
& \beta^{1-}  \ar@{--}[r]  & \beta^{2-} \ar@{--}[r] & \beta^{3-} \ar@{--}[r]   & \beta^{4-}  \\
 & & c^{1+} \ar@{-}[u] \ar@/^1pc/@{-}[uu] & c^{1-} \ar@{-}[u]  \ar@/_1pc/@{-}[uu] &  & \\
}
\]

\[
\xymatrix{
& & c^{1+} & c^{1-} \ar@{--}[dr] & & \\
& c^{12+} \ar@{--}[dr] \ar@{--}[ur] & & & c^{12-} \\
c^{0+}  \ar@{--}[drr] \ar@{--}[ur] & &  c^{2+} & c^{2-} \ar@{--}[ur]  & & c^{0-} \ar@{--}[dll] \ar@{--}[ul]   \\
& & c^{3+} &  c^{3-} & & \\
}
\]
\vspace*{-0.5cm}
\caption{Clause gadget and above, its interface with the literal gadget. Dashed lines are paths of length $\ell$.}
\label{fig:2IDP-clause}
\end{figure}
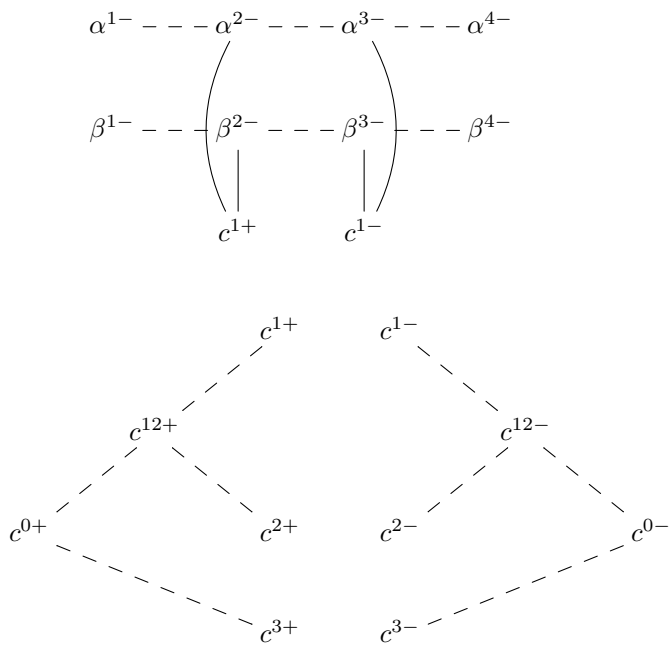
\begin{figure}
\vspace*{-0.7cm}
\[
\xymatrix{
& p^+_{i,1} \ar@{--}[dl] \ar@{--}[r] & p^+_{i,2} \ar@{--}[r] & \bullet \ar@{--}[r] & \bullet \ar@{.}[r]^{P^+} & \bullet \ar@{--}[r] & \bullet  \ar@{--}[r] & p^+_{i,2z^+_i-1} \ar@{--}[r] & p^+_{i,2z^+_i} \ar@{--}[dr] & \\
d^{+} & & & & & & &  & & d^{-} \\
& p^-_{i,1} \ar@{--}[ul] \ar@{--}[r] &  p^-_{i,2} \ar@{--}[r] & \bullet \ar@{--}[r] & \bullet \ar@{.}[r]_{P^-} & \bullet \ar@{--}[r] & \bullet  \ar@{--}[r] & p^-_{i,2z^-_i-1} \ar@{--}[r] & p^-_{i,2z^-_i} \ar@{--}[ur] & \\
}
\]
\vspace*{-0.5cm}
\caption{The variable gadget.
 Dashed lines are paths of length $\ell$. Dotted lines are  a continuation of the gadget.}
\label{fig:2IDP-variable}
\end{figure}
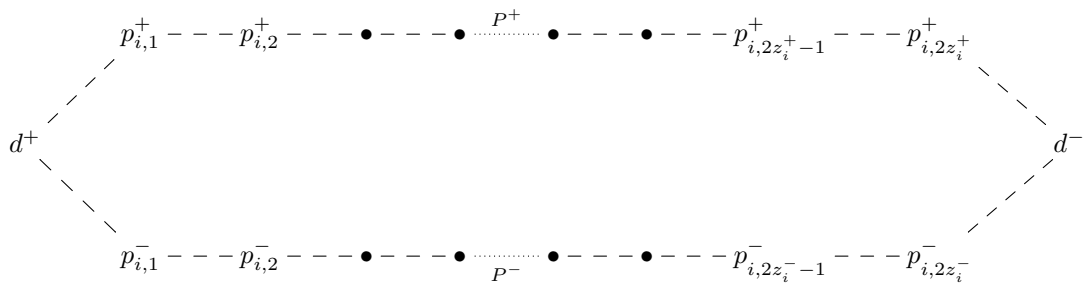

For each variable $z_i$ .($i = 1,\ldots,n$), prepare a graph $G(z_i)$ with $2z_i^- + 2z^+_
i$ vertices, where $z^-_i$ is the number of times $\overline{z}_i$ appears in clauses $C_1, \ldots,C_m$ and $z^+_i$ is the number of times $z_i$ appears in clauses $C_1, \ldots,C_m$.

Let $G(z_i)$ consist of two internally disjoint paths $P^+_i$ and $P^-_i$ with common endpoints $d^+_i$ and $d^-_i$ and lengths $1+ (2\ell)z^-_i$ and $1+ (2\ell)z^+_i$, respectively. Label the vertices of $P^+_i$ and $P^-_i$ as in Figure~\ref{fig:2IDP-variable}.

The final graph $G^\ell_\phi$ will be constructed from the disjoint union of all the graphs $G(y_j)$, $G(C_i)$, and $G(z_i)$ with the following modifications:
\begin{itemize}
\item For $j = 1,\ldots, 3m-1$, add paths of length $\ell$ between the pairs: $\alpha'_j$ and $\alpha_{j+1}$; $\beta'_j$ and $\beta_{j+1}$.
\item For $j = 1,\ldots, m-1$, add a path of length $\ell$ between $c^{0-}_j$ and $c^{0+}_{j+1}$.
\item For $j = 1,\ldots, n-1$, add a path of length $\ell$ between $d^{-}_i$ and $d^{+}_{i+1}$.
\item For $i = 1,\ldots,n-1$, let $y_{n_1},\ldots,y_{n_{z^-_i}}$ be the occurrences of $\overline{z}_i$ over all literals. For $j = 1,\ldots,z^-_i$, delete the path between $p^+_{i,2j-1}$ and $p^+_{i,2j}$ and add the four edges  $p^+_{i,2j-1}\alpha^{2+}_{n_j}$, $p^+_{i,2j-1}\beta^{2+}_{n_j}$, $p^+_{i,2j}\alpha^{3+}_{n_j}$, $p^+_{i,2j}\beta^{3+}_{n_j}$.
\item For $i = 1,\ldots,n-1$, let $y_{n_1},\ldots,y_{n_{z^+_i}}$ be the occurrences of $z_i$ over all literals. For $j = 1,\ldots,z^+_i$, delete the path between $p^-_{i,2j-1}$ and $p^-_{i,2j}$ and add the four edges  $p^-_{i,2j-1}\alpha^{2+}_{n_j}$, $p^-_{i,2j-1}\beta^{2+}_{n_j}$, $p^-_{i,2j}\alpha^{3+}_{n_j}$, $p^-_{i,2j}\beta^{3+}_{n_j}$.
\item For $i = 1,\ldots,m$ and $j=1,2,3$, add the edges $\alpha^{2-}_{3(i-1)+j}c^{j+}_i$, $\alpha^{3-}_{3(i-1)+j}c^{j-}_i$, $\beta^{2-}_{3(i-1)+j}c^{j+}_i$, $\beta^{3-}_{3(i-1)+j}c^{j-}_i$.
\item Add a path of length $\ell$ between the pairs of vertices: $\alpha'_{3m}d^+_1$ and $d^+_1$; $\beta'_{3m}d^+_1$ and $c^{0+}_1$.
\item Add the vertex $x$ and add paths of length $\ell$ between the pairs of vertices: $x$ and $\alpha_1$; $x$ and $\beta_1$.
\item Add the vertex $y$ and add paths of length $\ell$ between the pairs of vertices: $y$ and $c^{0-}_m$; $y$ and $d^{-}_n$.
\end{itemize}

It is easy to verify that the maximum degree of $G^\ell_\phi$ is $3$. Moreover, the size of $G^\ell_\phi$  is polynomial (actually linear) in the size $n+m$ of $\phi$, and $x$ and $y$ are non-adjacent, and both have degree~$2$. 

We also claim that for $\ell\geq 4$, $G^\ell_\phi$ is $\H_4,\ldots,\H_{\ell}$-subgraph-free, which can be seen as follows (note that $G^\ell_\phi$ contains $\H_1$, $\H_2$ and $\H_3$ as subgraphs).
Owing to the length of the $\ell$-paths that populate our construction, we need only verify the omission of these graphs on the connected components of the graph $G^\ell_\phi$ after the removal of these paths (except a pendant edge from the corresponding connected component at the extremities of an instance of these paths). In this fashion, we only need to check for omission of the given graphs in the non-trivial cases drawn in Figure~\ref{fig:2IDP-omission-cases}.
Indeed, the two cases are isomorphic. Let $i=4,5$. Any two vertices of degree at least three that are separated by a path of length $i$ must be in the subgraph $C_6$ at distance $6-i$ from one another. If $i=4$ then these vertices have a common neighbour so the $\H_i$ can't be completed. If $i=5$ then these two vertices are adjacent. For $6\leq i \leq \ell-i$ it is not possible to find two vertices of degree at least three that are separated by a path of length $\ell$.

\begin{figure}[tbp]
\vspace*{-0.2cm}
$
\xymatrix{
& & \alpha^{1+} \ar@/_1.5pc/@{-}[ddddd]  \ar@{-}[r] & \bullet \\
\bullet \ar@{-}[r]  & \alpha \ar@{-}[ur] \ar@{-}[dr] & & & \\
& & \alpha^{1-} \ar@{-}[d] \ar@{-}[r]  &  \bullet \\
& & \beta^{1+} \ar@{-}[r] &\bullet \\
\bullet \ar@{-}[r] & \beta \ar@{-}[ur] \ar@{-}[dr] & & &  \\
& & \beta^{1-} \ar@{-}[r] &  \bullet \\
}
$
\ \ \
$
\xymatrix{
\bullet  \ar@{-}[r] &  \alpha^{2-} \ar@{-}[r] & \alpha^{3-} \ar@{-}[r] & \bullet  \\
\bullet \ar@{-}[r] & \beta^{2-} \ar@{-}[r] & \beta^{3-} \ar@{-}[r]  & \bullet  \\
& c^{1+} \ar@{-}[u]  \ar@/^1pc/@{-}[uu] & c^{1-} \ar@{-}[u]  \ar@/_1pc/@{-}[uu] &  \\
\bullet \ar@{-}[ur] & & & \bullet \ar@{-}[ul] \\
}
$
\vspace*{-0.2cm}
\caption{Cases that need to be checked for omission of graphs $\H_\ell$.}
\label{fig:2IDP-omission-cases}
\vspace*{-0.0cm}
\end{figure}
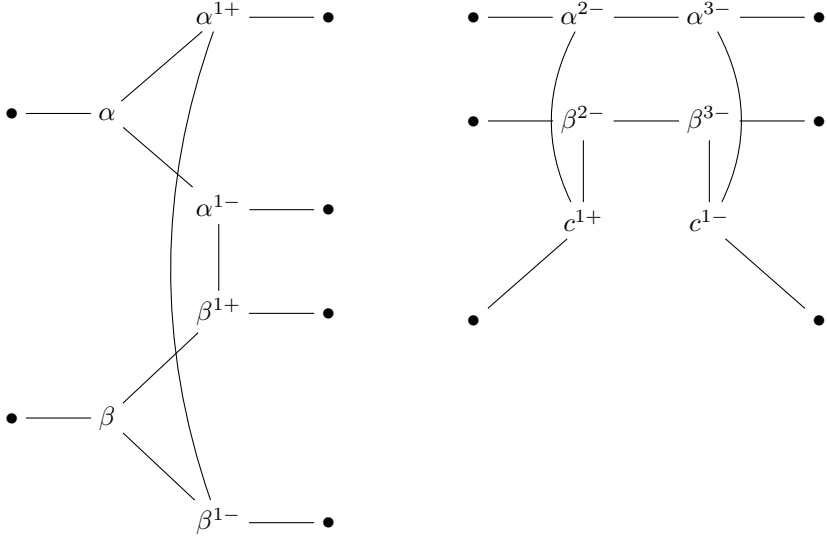

We now show that  $\phi$ is satisfied by a truth assignment if and only if $G^\ell_\phi$ contains a hole passing through $x$ and $y$.

\medskip
\noindent
$``\Rightarrow''$
First assume that $\phi$ is satisfied by a truth assignment $\xi \in \{0,1\}^n$. We will pick a set of vertices that induce a hole containing $x$ and $y$.
\begin{enumerate}
\item Pick vertices $x$ and $y$.
\item For $i = 1,\ldots,3m$, pick the vertices $\alpha_i,\alpha'_i,\beta_i,\beta'_i$.
\item For $i = 1,\ldots,3m$, if $y_i$ is satisfied by $\xi$ , then pick the vertices $\alpha^{1+}_i,\alpha^{2+}_i,\alpha^{3+}_i,\alpha^{4+}_i,\beta^{1+}_i,\beta^{2+}_i$, $\beta^{3+}_i,\beta^{4+}_i$. Otherwise, pick the vertices $\alpha^{1-}_i,\alpha^{2-}_i,\alpha^{3-}_i,\alpha^{4-}_i,\beta^{1-}_i,\beta^{2-}_i,\beta^{3-}_i,\beta^{4-}_i$.
\item For $i = 1,\ldots,n$, if $\xi(i) = 1$, then pick all the vertices of the path $P^+_i$ and all the neighbours of the vertices in $P^+_i$ of the form $\alpha^{2+}_k$ or $\alpha^{3+}_k$ for any $k$.
\item For $i = 1,\ldots,n$, if $\xi(i) = 0$, then pick all the vertices of the path $P^-_i$ and all the neighbours of the vertices in $P^-_i$ of the form $\alpha^{2+}_k$ or $\alpha^{3+}_k$ for any $k$.
\item For $i = 1,\dots,m$, pick the vertices $c^{0+}_i$ and $c^{0-}_i$. Choose any $j \in \{3i-2, 3i-1, 3i\}$ such that $\xi$ satisfies $y_j$. Pick vertices $\alpha^{2-}_j$ and $\alpha^{3-}_j$. If $j = 3i - 2$, then pick the vertices $c^{12+}_j,c^{1+}_j,c^{12-}_j,c^{1-}_j$.  If $j = 3i - 1$, then pick the vertices $c^{12+}_j,c^{2+}_j,c^{12-}_j,c^{2-}_j$.   If $j = 3i$, then pick the vertices $c^{3+}_j,c^{3-}_j$.
\end{enumerate}
The given vertices do not yet induce a connected component, because we need to add the vertices of $\ell$-paths in between. Thus, if $p$ and $q$ are vertices which we selected that have an $\ell$-path between them (drawn as a dashed edge in the associated gadget), then we need to add the interior vertices of this path also.

It suffices to show that the chosen vertices induce a hole containing $x$ and $y$. The only potential problem is that for some $k$, one of the vertices $\alpha^{2+}_k,\alpha^{3+}_k,\alpha^{2-}_k,\alpha^{3-}_k$ was chosen more than once. If $\alpha^{2+}_k$ and $\alpha^{3+}_k$ were picked in Step 3, then $y_k$ is satisfied by $\xi$. Therefore, $\alpha^{2+}_k$ and $\alpha^{3+}_k$ were not chosen in Step 4 or Step 5. Similarly, if $\alpha^{2-}_k$ and $\alpha^{3-}_k$ were picked in Step 6, then $y_k$ is satisfied by $\xi$ and $\alpha^{2-}_k$ and $\alpha^{3-}_k$ were not picked in Step 3. Thus, the chosen vertices induce a hole in $G^\ell_\phi$ containing vertices $x$ and $y$.

\medskip
\noindent
$``\Leftarrow''$
To show the reverse implication, assume $G^\ell_\phi$ contains a hole $H$ passing through $x$ and $y$.
The hole $H$ must contain $\alpha_1$ and $\beta_1$, and the paths leading to them, since they are the only two path neighbours of $x$. 
Next, either both $\alpha^{1+}_1$ and $\beta^{1+}_1$ are in $H$ or both $\alpha^{1-}_1$ and $\beta^{1-}_1$ are in $H$.

Without loss of generality, let $\alpha^{1+}_1$ and $\beta^{1+}$ be in $H$ (the same reasoning that follows will hold true for the other case). Since $\alpha^{1-}_1$ and $\beta^{1-}$ are both neighbours of two members of $H$, they cannot be in $H$. Thus, $\alpha^{2+}_1$ and $\beta^{2+}_1$, and the paths to them, must be in $H$. 
Since $\alpha^{2+}_1$ and $\beta^{2+}_1$ have the same neighbours outside $G(y_1)$, it follows that $H$ must contain $\alpha^{3+}_1$ and $\beta^{3+}_1$, and the paths that lead to them.  Also, $H$ must contain $\alpha^{4+}_1$ and $\beta^{4+}_1$, and the paths that lead to them. Suppose that $\alpha^{4-}_1$ and $\beta^{4-}_1$ are in $H$. Because $\alpha^{i-}_1$ has the same neighbour as $\beta^{i-}_1$ outside $G(y_1)$ for $i = 2, 3$, it follows that H must contain $\alpha^{3-}_1$, $\alpha^{2-}_1$, $\alpha^{1-}_1$. But then $H$ is not a hole containing $x$, a contradiction. Therefore, $\alpha^{4-}_1$ and $\beta^{4-}_1$ cannot both be in $H$, so $H$ must contain $\alpha'_1$, $\beta'_1$, $\alpha_2$, $\beta_2$, and the paths to them.

By induction, we see for $i = 1, 2, \ldots,3m$ that $H$ must contain $\alpha_i,\alpha'_i,\beta_i,\beta'_i$. Also, for each $i$, either $H$ contains $\alpha^{1+}_i,\alpha^{2+}_i,\alpha^{3+}_i,\alpha^{4+}_i,\beta^{1+}_i,\beta^{2+}_i,\beta^{3+}_i,\beta^{4+}_i$ or $H$ contains $\alpha^{1-}_i,\alpha^{2-}_i,\alpha^{3-}_i,\alpha^{4-}_i$, $\beta^{1-}_i,\beta^{2-}_i,\beta^{3-}_i,\beta^{4-}_i$.

As a result, $H^\ell_\phi$ must also contain $d^+_1$ and $c^{0+}_1$ and the paths to them. By symmetry, we may assume $H^\ell_\phi$ contains $p^+_{1,1}$ and $\alpha^{2+}_k$, for some $k$. Since $\alpha^{1+}_k$ is adjacent to two vertices in $H$, $H$ must contain  $\alpha^{3+}_k$ and the path of length $\ell$ toward it. Similarly, $H$ cannot contain $\alpha^{4+}_k$, so $H$ contains $p^+_{1,2}$ and $p^+_{1,3}$, as well as the paths through these. By induction, we see that $H$ contains $p^+_{1,i}$ for $i = 1, 2, \ldots, z^+_i$ and $d^-_1$ and the $\ell$-paths in between. If $H$ contains $p^-_{1,z^-_i}$, then $H$ must contain $p^-_{1,i}$ for $i = z^-_i, \ldots, 1$, a contradiction. Thus, $H$ must contain $d^+_2$ and the $\ell$-path to it. By induction, for $i = 1, 2, \ldots,n$, we see that $H$ contains all the vertices of the path $P^+_i$ or $P^-_i$ and by symmetry, we may assume $H$ contains all the neighbours of the vertices in $P^+_i$ or $P^-_i$ of the form $\alpha^{2+}_k$ or $\alpha^{3+}_k$, for any k.

Similarly, for $i = 1, 2, \ldots,m$, it follows that $H$ must contain $c^{0+}_i$ and $c^{0-}_i$. Also, $H$ contains one of the following:
\begin{itemize}
\item $c^{12+}_i,c^{1+}_i,c^{12-}_i,c^{1-}_i$ and either $\alpha^{2-}_j$ and $\alpha^{3-}_j$ or $\beta^{2-}_j$ and $\beta^{3-}_j$ (where $\alpha^{2-}_j$ is adjacent to $c^{1+}_i$). 
\item $c^{12+}_i,c^{2+}_i,c^{12-}_i,c^{2-}_i$ and either $\alpha^{2-}_j$ and $\alpha^{3-}_j$ or $\beta^{2-}_j$ and $\beta^{3-}_j$ (where $\alpha^{2-}_j$ is adjacent to $c^{2+}_i$). 
\item $c^{3+}_i,c^{3-}_i$ and either $\alpha^{2-}_j$ and $\alpha^{3-}_j$ or $\beta^{2-}_j$ and $\beta^{3-}_j$ (where $\alpha^{2-}_j$ is adjacent to $c^{3+}_i$). 
\end{itemize}We can recover the satisfying assignment $\xi$ as follows. For $i = 1, 2, \ldots,n$, set $\xi(i) = 1$ if the vertices of $P^+_i$ are in $H$ and set $\xi(i) = 0$ if the vertices of $P^-_i$ are in $H$. By construction, it is easy to verify that at least one literal in every clause is satisfied, so $\xi$ is indeed a satisfying assignment. 
\end{proof}

\noindent
We note that neither the gadget of L{\'{e}}v{\^{e}}que et al.~\cite{LLMT09} nor our gadget in Theorem~\ref{t-iii} has high girth.

As mentioned, we also need the following known lemma from~\cite{MPSL23}, which we rephrase more explicitly (in order to ensure $(\H_4,\ldots,\H_\ell)$-subgraph-freeness) and therefore we added its proof from~\cite{MPSL23} as well.

	\begin{figure}
\resizebox{10cm}{!}{
	\begin{tikzpicture}
	[dot/.style={circle,draw=black, fill,inner sep=1pt},]	
	\node[dot] at (0,12){$s_1$ };
	\node[dot] at (2,12){$p_1$};
	\node[dot] at (4,12){$q_1$};
	\node[dot] at (6,12){$x_1$};
	\node[dot] at (4,10) {$r_1$};
	\node[dot] at (4,8) {$s_2$};
	\node[dot] at (6,8){$r_2$};
	\node[dot] at (8,8){$q_2$};
	\node[dot] at (10,8){$x_2$};
	\node[dot] at (0,6){$p_2$};
	\draw(4,8)--(4.75,8);
	\draw(5.25,8)--(6,8);
	\draw(6,8)--(6.75,8);
	\draw(7.25,8)--(8,8);
	\draw(8,8)--(8.75,8);
	\draw(9.25,8)--(10,8);
	\node[circle, draw=black, inner sep=5pt]() at (5,8){};
	\node[circle, draw=black, inner sep=5pt]() at (7,8){};
	\node[circle, draw=black, inner sep=5pt]() at (9,8){};
	\draw(4,12)--(4,11.25);
	\draw(4,10.75)--(4,10);
	\draw(4,10)--(4,9.25);
	\draw(4,8.75)--(4,8);
	\node[circle, draw=black, inner sep=5pt] at (0,9){};
	\draw(0,6)--(0,8.75);
	\draw(0,9.25)--(0,12);
	\node at (0,12.5){$s_1$};
	\node at (2,12.5){$p_1$};
	\node at (4,12.5){$q_1$};
	\node at (6,12.5){$x_1$};
	\node at (4.5,10){$r_1$};
	\node at (4,7.5) {$s_2$};
	\node at (6,8.5) {$r_2$};
	\node at (8,8.5) {$q_2$};
	\node at (10,8.5){$x_2$};
	\node at (0,5.5) {$p_2$};
	\draw(6,12)--(7,13);
	\draw(6,12)--(7,11);
	\draw(10,8)--(11,9);
	\draw(10,8)--(11,7);
	\node[circle, draw=black, inner sep=5pt](d1) at (1,12){};
	\node[circle, draw=black, inner sep=5pt](d2) at (3,12){};
	\node[circle, draw=black, inner sep=5pt](d3) at (5,12){};
	\draw(0,12)--(0.75,12);
	\draw(1.25,12)--(2,12);
	\draw(2,12)--(2.75,12);
	\draw(3.25,12)--(4,12);
	\draw(4,12)--(4.75,12);
	\draw(5.25,12)--(6,12);
	\node[circle, draw=black, inner sep=5pt](d3) at (4,11){};
	\node[circle, draw=black, inner sep=5pt](d3) at (4,9){};
	\draw(0,6)--(3.75,7);
	\draw(4.25,7.1)--(8,8);
	\node(d)[circle, draw=black, inner sep=5pt] at (4,7){};
\end{tikzpicture}
}
\caption{The part of the graph $G^1$ that corresponds to a vertex $x$ in a graph $G$ that has exactly two neighbours~$x_1$ and $x_2$; figure taken from~\cite{MPSL23}.}\label{f-nnew}
\end{figure}
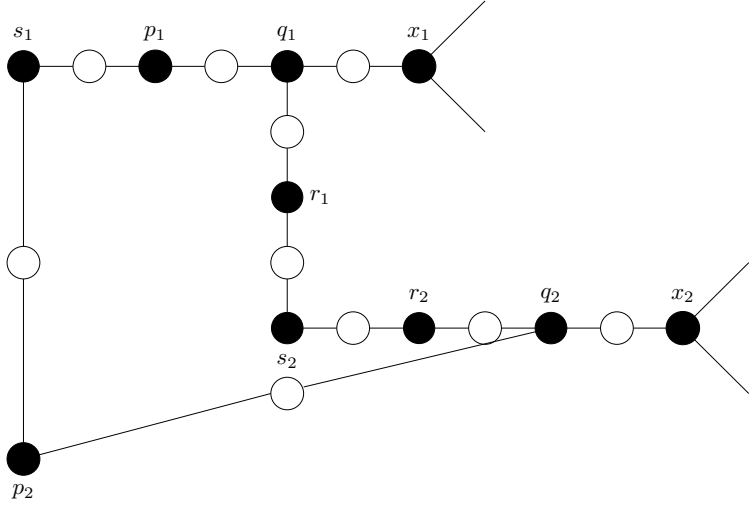

\begin{lemma}[\cite{MPSL23}]\label{l-i}
Let $\ell \geq 4$ and $G$ be a subcubic $(\H_4,\ldots,\H_\ell)$-subgraph-free graph with two non-adjacent vertices $x$ and $y$ of degree~$2$. 
It is possible to construct in polynomial time a subcubic $(\H_4,\ldots,\H_\ell)$-subgraph-free graph $G'$ that contains vertices $s_1,t_1,s_2,t_2$ such that $(G,x,y)$ is a yes-instance of $2$-{\sc Induced Cycle} if and only if  
$(G',\{(s_1,t_1),(s_2,t_2)\})$ is a yes-instance of {\sc $2$-Induced Disjoint Paths}.
\end{lemma}

\begin{proof}
Let $x$ and $y$ have neighbours $x_1$, $x_2$ and $y_1$, $y_2$ respectively. We replace $x$ and its incident edges by the following. Create vertices $p_1$, $q_1$, $r_1$, $p_2$, $q_2$, $r_2$, $s_1$, $s_2$. Add edges $s_1p_1, p_1q_1,q_1x_1, q_1r_1, r_1s_2$ and $s_2r_2, r_2q_2, q_2x_2,
p_2s_1, p_2q_2$.
In the same way, we replace $y$ and its incident edges by vertices $a_1,b_1,c_1,a_2,b_2,c_2,t_1,t_2$ and edges $a_1t_1, a_1b_1, b_1y_1, b_1c_1, c_1t_2$ and $a_2t_1, a_2b_2, b_2y_2, b_2c_2, c_2t_2$.
Note that any vertex that is introduced has degree at most~$3$. For an integer $d\geq 0$, we now subdivide every incident edge of every newly introduced vertex $d$ times. We denote the new graph by $G^d$ and say that $G^d$ is the {\it $d$-replacement} for $G$. Note that $G^d$ has maximum degree~$3$ as well. See also Fig.~\ref{f-nnew}, where we display the replacement of $x$ for $d=1$.
It is readily seen that for every integer~$d\geq 0$, $G^d$ has maximum degree ~$3$.

To prove the remainder of the lemma, first assume that $d=0$. We observe that the paths $s_1,p_1,q_1,x_1$ and $s_2,r_2,q_2,x_2$ are mutually induced. The paths $s_1,p_2,q_2,x_2$ and $s_2,r_1,q_1,x_1$ are also mutually induced. Moreover, these are the only two options that can co-exist, in the following sense.
 A path from $s_1$ to~$x_1$ that uses only edges of this gadget and that does not have $s_2$ as an internal vertex has to pass through $q_1$.
 Similarly, a path from $s_2$ to $x_2$ that uses only edges of this gadget and that does not have $s_1$ as an internal vertex has to pass through $q_2$. 
The same observations hold with respect to the gadget replacing vertex~$y$.
Hence, $G$ has a hole containing $x$ and $y$ if and only if $G^d$ has mutually induced paths between $s_1$ and $t_1$ and between $s_2$ and $t_2$.

Note that any incident edges of every vertex of $G^d$ that does not belong to $G$ can be subdivided an arbitrary number of times without affecting the correctness of the reduction ($q_1,q_2,b_1,b_2$ remain bottlenecks). So, the claim also holds for $d\geq 1$, 
and by taking $d$ sufficiently large, we obtain the desired subcubic $(\H_4,\ldots,\H_\ell)$-subgraph-free graph $G'$.
\end{proof}

\noindent
We are now ready to prove Theorem~\ref{t-khard}, which we restate below.

\medskip
\noindent
{\bf Theorem~\ref{t-khard} (restated).}
{\it For all $k\geq 2$, {\sc $k$-Induced Disjoint Paths} is \NP-complete for subcubic $(\H_4,\ldots,\H_\ell)$-subgraph-free graphs for all $\ell\geq 4$.}

\begin{proof}
The proof follows from combining Theorem~\ref{t-iii} with Lemma~\ref{l-i}.
\end{proof}

\section{The Proof of Theorem~\ref{thm:C5-char}}\label{a-c5c}

In this section we prove Theorem~\ref{thm:C5-char}, which we restate below.

\medskip
\noindent
{\bf Theorem~\ref{thm:C5-char} (restate).}
{\it The only $\H_3$-subgraph-free $C_5$-critical graphs are $K_3$, 
    odd flowers $F_n$ for $n \geq 3$, and the exceptional graphs $E_1, E_2$ and $E_3$. Equivalently,  the following three statements are true:
  \begin{enumerate}
    \item All $\H_3$-subgraph-free graphs of girth at least $6$ are $C_5$-colourable. 
    \item The only $\H_3$-subgraph-free $C_5$-critical graphs 
    of girth $5$ are $E_1$, $E_2$ and odd $C_5$-flowers $F_n$.
    \item The only $\H_3$-subgraph-free $C_5$-critical graph of
    girth $4$ is $E_3$. 
  \end{enumerate}}

\medskip
\noindent
For doing this, we will need the following notions for a graph $G$.   
 \begin{itemize}  
    \item A {\em $k$-vertex} is a vertex of degree $k$ and a $k^+$-vertex is a vertex of degree at least $k$. 
    \item A {\em $k$-thread} (respectively, $k^+$-thread) is a path consisting of $k$ (respectively, at least $k$) distinct $2$-vertices. (We often refer to a $k$-thread between endpoints not in the thread.)
    \item Two non-adjacent vertices are {\em clones} if they have the same neighbourboods. A vertex $u$ is a {\em subclone} of a nonadjacent vertex $v$ if its neighbourhood is contained in that of $v$. 
 \end{itemize}

 \subsection{$C_5$-critical graphs}
 
 We list some easy observations about $C_5$-critical graphs.  
\begin{fact}\label{fact:critical}
  The following hold for all $C_5$-critical graphs $G$.
  \begin{enumerate}
            \item $G$ cannot properly contain a $C_5$-critical graph. In particular, it cannot contain a $K_3$. 
            \item $G$ is $2$-connected. 
            \item $G$ has no {\em removable threads}:  $3^+$-threads. 
            \item $G$ has no {\em redundant threads}: $k$-threads between end vertices that are the end vertices of a 
                  $k$-path that is disjoint from the $k$-thread.   
            \item $G$ has no {\em redundant vertices}: clones or subclones of another vertex.  
            \item Any vertex in a $C_4$ is a $3^+$-vertex. 
  \end{enumerate}
\end{fact}
\begin{proof}
   That the first two items hold are clear.  To see that $C_5$-critical $G$ cannot have a removable thread, remove the thread and $C_5$-colour the remaining graph $G'$. As the thread has length at least $k$, we can extend the $C_5$-colouring of $G'$ to $G$; contradicting the fact that $G$ was $C_5$-critical.  To see that $G$ cannot have a redundant thread, remove it, and get a $C_5$-colouring of the remaining graph $G'$.  We can extend this to a $C_5$-colouring of $G$ be letting the colouring on the redundant thread be the same as on the path of $G'$ of the same length between the endpoints of the thread. That redundant vertices cannot occur in a critical graph is clear-- they can always get the same colour as the vertex they are a subclone of, so removing them does not change the possible colourings. That any vertex in a $C_4$ is a $3^+$ vertex is because a $2$-vertex in a $C_4$ is redundant.
   That 
\end{proof}
  
  As $G$ is $2$-connected, we will often apply the following well known version of the classical Ear Decomposition result of Whitney. 
  Recall that an ear of $G'$ in $G$ is a  path in $G$, possibly a single edge, whose intersection with $G'$ is its endpoints. 

  \begin{lemma}\label{lem:Ear}
    Any $2$-connected graph $G$ that is not just a cycle can be constructed from any $2$-connected subgraph by adding successive 
    ears. 
  \end{lemma}
   
   We cannot find a direct reference to this version of the Ear Decomposition theorem. The standard version says that $G$ can be constructed from a cycle by adding successive ears. The standard proof though, such as is found in \cite{GraphTheory}, can be used to get the version we have stated.  

   \subsection{$\H_3$-subgraph-free graphs} 
   
   Many of our arguments will require finding a copy of $\H_3$ in a graph that we have constructed.  We will denote a copy of $\H_3$ in a graph as 
         \[ \{t,t'\},u,v,w,x,\{y,y'\} \]
   where the central path of the $\H_3$ is the path $u,v,w,x$, the leaves adjacent to $u$ are $t$ and $t'$ and those adjacent to $x$ are $y$ and $y'$.   

   The following basic fact can be readily seen.
   
  \begin{fact}\label{fact:H3}
   In an $\H_3$-subgraph-free graph, two $3^+$-vertices with a $3$-path between them are at distance at most $2$. They are either adjacent in a $C_4$ or at distance $2$ in a $C_5$.
   \end{fact}

   \subsection{$\H_3$-subgraph-free $C_5$-critical graphs}

We first prove the following lemma. 

  \begin{lemma}\label{lem:c4}
    The following hold for any $\H_3$-subgraph-free $C_5$-critical graph $G$. 
   \begin{enumerate}
      \item Any $3^+$-vertex at distance $1$ in a $C_4$ to a $4^+$-vertex is a $3$-vertex.  
     \item Any $4^+$-vertex at distance $2$ in a $C_5$ to $3^+$-vertex is a $4$-vertex.   
  \end{enumerate}
  \end{lemma}
  \begin{proof}
    Let $G$ be $\H_3$-subgraph-free and $C_5$-critical.  If it is a $K_3$ the lemma is vacuously true, so we may assume by criticality that it is $K_3$-free.  

    For item (1), assume that $u$ is a $3^+$ vertex in a $C_4$ with its $4^+$-neighbour $x$, in $G$.   As $G$ is $K_3$-free, we may assume that the $C_4$ is $u,v,w,x$, and is induced.  As $x$ is a $4^+$ vertex it has neighbours $y$ and $y'$ not in the $C_4$. If $u$ is a $4^+$-vertex then it has neighbours $t$ and $t'$ not in the $C_4$, and as $G$ is $K_3$-free, they are distinct from $y$ and $y'$. But then $\{t,t'\},u,v,w,x,\{y,y'\}$ is an $\H_3$, which is impossible.  

    For item (2), assume that $C$ is the $5$-cycle $a,b,c,d,e,a$, that $a$ is a $3^+$ vertex, and that $d$ is a $4^+$ vertex.  As $G$ is $K_3$-free, $C$ is induced, and so $a$ has neighbours $x,x'$  not in $\{b,c,d\}$. If $d$ has degree $5$, then it has two neighbours $y,y'$ not in $\{a,b,c,x,x'\}$ and so 
    $\{x,x'\},a,b,c,d,\{y,y'\}$ is an $\H_3$.
  \end{proof}

The notion of saturation will be used extensively in our proofs.
A vertex $v$ of a graph $G$ is {\em saturated in a subgraph $G' \leq G$} if all of its edges in $G$ are in $G'$. A proper subgraph $G'$ of a graph $G$ is {\em saturated} if it has at least two vertices, and all but at most one of its vertices are saturated.   

\begin{fact}\label{fact:saturated}
  A $C_5$-critical graph $G$ contains no saturated proper subgraph.
\end{fact}
\begin{proof}
    If $C_5$-critical $G$ contained a proper subgraph $G'$ that was saturated, then that subgraph is disconnected from $G$ by the removal of its one unsaturated vertex, but this contradicts the fact that $G$, being $C_5$-critical, is $2$-connected by Fact \ref{fact:critical}.
\end{proof}

\subsection{Step 1}

We are now ready to prove statement 1 of Theorem~\ref{thm:C5-char}.

\begin{proposition}\label{prop:C5-char-1}
  All $\H_3$-subgraph-free graphs of girth at least $6$ are $C_5$-colourable. 	
\end{proposition}
\begin{proof}
   Towards contradiction, let $G$ be a $\H_3$-subgraph-free graph of girth at least $6$ that is $C_5$-critical. As $G$ has no $C_5$-colouring it must contain an odd cycle, and so by assumption, it has odd girth $2k+1 \geq 7$.  Let $C$ be a $C_{2k+1}$ in $G$ with vertex set $\{-k,-(k-1), \dots, 0, \dots, k\}$ for which two vertices are adjacent if they differ by $1$ modulo $2k+1$.  Let $X$ be the set of $3^+$-vertices of $C$.

   As $C$ is a shortest odd cycle in $G$ it is induced, and there is no vertex $x \not\in C$ adjacent to two vertices of $C$ at distance greater than $2$. So because $G$ is $\H_3$-subgraph-free, $X$ contains no two vertices that are distance $3$ apart.  On the other hand, as $G$ is $C_5$-critical and so contains no $3$-threads $X$ contains at least one of $i-1, i,$ and $i+1$ for every vertex $i$ of $C$.  It follows that if $X$ contains $i$ and $i+1$, then it does not contain $i + 3$ and $i+4$, and so it must contain $i+2$ and $i+5$. But these are distance $3$ apart; so we conclude that $X$ contains no two consecutive vertices. 

   We may assume, w.l.o.g, that $X$ contains the vertex $0$. So it contains neither $-1$ or $1$, or $-3$ or $3$.  
   It must therefore contain $-2$ and $2$. Shifting this argument cyclically by $2$, we see that $X$ must contain every second vertex of $C$.  But this is impossible as $C$ is an odd cycle. 
\end{proof}

\subsection{Step 2}

We now prove statement 2 of Theorem~\ref{thm:C5-char}.

\begin{proposition}\label{prop:C5-char-2}The only $\H_3$-subgraph-free graphs of girth $5$ are $E_1, E_2$ and $F_n$ for odd $n \geq 3$. 
\end{proposition}

 Before we get started with the proof, we observe, among other things, some strengthenings we can get of Lemma \ref{lem:c4} now that we have girth $5$.

 \begin{lemma} \label{lem:c5}   
   The following hold for all $\H_3$-subgraph-free $C_5$-critical graphs $G$ of girth at least $5$. 
  \begin{enumerate}
      \item Any $3^+$-vertices at distance $2$ in a $C_5$ are $3$-vertices. 
      \item Any $3^+$-vertices with a $3$-path between them are $3$-vertices at distance $2$ in a $C_5$.  
      \item Any $5$-cycle in $G$ has at least three $3^+$-vertices; some pair of them are adjacent. 
  \end{enumerate}
  \end{lemma}
  \begin{proof}
   Item $(1)$ is straightforward from the fact that $G$ has no induced $\H_3$ or $C_4$, and $(2)$ comes from Fact \ref{fact:H3} and item $(1)$.
    
   We show item (3). Any $5$-cycle $C$ in $G$ is induced and has no $3^+$-threads, so (by an easy version of the argument in the proof of Proposition \ref{prop:C5-char-1}) must contain two $3^+$-vertices, call them $1$ and $4$, at distance $2$ in $C$. By $(1)$, these are, in fact, $3$-vertices. Let $a$ and $b$, be their respective neighbours not in $C$. 
   
 If there is no other $3^+$-vertex in $C$ not, then remove $C$, and $C_5$-colour $G' = G \setminus C$. The only edges from $C$ to $G'$ are the edges $1a$ and $4b$, and whatever colours $a$ and $b$ get, we can extend this colouring to a $C_5$-colouring of $G$, contradicting the fact that $G$ was $C_5$-critical. 
  Thus $C$ has at least three $3^+$-vertices, as needed. 
  \end{proof} 

    \begin{figure}
    \begin{center}
    \begin{tikzpicture}[every node/.style={vert}]
     
      \begin{scope}[xshift = 2cm,yshift = 1cm]
      \begin{scope}
       \Verts[:]{v1/36/1, v2/108/1, v3/180/1, v4/-108/1, v5/-36/1}
       \Edges{v1/v2,v2/v3,v3/v4,v4/v5,v5/v1} 
       \draw (1,-1.5) node[empty]{$S_1$};
       
      \end{scope}
      \begin{scope}[xshift = 1.65cm]
        \Verts[:]{u3/0/1,u2/72/1,u4/-72/1, m/-90/0} 
       \end{scope} 
       \Edges{v1/u2, u2/u3, u3/u4, u4/v5, v3/m, m/u3}
       \end{scope}
       
    \end{tikzpicture}
    \caption{Saturated $\H_3$-free (but $C_5$-colourable) graph $S_1$}\label{fig:h3counters2} 
    \end{center}
    \end{figure}
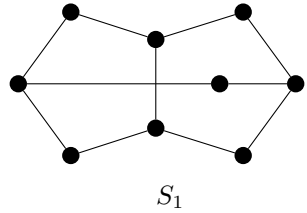
 We highlight some saturated graphs. 
   \begin{fact}\label{fact:saturated-2}
        As subgraphs of an $\H_3$-subgraph-free graph $G$ of girth $5$,  graph $S_1$ from Figure \ref{fig:h3counters2}, and all even $C_5$-flowers, $F_n$ for $n \geq 4$, are saturated and $C_5$-colourable.  So they cannot be subgraphs of $G$ if it is $C_5$-critical.  
   \end{fact}
   \begin{proof}
  To see that $S_1$ is saturated, note that all $3$-vertices have $3$-paths to other $3$-vertices, so are saturated.
  For all $2$-vertices but the one in the middle of the right cycle, there is a $3$-path to a $3$-vertex not in a $C_5$ with it. Adding an edge to such a $2$-vertex, the edge would then have to go to a neighbour of the $3$-vertex, but checking each case, this would make a $C_3$ or $C_4$.  The middle $2$-vertex is not saturated, but one non-saturated vertex is okay. 

  In an even $C_5$-flower, the $3$-vertices all have $3$-paths to other $3$-vertices, so are saturated. The $2$-vertices all have $3$-paths to the middle $3^+$-vertex, so if they have another edge, it must be to one of its neighbours. If  this neighbour is a new vertex, we have a $C_4$, and if it is a neighbour in $F_n$, then it is a $3$-vertex, which is already saturated. So the $2$-vertices are saturated. Only the middle vertex is not.  But again, one unsaturated vertex is allowed.  
   \end{proof}

  With these tools, we are now ready to prove the Proposition.

  \begin{proof}[Proof of Proposition \ref{prop:C5-char-2}] 
  Let $G$ be an $\H_3$-subgraph-free $C_5$-critical graph of girth $5$. We show that $G$ is either $E_1$, $E_2$ or an odd $C_5$-flower $F_n$.

  We start with two claims that deal with the more difficult cases.
  Let $B_1$ and $B_2$ be the graphs made from two copies of $C_5$ shown in Figure \ref{fig:BB}. 

  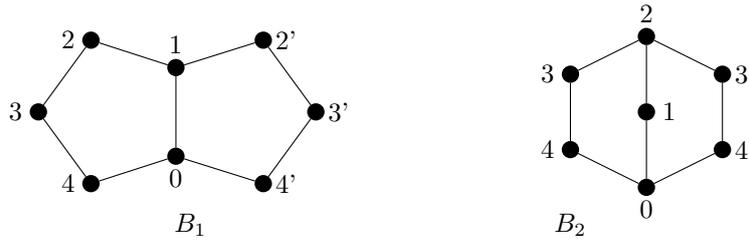
\begin{figure}
   \begin{center}
   \begin{tikzpicture}[every node/.style={vert}]
      \begin{scope}
      \begin{scope}
       \Verts[:]{v1/36/1, v2/108/1, v3/180/1, v4/-108/1, v5/-36/1}
       \Edges{v1/v2,v2/v3,v3/v4,v4/v5,v5/v1} 
       \draw (1,-1.5) node[empty]{$B_1$};
      \end{scope}
      \begin{scope}[xshift = 1.65cm]
        \Verts[:]{u3/0/1,u2/72/1,u4/-72/1} 
      \end{scope} 
      \Edges{v1/u2, u2/u3, u3/u4, u4/v5}
      \Vlabels{v1/above/$1$,v2/left/$2$,v3/left/$3$,v4/left/$4$,v5/below/$0$,u2/right/$2'$,u3/right/$3'$,u4/right/$4'$}
      \end{scope}
      
     \begin{scope}[xshift = 7cm]
       \Verts{v0/0/-1, v1/0/0, v2/0/1, v3/-1/.5, v4/-1/-.5, u3/1/.5, u4/1/-.5}
       \Edges{v0/v1,v1/v2,v2/v3,v3/v4,v4/v0,v2/u3, u3/u4, u4/v0} 
\Vlabels{v0/below/$0$, v1/right/$1$, v2/above/$2$, v3/left/$3$, v4/left/$4$, u3/right/$3'$, u4/right/$4'$}
       \draw (-1,-1.5) node[empty]{$B_2$};
      \end{scope}
   
   \end{tikzpicture}
   \caption{The graphs $B_1$ and $B_2$}\label{fig:BB} 
   \end{center}
   \end{figure}

  \begin{claim}
      If $G$ contains $B_1$ as a subgraph, then $G$ is an odd $C_5$-flower or one of the exceptional graphs $E_i$. 
  \end{claim}
  \begin{proof}
    Assume that $G$ contains a copy $B$ of $B_1$ labelled with labels $0,1,2,3,4$ and $2',3',4'$ as in Figure \ref{fig:BB} so that ignoring the primes we get a $C_5$-colouring. 
    As $G$ has no $3$-threads, at least one of each of $\{2,3,4\}$ and $\{2',3',4'\}$ must be a $3^+$-vertex. Whichever is, it is distance $2$ to one of the $3$-vertices $0$ or $1$, so is also a $3$-vertex by Lemma \ref{lem:c5}. If these two $3$-vertices have a $3$-path between them, then, again by Lemma \ref{lem:c5}(2), they must share a new neighbour not in the graph yet, so if one is $3$ and the other is $2'$, the new vertex makes a copy of $E_1$ and we are done. 
    With symmetric arguments, we are done if $3$ and $2'$ or $4'$ are three vertices, or if $3'$ and $2$ or $4$ are.
    Up to symmetry, we have three cases left to consider: those where the pairs $(2,4')$ or $(4,4')$ or $(3,3')$ are pairs of $3$-vertices.  \\
 
    {\bf Case: $3$ and $3'$ are $3$-vertices.}  In this case, we may assume that the $2$-vertices $2,2',4$ and $4'$ are all saturated in $B$, as if any one of them is a $3$-vertex, we appeal to an earlier case. Further, the vertices $0$ and $1$, $3$, and $3'$ are all $3$-vertices by Lemma \ref{lem:c5} because they are $3^+$-vertices with $3$-paths to other $3$-vertices.  So only $3$ and $3'$ are unsaturated in $B$, each having an edge out of $B$.  As $G$ is $2$-connected, there is an ear from $3$ to $3'$. Let $P$ be the shortest such ear.   

If the ear $P$ has length $1,2,$ or $3$ we have $E_2, S_1,$ or an induced $\H_3$.  In the first case we are done, and in the other two would yield contradictions. So 
we may assume that $P$ has length at least $4$. Let $a$ and $a'$ be the vertices in $P$ ear adjacent to $3$ and $3'$ respectively.  They cannot be $3^+$ vertices, as they both have $3$-paths to the $3$-vertex $1$, all of whose neighbours are saturated. So $a$ and $a'$ are $2$-vertices. Now, $B$ has no edges to $G \setminus B$ except $3a$ and $3'a'$, and the vertices $3$ and $3'$ get the same colouring under any $C_5$-colouring, so the graph $G'$ we get from $G$ by replacing all of $B$ with one new vertex $b$ adjacent just to $a$ and $a'$ has exactly the same $C_5$-colourings as $G$ does when restricted to $G \setminus B$.  In particular, if $G$ is $C_5$-critical then $G'$ is.  But $G'$ is not critical, as the $3$-thread $a,b,a'$ is removable.  So $G$ is not $C_5$-critical, a contradiction. \\

    {\bf Case: $2$ and $4'$ are $3$-vertices.}  The vertices $2$ and $4'$ have a $3$-path between them so by Lemma \ref{lem:c5} share a new neighbour $3''$, and with this become saturated. The $2$-vertices $3$ and $3'$ are saturated or we can apply a previous case.  The vertices $0,1,2$ and $4'$ are $3$-vertices with $3$-paths to $3$-vertices, so are saturated. 

 The only three unsaturated vertices are those in $X:= \{2',3'',4\}$. They are pairwise distance $3$ apart, so by girth, cannot have edges between them, and for two any of them that are $3$-vertices, the third edges go to a common new third neighbour $x$. If all vertices of $X$ are $3^+$-vertices, then we get a copy of $E_2$ with $5$-cycles $x,3'',4',0,4$ and $x, 3'', 2, 1, 2'$. If only one vertex in $X$ is a $3$-vertex, then $G$ is not $2$-connected, a contradiction. So we may assume that exactly two of them are $3$-vertices. They are saturated in the graph $G'$ made so far on $B \cup \{3'',x\}$; as is the other vertex of $X$, because it is assumed not to be a $3^+$-vertex.  So the only possibly unsaturated vertex of $G'$ is $x$, making $G'$ saturated. But $G'$ is $C_5$-colourable, so is a proper subgraph of $G$, contradicting Fact \ref{fact:saturated}.   \\

    {\bf Case: $4$ and $4'$ are $3$-vertices.} 
  The neighbours $a$ of $4$, and $a'$ of $4'$ are new and distinct by girth considerations. The $3$-vertices $1,4$ and $4'$ have $3$-paths to other $3$-vertices so are saturated in the graph build on $B \cup \{a,a'\}$ by adding these edges.  
    
     The vertex $2'$ has a $3$-path to $4$ so if it has a third edge, this edge is to one of the neighbours $0,3$ or $a$ of $4$.  In the first two cases we violate girth, and in the last case $2$ and $4'$ are $3$-vertices, which is dealt with above. So the $2$-vertices $2$ and $2'$ are saturated $2$-vertices.

       The vertex $3'$ has a $3$-path to $4$ so if it has a third edge, this is  to one of $0,3$ or $a$. The first case violates girth, the second yields an $E_2$, and the third case is already dealt with above. So $3'$ and $3$ are saturated $2$-vertices. 

     The only vertex of $B$ that may have another edge is $0$.  So $a$ and $a'$ have no other neighbours in $B$.

    We view the $B$ as the start of a $C_5$-flower with the vertex $0$ as its centre, in which only $0$ is not saturated.  If we extend this to a partial $C_5$-flower -- a $C_5$-flower from which we have removed the two adjacent $2$-vertices of some petal -- the argument for $C_5$-flowers in Fact \ref{fact:saturated-2} works to argue that the only unsaturated vertices are $0$ and the corner vertices $c$ and $c'$: those two $3$-vertices who have been made $2$-vertices by the removal of the two vertices. 

    Let $G'$ be the largest (partial) $C_5$-flower with centre $0$ that contains  our original $B$.  If $G'$ is a full even $C_5$-flower, we have a saturated subgraph of $G$; this is a contradiction. If $G'$ is a full odd $C_5$-flower, then it is critical, so is $G$, as needed.

    So we may assume that $G'$ is a partial $C_5$-flower. As $G'$ is $C_5$-colourable, there is an ear of $G'$ in $G$; and as the only three non-saturated vertices of $G'$ are $0,c$ and $c'$, we may assume the ear begins at $c$. It must have length at least $3$ by girth, so assume it starts with the path $c,x,y,z$, where $x$ and $y$ are necessarily new vertices. If $z = 0$ then we have a $C_4$, and if $z = c'$ then this is a bigger $C_5$-flower, contradicting the choice of $G'$, so $z$ is also a vertex not in $G'$.   The vertex $x$ has a $3$-path to the $3$-vertex $1$, and the only non-saturated neighbour of $1$ is $0$, but $x$ is not adjacent to $0$ by girth, so $x$ is saturated.  As $y$ has a $3$-path to $0$ it must be a $2$-vertex, or it would have to share a neighbour, which we may assume is $z$, with $0$. But this makes a larger partial $C_5$-flower, contradicting the choice of $G'$.  Thus $y$ is a saturated $2$-vertex, and $z$ is not adjacent to $0$.
 Having a $3$-path to the $3$-vertex $c$, $z$ is a $2$-vertex, and so we have a removable thread, contradicting the assumption that $G$ is 
critical. 
  \end{proof}

  Recall that $B_2$, shown in Figure \ref{fig:BB}, consists two $C_5$s whose intersection is a path of length two.  

   \begin{claim}\label{cl:B2}
       $G$ contains no copy of $B_2$. 
   \end{claim}
   \begin{proof}
      Towards contradiction, assume that $G$ contains a copy $B$ of $B_2$. Label the vertices of $B$, as in the figure, so that the two $5$-cycles are $4,0,1,2,3$ and $4,0,1,2',3'$ The $3$-vertices $4$ and $1$ are saturated.  At least one vertex in each of $\{2,3\}$ and $\{2',3'\}$ must also be a $3^+$-vertex, or we have a redundant path. 
       
    If $2$ and $3'$ are $3^+$-vertices, then being distance $3$ apart, they share a new common neighbour, and this makes a copy of $B$, so we are done by the previous lemma. The same holds if $2'$ and $3$ are $3^+$-vertices.  
    We may thus assume that $3$ and $3'$ are $3^+$-vertices, and that the
    $2$-vertices $2$ and $2'$ are saturated.  As $3$ and $3'$ are $3^+$ vertices with $3$-paths to $1$, they are in fact $3$-vertices. 
 The non-saturated vertices of $B$ are $0,3$ and $3'$. The third neighbours  of $3$ and $3'$ can neither be $0$ or a common new vertex by girth, so are distinct new vertices $a$ and $a'$ respectively. Adding the edges from $3$ and $3'$ to $a$ and $a'$ to our subgraph of $G$, the vertices $3$ and $3'$ become saturated. 
 
 If $a$ and $a'$ are adjacent we have a copy of $B_1$, and they cannot have edges to $0$ by girth. As all other vertices are saturated, they have respective neighbours $b$ and $b'$ not in $B$, which may be the same vertex.
 Adding these to our subgraph, the $2$-vertices $a$ and $a'$ are saturated, as they have $3$-paths to the $3$-vertex $1$ whose only unsaturated neighbour is $0$, and they cannot have edges to $0$, as it would contradict the girth.  
 
If $b$ has an edge to $0$, then we have a $B_1$ with $5$-cycles $0,4,3,a,b$ and $0,4,3',2',1$, so we may assume that $b$ has no edge to $0$. Similarly we assume that $b'$ has no edge to $0$. 
 As $b$ and $b'$ are distance $3$ from $4$, and the only unsaturated neighbour of $4$ is $0$, we get that $b$ and $b'$ are $2$-vertices.

 If $b = b'$ or $bb'$ is an edge we have a removable thread, so we may assume that $b$ has new neighbour $c$. But $c$ has a $3$-path to the $3$-vertex $3$ all of whose neighbours are saturated. So $c$ is also a $2$-vertex and we have a removable thread. 
   \end{proof}

  Now, we have assumed that $G$ has girth $5$. Let $C$ be a $5$-cycle in $G$ with vertices $0,1,2,3,4$. By part (4) of Fact \ref{fact:saturated}, there are at least three $3^+$-vertices in $C$, so we may assume that $1$ and $4$ are $3^+$-vertices. By girth, they have distinct neighbours, $a$ and $z$ respectively, that are not in $C$.  Further, $a$ and $z$ are non-adjacent as an edge between them would make a $B_2$, which contradicts Claim 
  \ref{cl:B2}.
  
  By girth, neither of $a$ or $z$ have a second edge to $C$, and they both have $3$-paths to $3$-vertices in $C$, so they are both $2$-vertices.  They have new neighbours $b$ and $y$ respectively, not necessarily distinct.

  As $C$ has at least one more $3$-vertex, we may assume that $b$ has a $3$-path to this vertex $i$. If $b$ is a $3$-vertex, then it must have an edge to a neighbour of $i$, other than $1$ and $4$. Whatever this edge is, would make a $B_1$ or a $B_2$, so we would be done by one of the claims.  Thus $b$ is a $2$-vertex. 

  If $b=y$ we have a removable thread, so let $c$ be the new neighbour of $b$. 
       As $c$ has a $3$-path to $1$, if it is a $3$-vertex then it has an edge to $0$ or $2$, making a $B_1$. 
       So we may assume that $c$ is a $2$ vertex, and so we have a removable thread.  Thus $G$ is not critical, which is a contradiction. 
\end{proof}

 \subsection{Step 3} 

Finally, we prove statement 3 of Theorem~\ref{thm:C5-char}.
  
 \begin{proposition}\label{prop:C5-char-3}
    If $G$ is a $C_5$-critical $\H_3$-subgraph-free graph of girth $4$ then it is $E_3$ from Figure \ref{fig:h3counters}.   
  \end{proposition}  

  Before we begin the actual proof, we prove two lemma. 

  \begin{lemma}\label{lem:c4b}
       In an $\H_{3}$-subgraph-free $C_5$-critical graph $G$, every $3^+$-vertex $v$ has an edge to every $C_4$ in $G$.   Moreover,  $v$ is in a $P_2$ or $P_3$ between non-adjacent vertices of the $C_4$.  
  \end{lemma}
   \begin{proof}
    
     Fix a $4$-cycle $C$ in $G$ on the vertices $0,1,2,3$. 
     There can be no $3^+$ vertex at distance exactly $3$ from $C$ as it would make an $H_3$ with the vertex in $C$ closest to it.

     \begin{claim*}
         There can be no $3^+$-vertex at distance exactly $2$ from $C$. 
     \end{claim*} 
     \begin{proof}
Assume that $v$ is a $3^+$ vertex at distance $2$ from $C$. We may assume that there is a path $v \sim 0' \sim 0$ in $G$. 
As there is a $3$-path from $v$ to the $3^+$-vertex $1$, we have by Fact \ref{fact:H3} that $v$ an edge to, or shares a neighbour with, $1$. By the assumption that $v$ has distance $2$ to $C$, it must share a neighbour with $1$, and this must be a vertex not in $C$ or we have a $K_3$. The same argument works to show that $v$ shares neighbours not in $C$ with $2$ and $3$. Thus $G$ contains the following graph where all vertices are distinct except that $0' = 2'$ and/or $1' = 3'$.  (The graph has a $\H_3$ so we must make at least one of the identifications.)

 \begin{center}
  \begin{tikzpicture}[every node/.style={vert}]
      \begin{scope}[xshift = -5cm]
       \foreach \i/\x/\y in {W/-2/0, v/0/0, E/2/0, N/0/2, S/0/-2}{\draw (\x,\y) node (\i){};}
        \foreach \i/\x/\y/\col in {w/-1/0/red, e/1/0/red, n/0/1/blue, s/0/-1/blue}{\draw (\x,\y) node[fill=\col] (\i){};}
        \Edges{W/N, N/E, E/S, S/W, W/w, w/v, N/n, n/v, E/e, e/v, S/s, s/v}
        \Vlabels{v/below right/v, N/right/0, E/right/1, S/right/2, W/left/3, n/right/0', e/below/1', s/right/2', w/below/3'}  
      \end{scope}     
  \end{tikzpicture} 
  \end{center}
  
 Notice that the possibly identified vertices get the same colour under the essentially unique $C_5$-colouring of this graph, so whether we make one or both of these identifications, we get a $C_5$-colourable graph.  Making one of the identifications, say $0' = 2'$, the vertices $0$ and $2$ become clones, so must each have a new edge to distinct vertices $0''$ and $2''$. These must be new vertices, or we make a $C_3$.  We thus have an $\H_3$ on $\{3,0''\},0,1,1',v,\{2',3'\}$.  

Making both identifications the vertices $0$ and $2$ are clones, as are the vertices $1$ and $3$. Arguing as above they each have new distinct neighbours, so we have a $C_4$ of $4^+$-vertices, giving an $\H_3$.  
\end{proof}

 Now, there are no $3^+$-vertices at distance $2$ or $3$ from $C$. Assuming there are $3^+$-vertices at distance greater than $3$ from $C$, and let $v$ be a closest such vertex to $C$.  If $v$ has distance $5$ or more from $C$, then since the closest $3^+$-vertex on its shortest path to $C$ is at distance $1$ from $C$, it contains a removable $3^+$-thread. So $v$ must have distance $4$ from $C$. The vertex $1'$ at distance $1$ from $C$ on a shortest path to $C$ must then be a $2$-vertex, or we have a $H_3$-- $v$ and $1'$ would be $3^+$ vertices at distance $3$ and from the fact that $v$ has distance $4$ to $C$ they cannot be adjacent or share a neighbour.  So there is a removable $3$-thread between $v$ and $C$, a contradiction.   

 For the 'Moreover' statement, assume that a $3^+$-vertex $v$ with an edge to $C$, has an edge to $0$. As there is a $3$-path to $2$, then $v$ and $2$ are either adjacent or share a neighbour, and as $G$ is $C_3$-free the shared neighbour would have to be a new vertex. 
 \end{proof}
       
 There is one more argument that we use several times, so we set it up as a lemma.

  \begin{lemma}\label{lem:BB'}
    Let $G$ be a $C_5$-critical graph. Let $G'$ be a $C_5$-colourable subgraph on the vertex set $B \cup B'$ where:
     \begin{enumerate}
          \item  $G'|_B$ is $2$-connected and contains a $C_4$, 
          \item  any $b \in B$ is saturated in $G'$, and
          \item  vertices in $B'$ are $2$-vertices of $G$ that have degree $1$ in $G'$.
     \end{enumerate}
     We can build $G$ from $G'$ by adding edges between vertices of $B'$.  
  \end{lemma}
  
  \begin{proof}
    Assume the setup of the lemma. As $G$ is $C_5$-critical it is $2$-connected; as the subgraph $B$ is $2$-connected, Lemma \ref{lem:Ear} says that we can construct $G$ from $B$ by adding successive ears.  Let $P$ be a shortest ear of $B$ in $G$;  say it is between vertices $u$ and $v$ of $B$.  As vertices in $B$ are saturated, $P$ must start and end with edges $u,u'$ and $v',v$ to $2$-vertices $u'$ and $v'$ in $B'$. As vertices of $B'$ have degree $1$ in $G'$, $u'$ and $v'$ are distinct.  If there are any $3^+$-vertices on $P$, then they must be adjacent to the $C_4$ in $B$; this contradicts the choice of $P$ as the shortest ear.   If there are any other $2$-vertices in $P$, then we have a $3$-thread, contradicting the criticality of $G$. So $P$ is a $3$-path through $B'$. Thus one gets $P$ by adding an edge between vertices of $B'$, as needed.    
 
 When we apply this lemma to a graph to extend it, it adds an edge between vertices of $B'$. If the graph remains $C_5$-colourable, then the resulting graph still satisfies the conditions of the lemma (with two vertices of $B'$ moved into $B$), and we can apply it again.  If the graph becomes non-$C_5$-colourable, then it is $G$.  
  \end{proof}

We are now ready to prove the Proposition.

\begin{proof}[Proof of Proposition \ref{prop:C5-char-3}]

 Let $G$ be a $\H_{3}$-subgraph-free $C_5$-critical graph of girth $4$.  
 We start with a pair of claims showing that $G$ can contain no biclique other than a $C_4$.   

 \begin{claim}\label{cl:k33}
   $G$ cannot contain a $K_{3,3}$. 
 \end{claim}
 \begin{proof}
     Assume, towards contradiction, that $G$ contains a $K_{3,3}$ with partite sets $\{x,y,z\}$ and $\{1,2,3\}$.  As $G$, being critical,  contains no clones, each vertex has at least one more neighbour, and so we get two adjacent $4^+$-vertices in a $C_4$, which gives an $\H_3$ by Lemma \ref{lem:c4}.   
 \end{proof}

 \begin{claim}\label{cl:k23}
     $G$ cannot contain a $K_{2,3}$. 
 \end{claim}
\begin{proof}
    Assuming that $G$ contains a $K_{2,3}$ let $B = \{x,y\} \cup [n]$ be a maximum biclique in $G$, where $[n] = \{1,2, \dots, n\}$ for some $n \geq 3$. 
 Let $X'$ and $Y'$ be the sets of neighbours of $x$ any $y$ respectively that are not in $B$.  These sets are disjoint by the choice of $B$, and non-empty or $x$ or $y$ would be redundant. To be non-redundant, each vertex $i \in [n]$ must also have a neighbour $i'$ that is not in $B$; this makes it a $3^+$-vertex adjacent, in a $C_4$,  to the $4^+$-vertex $x$, so has degree exactly $3$ by Lemma \ref{lem:c4}. 
\begin{center}

    \begin{tikzpicture}[every node/.style={labvert}]
      \begin{scope}[xshift=0]
       \LabVerts{x/.5/2, y/2.5/2, 1/0/1, 2/1/1, 3/2/1, n/4/1}
       \LabVerts{1'/0/0, 2'/1/0, 3'/2/0, n'/4/0}
       \Edges{x/1, x/2, x/3, x/n, y/1, y/2, y/3, y/n,  1/1', 2/2', 3/3', n/n'}
       \draw (3,.5) node[empty]{$\dots$};
       \draw (x) + (0,1) node[ellipse,minimum width=40pt, minimum height=20pt] (X') {$X'$};
       \draw (X'.east) edge (x); \draw (X'.west) edge (x);
       \draw (y) + (0,1) node[ellipse,minimum width=40pt, minimum height=20pt] (Y') {$Y'$};
       \draw (Y'.east) edge (y); \draw (Y'.west) edge (y);
      \end{scope}
\end{tikzpicture}
\end{center}

 We have seen that all vertices in $X'$ and $Y'$ are distinct, the are distinct from vertices of $[n]'$ as $G$ has girth $4$. So all vertices of $B':= [n]' \cup X' \cup Y'$
 have degree $1$ to $B$. 

\begin{subclaim*}
 All of the vertices in  $B'$ are $2$-vertices of $G$. 
\end{subclaim*}
\begin{proof}
  All vertices of $B'$ must have degree at least $2$ as $G$ is $2$-connected.  If $x' \in X'$ has degree at least $3$,  then we get an $\H_3$ on $\{z,z'\},x',x,1,y,\{2,3\}$ with its two new neighbours $z$ and $z'$.  So vertices of $X' \cup Y'$ have degree exactly $2$.    The vertices in $[n]'$ have degree at least $2$, so we must show they are distinct and do not have degree $3$ or more.  

  We continue the proof in two cases. First assume that  $n = 3$.  If the vertices in $[3]'$ are all the same, then we have a $K_{3,3}$ contradicting Claim \ref{cl:k33}. If two are the same, say $2' = 3'$ is distinct from $1'$, then $2$ and $3$ are clones, so must each have another edge. This makes an  $\H_3$ with central path $2,y,1,x$. So the three vertices $1',2'$ and $3'$ are distinct.    
If $2'$, say, has degree $3$, then it has an edge to $3$ (or $1$) by Lemma \ref{lem:c4b}, and ignoring the vertex $3'$ we are back in the case that $2' = 3'$ is distinct from $1'$, and are done.  So all of $1',2'$ and $3'$ have degree exactly $2$.  

We may therefore assume that $n \geq 4$.   If $1' \in [n]'$ has degree $3$, then by Lemma \ref{lem:c4b} it is adjacent to something in each of the $C_4$s in $B$ that do not contain $1$.  It cannot be adjacent to $x$ or $y$ as this would make a $C_3$, and so it is adjacent to at least two of $\{2, \dots, n\}$.  But then we have a $K_{3,3}$. We conclude, by symmetry, that everything in $[n]'$ has degree exactly $2$.  If any two vertices in $[n]'$ were the same, then they would be in a $C_4$, contradicting the fact that they have degree $2$.  Thus they, and so all vertices of $B'$ have degree $2$.  
  \end{proof}

Now, as $G, B$ and $B'$ satisfy the setup of Lemma \ref{lem:BB'} we can get $G$ by adding edges between vertices of $B'$. If we add $x'y'$ the the graph remains $C_5$-colourable so we apply the Lemma to get an edge between vertices of $[n]'$--this makes a copy of $E_3$ as a proper subgraph of $G$, which contradicts the criticality of $G$. So $x'y'$ is not an edge.    

Adding any matching to $B'$ that does not include the edge $x'y'$ leaves a $2$-colourable graph. Adding a maximal such matching leaves at most one vertex that has degree $1$.  Thus we have a $C_5$-colourable saturated subgraph of $G$, which is impossible by Fact \ref{fact:saturated}. 
\end{proof}

 With Claim \ref{cl:k23} proved, we are ready to finish our proof of the proposition.  As $G$ has girth $4$ we may assume it contains a $4$-cycle $C$ on the vertices $0,1,2,3$. 
 Every vertex $i$ of $C$  has degree at least $3$, so has at least one more neighbour $i'$.
    Let $N_i$ be the set of neighbours of $i$ in $G \setminus C$ for vertex $i$ of $C$. No two vertices of $C$ share a common neighbour not in $C$ or we would have a $C_3$ or a $K_{2,3}$, so the sets $N_i$ are disjoint; and so there are no edges from $N_i$ to $j \in V(C)$ for $i \neq j$.   

Notice the vertices of $C$ are saturated in $C \cup \bigcup N_i$, as by definition of the $N_i$ they can have no other neighbours in $G$.     
 
\begin{claim}
  If all vertices in $N= \cup_{i \in V(C)} N_i$ have degree exactly $2$, then $G$ is $E_3$. \end{claim}
 \begin{proof}
   Assume they do.  By Lemma \ref{lem:BB'} with $B = V(C)$ and  $B' = N$, we can get $G$ by adding edges between vertices of $B'$.  There can be no edge from $N_i$ to $N_{i+1}$ or its endpoints are in a $C_4$  contradicting the fact they have degree $2$.  So we may assume we have an edge from $1' \in N_1$ to $3' \in N_3$.  The graph is still $C_5$-colourable, so there is another edge between vertices of $B'$ in $G$. 
   If it is between $N_1$ and $N_3$ then the path $1 \sim 1'' \sim 3'' \sim 3$ that it makes is redundant with $1 \sim 1' \sim 3' \sim 3$. If it is between $N_2$ and $N_0$ then $G$ is $E_3$, and we are done.  
  \end{proof}
  
We may therefore assume that $1' \in N_1$ has degree at least $3$. It must have an edge to some $3'$ in $N_3$ or we have an $\H_3$. For $i = 1,3$ let $N'_i = N(i') \setminus (V(C) \cup N)$. So our graph looks like this.

   \begin{center}
    \begin{tikzpicture}[every node/.style={labvert}]
      \begin{scope}[xshift=0]
       \LabVerts{0/0/0, 1/1/1, 2/2/0, 3/1/-1, 1'/1/2, 3'/1/-2, 0'/-1/0, 2'/3/0}
       \Edges{1/2, 2/3, 3/0, 0/1, 1/1', 2/2', 3'/3, 0'/0}
       \Edges[bend right=90]{1'/3'}
       
       \draw (1') node[ellipse,minimum width=40pt, minimum height=20pt] (N1) {};
       \draw (1') + (1,.1) node[empty] (){$N_1$};
       \draw (2') node[ellipse,minimum width=20pt, minimum height=40pt, gray, dashed] (N2) {};
       \draw[gray] (2') + (.6,.5) node[empty] (){$N_2$};
       \draw (3') node[ellipse,minimum width=40pt, minimum height=20pt, gray, dashed] (N3) {};
       \draw[gray] (3') + (1,.1) node[empty] (){$N_3$};
       \draw (0') node[ellipse,minimum width=20pt, minimum height=40pt, gray, dashed] (N4) {};
       \draw[gray] (0') + (-.6,.5) node[empty] (){$N_4$};
       \draw (1') + (0,1) node[ellipse,minimum width=40pt, minimum height=20pt] (N1') {$N'_1$};
       \draw (3') + (0,-1) node[ellipse,minimum width=40pt, minimum height=20pt] (N3') {$N'_3$};
       \draw[dashed,gray] (N4.north) edge (0.north); \draw[dashed,gray]  (N4.south) edge (0.south);
       \draw[dashed,gray]  (N2.north) edge (2.north); \draw[dashed,gray]  (N2.south) edge (2.south);
       \draw (N3'.east) edge (3'.east); \draw (N3'.west) edge (3'.west);
       \draw (N1'.east) edge (1'.east); \draw (N1'.west) edge (1'.west);
       \draw (N1.east) edge (1.east); \draw (N1.west) edge (1.west);
       \draw[dashed] (N3.east) edge[gray] (3.east); \draw[dashed] (N3.west) edge[gray] (3.west);
      \end{scope}
    \end{tikzpicture}  
   \end{center}

 Now, the following hold, reducing the figure to what you get by ignoring the dashed bits. 
\begin{enumerate}
  \item $N_1$ and $N_3$ are disjoint, as otherwise we would have a $K_{2,3}$.   If there is some $1^* \in N_1\setminus \{1'\}$ and some $3^* \in N_3\setminus \{3'\}$,
   then we have a $\H_3$  on $\{2,1^* \},1',3',3, \{0,3^* \}$. So we may assume that $N_3 = \{3'\}$. 
 \item If $N_2$ has two vertices, then we can find an $\H_3$ unless $N_1 = \{1'\}$ and $N'_1$ and $N'_3$ are empty. But if all these conditions hold, then taking $B = \{0,1,2,3,1',3'\}$ and $B' = \{0',2'\}$ in Lemma \ref{lem:BB'} we get that $0',2'$ is an edge in $G$. But this makes a $E_3$.   So we may assume that $N_2 = \{2'\}$ and similarly that $N_0 = \{0'\}$.    
 \item  By Lemma \ref{lem:c4b}, if $2'$ were a $3^+$ vertex, it would have an edge to $0$ or $0'$, making a $K_{2,3}$ or a copy of $E_3$ respectively. Thus $2'$, and similarly $0'$, is a $2$-vertex.    
 \item  So $0'$ and $2'$ are not in any $C_4$, so have no edges to  $ \{1',3'\} \cup N_1^*$, and so in particular, are not in $N'_1 \cup N'_3$, (meaning $N'_1 \cup N'_3$ has no edges to $\{0,2\}$). 
  \item No vertex in $N_1$ has degree greater than $3$.  Indeed, the only possible neighbours of a vertex $v$ of $N_1$ are in $N'_1 \cup N'_3 \cup \{3'\}$ or are vertices not shown in the figure.  So using the two of them that are not $1$ or $3'$ as leaves, we have an $\H_3$ with the vertices
  $v, 1,2,3,\{0,3'\}$. The same holds for $3'$-- it has degree at most $3$.   
  \item A vertex $1''$ in $N'_1$ can only have as neighbours $0',1'$ or $2'$, or vertices not shown in the figure. If $1''$ has degree at least $3$, then using the two neighbours that are not $1'$ as leaves, there is an $\H_3$ on $1',3',3,\{0,2\}$.  Thus vertices in $N'_1$ and $N'_3$ are $2$-vertices.    
  \item If there are distinct  $1^* \in N_1 \setminus \{1'\}$ and  $3'' \in N'_3$, then we have a $\H_3$  on the vertices 
 $\{1',3''\},3',3,2,1,\{0,1^*\}$.  So either  $N_3' =  N_1 \setminus \{1'\} = \{1^{*}\}$, or at least one of $N_3'$ and $N_1 \setminus \{1'\}$ are empty. Assume the former.  This vertex $1^*$ has the same two neighbours as $1'$ does, so they must both have distinct third neighbours $1^{**}$ and $1''$ respectively. These are $2$-vertices by (6), and the vertices $1^*$ and $1'$ are $3$-vertices by (5).  Applying Lemma \ref{lem:BB'} with $B = \{0,1,2,3,1',1^*,3'\}$ and $B' = \{0',2'\}$ gives that $2',0'$ is an edge of $G$, but this is not true.  Thus either $N'_3$ is empty, or $N_1 = \{1'\}$.   
\end{enumerate}

In either case we apply Lemma \ref{lem:BB'} with $B = V(C) \cup \{3'\} \cup N_1$ and $B' = N'_1 \cup \{0'\} \cup N'_3 \cup \{2'\}$. (By definition, all neighbours of elements of $B$ are in $B \cup B'$, so elements of $B$ are saturated; elements of $B'$ have been argued to be $2$-vertices.)  
 Edges from $0'$ to $2'$ give a proper subgraph $E_3$ which is impossible.  Edges from $N'_1$ to $N'_3$, or from $\{0',2'\}$ to $N'_1 \cup N'_3$ leave the graph $C_5$-colourable, so adding them until there is at most one unsaturated vertex left gives a saturated $C_5$-colourable subgraph, contradicting the criticality of $G$. 
    \end{proof}

\section{The Proofs of Theorems~\ref{t-bipartite} and~\ref{t-general}}\label{a-cer}

We first give a general lemma (see also~\cite{BGMOPS22}, which only proves the first part of the lemma explicitly).

\begin{lemma}\label{lem:s3c-c4-simple}
Suppose $G$ is star $3$-colourable and contains a $C_4$ as a subgraph. Then only two vertices of this $C_4$ can have degree greater than $2$, and if they exist they must be opposite vertices of the $C_4$. Moreover, in any star $3$-colouring of $G$, the vertices in the $C_4$ of degree $2$ must have the same colour.
\end{lemma}

\begin{proof}
Suppose we have a $C_4$ in which three vertices have degree greater than $2$. Two of these must be adjacent in the $C_4$, let us call them $u$ and $v$. If they both have private neighbours outside of the vertices of the $C_4$, then $G$ contains $\mathbb{A}$ as a subgraph, and is a no-instance. If they have a common neighbour outside of the vertices of the $C_4$, then $G$ contains a graph built from $C_4$ and $K_3$ by identifying an edge of the former with an edge of the latter as a subgraph (which contains a $C_5$), and is a no-instance. If they both have their third neighbour in the vertices of the $C_4$, then the vertices of the $C_4$ induce a $K_4$ in $G$, and this is a no-instance. Finally, is the possibility that one has a new neighbour outside of the $C_4$, while the other has a new neighbour inside of the $C_4$. Then $G$ contains as a subgraph a diamond with a pendant vertex attached to one of the vertices of degree $2$, and is a no-instance.
\end{proof}

\begin{figure}
	\centering
	\includegraphics[width=0.15\textwidth]{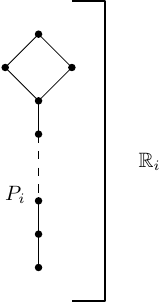}
	\caption{The graph $\mathbb{R}_i$.}
	\label{f-rack}
\end{figure}

Let $\mathbb{R}_i$ be the \emph{squash racket} graph built from a path on $i$ vertices by identifying one of the ends of the path with some vertex on a $C_4$; see also Figure~\ref{f-rack}. In $\mathbb{R}_i$, we call the vertex on the $C_4$ that is of degree $2$, and whose neighbours are also of degree $2$, \emph{distinguished}. In the following we often give a star $3$-colouring when we do not know the precise length of some path. It is sufficient to cover all residues modulo $3$ as we can repeat sequences $abc$ on a path, so long as $a,b,c \in \{1,2,3\}$ and $|\{a,b,c\}|=3$. We will not discuss this further.

\subsection{The Bipartite Case: The Proof of Theorem~\ref{t-bipartite}}

In this section we prove Theorem~\ref{t-bipartite}. We first prove the following lemma.

\begin{lemma}\label{l-not}
The graphs ${\mathbb A}$ and $\overline{2 P_4 + 3 P_6}$ are not star $3$-colourable.
\end{lemma}

\begin{proof}
By Lemma~\ref{lem:s3c-c4-simple} we find that ${\mathbb A}$ is not star $3$-colourable.

We now prove that $\overline{2 P_4 + 3 P_6}$ is not star $3$-colourable.
Assume otherwise. Let $x_1$ be the neighbour of $u$ on one copy of $P_4$ and $x_2$ the neighbour of $v$ on this path. Similarly, let $y_1$ be the neighbour of $u$ on the second copy of $P_4$ and $y_2$ the neighbour of $v$. Label the internal vertices of the first copy of $P_6$ $z_1 \dots z_4$ where $z_1$ is the neighbour of $u$ and $z_4$ is the neighbour of $v$. Let $w_1 \dots w_4$ be the internal vertices of the second copy of $P_6$ such that $w_1$ is adjacent to $u$ and $w_4$ is adjacent to $v$.

First consider a proper star $3$-colouring $c$ in which $u$ and $v$ are coloured differently, say $u$ is coloured $1$ and $v$ is coloured $2$. It must be that either $c(x_1)=2$ or $c(x_2)=1$. Without loss of generality, $c(x_1)=2$. Then $c(x_2)=3$. Every remaining neighbour $a$ of $v$ then has colour $1$, else we have a bichromatic $P_4$ with vertex set $\{a,v,x_2, x_1 \}$. To avoid a bichromatic $P_4$, $c(y_1)=3$. Each further neighbour $b$ of $u$ then has colour $2$, else we have a bichromatic $P_4$ with vertex set $\{b,u, y_1,y_2 \} $. We now consider the first copy of $P_6$. Since $c(z_1)=2$, $c(z_2)=3$, else there is a bichromatic $P_4$ with vertex set $\{x_1,u,z_1,z_2\}$. However, $c(z_3)=3$, else we have a bichromatic $P_4$ with vertex set $\{z_3, z_4, v, y_2\}$. This leaves two adjacent  vertices coloured $3$, a contradiction.

Therefore, $u$ and $v$ must be coloured alike in any proper star $3$-colouring $c$. Without loss of generality, they are both assigned colour $1$. Without loss of generality, $c(x_1)=2$ and $c(x_2)=3$. Now consider the first copy of $P_6$. First assume that $c(z_1)=2$. Then, to avoid a bichromatic $P_4$ with vertex set $\{z_2, z_1, u, x_1\}$, $c(z_2)=3$. Now, to avoid a bichromatic $P_4$ along this path, at least one of $z_3$ and $z_4$ must be coloured $2$. If $c(z_3)=2$ then $c(z_4)=3$ and we have a bichromatic $P_4$. Therefore, $c(z_4)=2$ and $c(z_3)=1$. Each remaining neighbour of $v$ then has colour $3$ to avoid a bichromatic $P_4$. Since $c(w_4)=3$, $c(w_3)=2$, else we have a bichromatic $P_4$. If $c(w_2)=3$, then $c(w_1)=2$ and we have a bichromatic $P_4$ on $\{w_1,w_2,w_3,w_4\}$. Thus, $c(w_2)=1$, whereupon $c(w_1)=3$ to avoid a bichromatic $P_4$ on $\{x_1,u,w_1,w_2\}$. The same arguments apply to the third copy of $P_6$ but now we have a bichromatic $P_4$ on $1313$ with vertex set $\{u, w_1, w_2 \}$ plus the first vertex on the third copy of $P_6$, a contradiction. 

It remains to consider the case where $c(z_1)=3$. If $c(z_2)=1$ then $c(z_3)=2$ and $c(z_4)=3$. Additionally, every remaining neighbour $c$ of $u$ is coloured $2$ to avoid a bichromatic $P_4$ with vertex set $\{z_2, z_1, u, c \}$. It must then be that $c(w_1)=2, c(w_2)=3$. In fact, this case is covered in the previous paragraph by an automorphism of $\overline{2 P_4 + 3 P_6}$  that swaps $u$ and $v$, as well as colours $2$ and $3$.
Therefore $c(z_2)=2$. To avoid a bichromatic $P_4$ along this copy of $P_6$, $c(z_3)=1$. Now $c(z_4)=2$ since otherwise we have a bichromatic $P_4$ with vertex set $\{z_3, z_4, v, x_2\}$. However,  this induces a bichromatic $P_4$ with vertex set $\{z_2,z_3,z_4,v\}$, a contradiction.
\end{proof}

\noindent
We continue with the following lemma. For an integer~$d$, a {\it $d^+$-vertex} in a graph is a vertex of degree at least~$d$.

\begin{lemma}\label{lem:bipartite-star-3-col-many}
Let $G$ be a bipartite connected $(\H_2,\H_4,\ldots)$-subgraph-free graph with more than two $3^+$-vertices. If $G$ is  $\mathbb{A}$-free, then $G$ is star $3$-colourable.
\end{lemma}

\begin{proof}
We may assume without loss of generality that $G$ is connected.
Suppose $G$ is $\mathbb{A}$-subgraph-free.
Note that $G$ cannot contain a $K_3$ since it is bipartite. By assumption, $G$ has more than two $3^+$-vertices. Let us choose three in a row $u,v,w$ on a path in that order, 
so as to minimise the total distance from $u$ to $w$. Note it is possible that there is a path from $u$ to $w$ that is shorter than the path from $u$ to $w$ via $v$.\footnote{We are not aware of an example like this that is both omits all $\H_{2i}$ as a subgraph and has such $u,v,w$ minimally chosen.}

Suppose $u,v,w$ are consecutive. Then $u$ and $w$ must have another common neighbour to avoid an $\H_2$. Now we violate Lemma~\ref{lem:s3c-c4-simple}. Thus, $u,v,w$ may not be consecutive.

Suppose now that $u$ and $w$ have two common neighbours $p$ and $q$. By the minimality of the length of the path from $u$ to $v$ to $w$ we may now assume that the subpaths from $u$ to $v$ and $v$ to $w$ have length at most $2$, and so as $u,v,w$ are not consecutive one of them has length $2$. Without loss of generality, let us assume this is $u$ to $v$ . According to Lemma~\ref{lem:s3c-c4-simple}, both $p$ and $q$ have degree $2$ and neither of these may have an edge to $v$, therefore $v$ has a distinct neighbour $z$. Since there is a path of length $2$ from $u$ to $v$ we have $\H_2$ as a subgraph and finish this case.

Suppose now that $u$ and $w$ have one common neighbour $p$ and two distinct neighbours $u'$ and $w'$, respectively. Since $u,v,w$ are not consecutive, there is an $\H_2$ using $u,p,w$ as the central path.

Suppose now that $u$ and $w$ have neighbours $u',u''$ and $w',w''$, respectively, so that these vertices are pairwise distinct. If the total length of the path on $u,v,w$ is even then we have some $\H_{2i}$. W.l.o.g., assume the path from $u$ to $v$ is even, whereupon the path from $v$ to $w$ must be odd. There must be an edge from $v$ to one of $u',u''$ (\mbox{w.l.o.g.} $u'$) to avoid an $\H_{2i}$. By minimality of the path $u,v,w$, the distance from $u$ to $v$ along the path we chose must be $2$ with middle vertex $s$. According to Lemma~\ref{lem:s3c-c4-simple}, both $s$ and $u'$ have degree $2$ in $G$. If $u$ has degree greater than $3$ there is an $\H_2$ with middle path $u,s,v$ (note that there cannot be an edge from $u$ to the neighbour $d$ of $v$ along the path to $w$, by our original minimal choice of the path along $u,v,w$).  Suppose $u''$ has no neighbours outside of $\{u,u',s,v\}$ (noting that there might be an edge from $u''$ to $v$). Then $v$ is a cut vertex. 
Note that none of the paths emanating from $w$ and not going through $v$ can have vertices of degree greater than $2$, unless they are part of a squash racket whose distinguished vertex is identified with $w$, as we would introduce an $\H_{2i}$. The case of the squash racket can be dealt with just like the case of an even cycle so we do not concern ourselves with it specially. Note however that there can be only one squash racket attached at $w$ and if there is a squash racket there can be no additional even cycles or paths attached at $w$. If the edge $u''$ to $v$ exists then there can be no new neighbours of $u''$ without introducing an $\mathbb{A}$ as a subgraph. This places us in the case shown in Figure~\ref{fig:bipartite-star-3-col-many-cut-2}, and the figure shows a star $3$-colouring in this case.

Suppose now that $u''$ has a neighbour outside $\{u,u',s,v\}$ which we call $t$ (noting that in this case there can be no edge from $u''$ to $v$ since then there would be three consecutive $3^+$ vertices). Let us  continue expanding from $t$ along some path away from $u''$ which we enumerate $t=t_1,t_2,\ldots$. Suppose no such path like this reaches $w$ (or $v$) avoiding $u$. Note that such a path cannot now arrive at either $v$ or $s$ since we have uncovered all of their neighbours (if $v$ had more neighbours then there would be an $\H_2$ with central path $u,s,v$). If one of the vertices $t_j$ has degree greater than $2$ then we have a $\H_{2i}$ involving a central path either from $t_j$ to $v$ via $u$ or from $t_j$ to $w$ via $u$. Thus we have nothing more than a path here which must terminate at some $t_{k}$. This case will be superseded by a more general case which will come later (see Figure~\ref{fig:bipartite-star-3-col-many-final}).

Suppose we have a path $t_1,t_2,\ldots$ that eventually joins the path between $v$ and $w$ exclusive of $v$ and $w$. Then we contradict minimality of the path from $u$ to $w$ via $v$. Thus we have a path $t_1,t_2,\ldots$ that eventually joins $w$. (Possibly it did so via, w.l.o.g., $w''$.) This path must be of odd length since otherwise it would produce an $\H_{2i}$. 
If any vertex on this path has degree greater than $2$ then it induces some $\H_{2i}$ on a path to either $u$ or $w$. Except in the case that it is distance $2$ from $w$, an odd distance from $u$, and shares two neighbours with $w$ (and these are $w$'s only neighbours). This case is drawn in Figure~\ref{fig:bipartite-star-3-col-many-middle}. Therefore, there is nothing but this path, except that we did not isolate the degree of $w$. Indeed, $w$ may have arbitrarily high degree. Any path emanating from $w$ and not eventually reaching $v$ (we have accounted for all such paths) may not have a vertex $c$ of degree greater than $2$, since otherwise this vertex would produce an $\H_{2i}$ with middle path to either $w$ or $v$. Except in the case that $c$ is at distance $2$ from $w$ and shares two neighbours the two neighbours $w'$ and $w''$ with $w$, in which case $c$ may have a path emanating from it in the form of $c=c_1$, $c_2$, etc. This case, in which there is a squash racket with distinguished vertex identified with $w$, will be superseded by the next. Note however that there can be only one squash racket attached at $w$ and if there is a squash racket there can be no additional even cycles or paths. It follows then that some path may leave $w$ and terminate in a vertex of degree $1$, or return to $w$ producing a cycle of even length. We claim the graph $G$ is a yes-instance of star $3$-colouring. Let us refer to Figure~\ref{fig:bipartite-star-3-col-many-final}.
\end{proof}

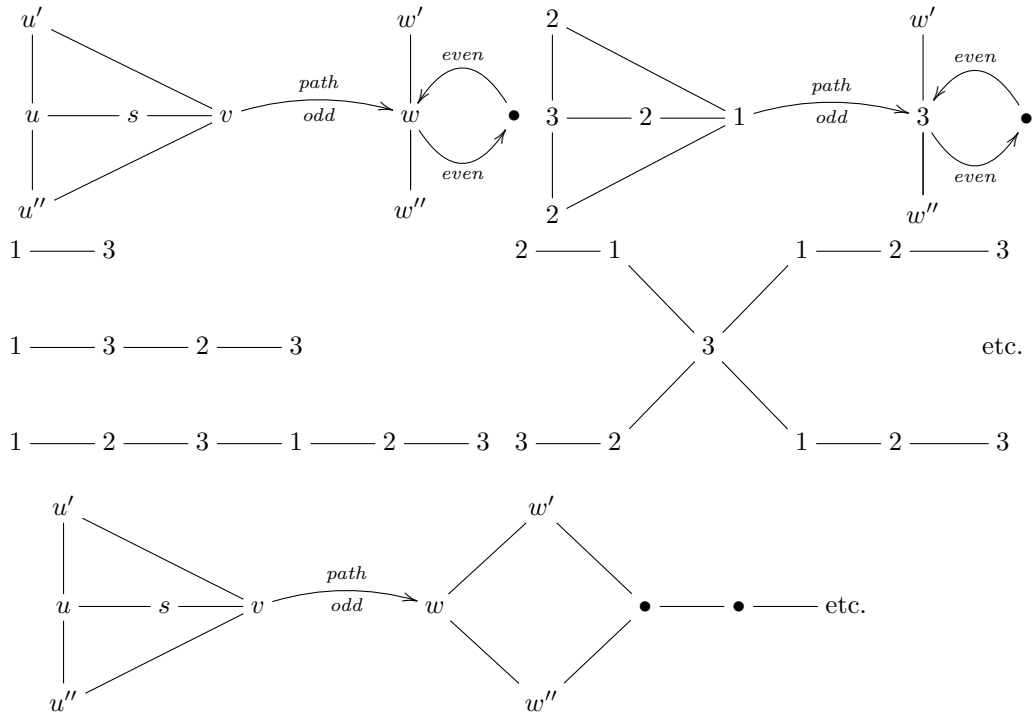
\begin{figure}
\mbox{
$
\xymatrix{
u' \ar@{-}[d] & & & & w' \ar@{-}[d] & \\
u \ar@{-}[r] & s \ar@{-}[r] & v \ar@{-}[lld] \ar@{-}[llu] \ar@/^0.5pc/^{\mathit{path}}_{\mathit{odd}}[rr] & &  w \ar@{-}[d] \ar@/_1.5pc/_{\mathit{even}}[r] & \bullet \ar@/_1.5pc/_{\mathit{even}}[l] \\ 
u'' \ar@{-}[u]  &  & & &  w'' \ar@{-}[u] & \\
}
$
$
\xymatrix{
2 \ar@{-}[d] & & & & w' \ar@{-}[d] & \\
3 \ar@{-}[r] & 2 \ar@{-}[r] & 1 \ar@{-}[lld] \ar@{-}[llu] \ar@/^0.5pc/^{\mathit{path}}_{\mathit{odd}}[rr] & &  3 \ar@{-}[d] \ar@/_1.5pc/_{\mathit{even}}[r] & \bullet \ar@/_1.5pc/_{\mathit{even}}[l] \\ 
2 \ar@{-}[u]  &  &  &  & w'' \ar@{-}[u] & \\
}
$
}

$
\xymatrix{
1 \ar@{-}[r] & 3 \\
1 \ar@{-}[r] & 3  \ar@{-}[r] & 2  \ar@{-}[r] & 3 \\
1 \ar@{-}[r] & 2  \ar@{-}[r] & 3  \ar@{-}[r] & 1 \ar@{-}[r] & 2  \ar@{-}[r] & 3 \\
}
$
$
\xymatrix{
2 \ar@{-}[r] & 1  \ar@{-}[dr]  & & 1 \ar@{-}[r] & 2  \ar@{-}[r] & 3 \\
  & & 3  \ar@{-}[ur]  \ar@{-}[dr] &  & & \mbox{etc.}\\
3 \ar@{-}[r]  & 2  \ar@{-}[ur] &  & 1 \ar@{-}[r] & 2  \ar@{-}[r] & 3 \\
}
$
\[
\xymatrix{
u' \ar@{-}[d] & & & &  & w' \ar@{-}[dr]& \\
u \ar@{-}[r] & s \ar@{-}[r] & v \ar@{-}[lld] \ar@{-}[llu] \ar@/^0.5pc/^{\mathit{path}}_{\mathit{odd}}[rr] & &  w \ar@{-}[ur] \ar@{-}[dr] & & \bullet \ar@{-}[r] & \bullet \ar@{-}[r] & \mbox{etc.} \\ 
u'' \ar@{-}[u]  &  & & & & w'' \ar@{-}[ur] & \\
}
\]
\caption{The second cut vertex case of Lemma~\ref{lem:bipartite-star-3-col-many} with its star $3$-colouring. There may be a squash racket attached at $w$ instead of even cycles or paths. The squash racket case is drawn on the bottom. For the colouring around $w$ we draw to the left the possible colours that may appear to the left and we draw to the right two paths that need to be identified to make an even cycle of squash racket (middle right).}
\label{fig:bipartite-star-3-col-many-cut-2}
\end{figure}
\begin{figure}
$
\xymatrix{
u' \ar@{-}[d] \ar@{-}[drr] & & & & & \\
u \ar@/_1.5pc/^{\mathit{odd}}_{\mathit{path}_2}[drrr] \ar@{-}[r] & s \ar@{-}[r] & v  \ar@/^0.5pc/^{\mathit{odd}}_{\mathit{path}_1}[rr] & & w \ar@{-}[d]  &  \\
&  & & t_k \ar@{-}[r]  \ar@{-}@/_1.5pc/[rr] & w'' &  w' \ar@{-}[ul] \\
}
$
$
\xymatrix{
2 \ar@{-}[d] \ar@{-}[drr] & & & & & \\
3 \ar@/_1.5pc/^{\mathit{odd}}_{\mathit{path}_2}[drrr]  \ar@{-}[r] & 2 \ar@{-}[r] & 1  \ar@/^0.5pc/^{\mathit{odd}}_{\mathit{path}_1}[rr] & & 3 \ar@{-}[d]  &  \\
&   & & 2 \ar@{-}[r]  \ar@{-}@/_1.5pc/[rr] & 1 &  1 \ar@{-}[ul] \\
}
$

\

$
\xymatrix{
1 \ar@{-}[r] & 3 & & \mathit{path}_1\\
1 \ar@{-}[r] & 3  \ar@{-}[r] & 2  \ar@{-}[r] & 3 \\
1 \ar@{-}[r] & 2  \ar@{-}[r] & 3  \ar@{-}[r] & 1 \ar@{-}[r] & 2  \ar@{-}[r] & 3 \\
}
$
$
\xymatrix{
3 \ar@{-}[r] & 2 & & \mathit{path}_2\\
3 \ar@{-}[r] & 1  \ar@{-}[r] & 3  \ar@{-}[r] & 2 \\
3 \ar@{-}[r] & 2  \ar@{-}[r] & 1  \ar@{-}[r] & 2 \ar@{-}[r] & 3  \ar@{-}[r] & 2 \\
}
$
\caption{The middle case of Lemma~\ref{lem:bipartite-star-3-col-many} with its star $3$-colouring (above). One could think of this as a squash racket with distinguished vertex identified with $w$ in which the path continues and actually reaches $u$. The concurrent star $3$-colouring of $\mathit{path}_1$ and $\mathit{path}_2$ is given below.}
\label{fig:bipartite-star-3-col-many-middle}
\end{figure}
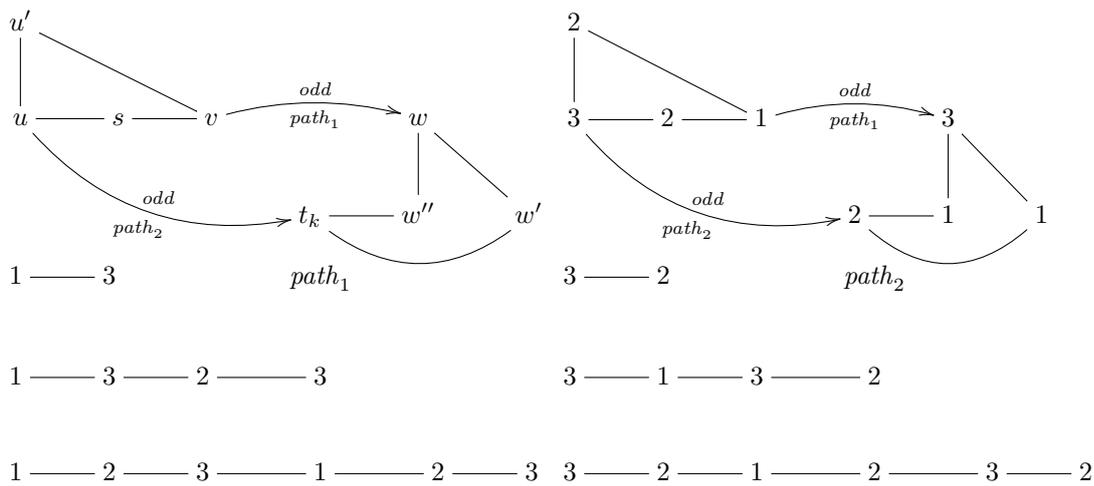

\begin{figure}
$
\xymatrix{
u' \ar@{-}[d] \ar@{-}[drr] & & & & w' \ar@{-}[d]  & \\
u \ar@{-}[d] \ar@{-}[r] & s \ar@{-}[r] & v  \ar@/^0.5pc/^{\mathit{odd}}_{\mathit{path}_1}[rr] & & w \ar@{-}[d] \ar@/_1.5pc/_{\mathit{even}}[r] & \bullet \ar@/_1.5pc/_{\mathit{even}}[l] \\
u'' \ar@{-}[r] & t \ar@/_1.5pc/^{\mathit{odd}}_{\mathit{path}_2}[rrru]  & & & w'' & \\
}
$
$
\xymatrix{
2 \ar@{-}[d] \ar@{-}[drr] & & & & w' \ar@{-}[d]  & \\
3 \ar@{-}[d] \ar@{-}[r] & 2 \ar@{-}[r] & 1  \ar@/^0.5pc/^{\mathit{odd}}_{\mathit{path}_1}[rr] & & 3 \ar@{-}[d] \ar@/_1.5pc/_{\mathit{even}}[r] & \bullet \ar@/_1.5pc/_{\mathit{even}}[l] \\
2 \ar@{-}[r] & 1 \ar@/_1.5pc/^{\mathit{odd}}_{\mathit{path}_2}[rrru]  & & & w'' & \\
}
$

$
\xymatrix{
1 \ar@{-}[r] & 3 & & & \\
1 \ar@{-}[r] & 3  \ar@{-}[r] & 2  \ar@{-}[r] & 3 \\
1 \ar@{-}[r] & 2  \ar@{-}[r] & 3  \ar@{-}[r] & 2 \ar@{-}[r] & 1  \ar@{-}[r] & 3 \\
}
$
$
\xymatrix{
2 \ar@{-}[r] & 1  \ar@{-}[dr]  & & 1 \ar@{-}[r] & 2  \ar@{-}[r] & 3 \\
  & & 3  \ar@{-}[ur]  \ar@{-}[dr] &  & & \mbox{etc.}\\
3 \ar@{-}[r] & 2  \ar@{-}[ur] &  & 1 \ar@{-}[r] & 2  \ar@{-}[r] & 3 \\
}
$

\

$
\xymatrix{
1 \ar@{-}[r] & 2  \ar@{-}[r] & 3  \ar@{-}[r] & 1 \\
}
$
$
\xymatrix{
2 \ar@{-}[r] & 3  \ar@{-}[dr]  & & 2 \ar@{-}[r] & 3  \ar@{-}[r] & 1 \\
  & & 1  \ar@{-}[ur]  \ar@{-}[dr] &  & & \mbox{etc.}\\
 &  &  & 2 \ar@{-}[r] & 3  \ar@{-}[r] & 1 \\
}
$
\caption{The coup de gr\^ace of Lemma~\ref{lem:bipartite-star-3-col-many} with its star $3$-colouring (above). The concurrent star $3$-colouring of $\mathit{path}_1$ and $\mathit{path}_2$ (middle left) and any even cycles or squash rackets (middle right, merge top and bottom paths to form even cycle or squash racket). At the bottom is the special case where both paths have length $1 \bmod 3$ in which case $w$ is coloured $2$.}
\label{fig:bipartite-star-3-col-many-final}
\end{figure}

\begin{figure}
\[
\xymatrixrowsep{0.1in}
\xymatrixcolsep{0.2in}
\xymatrix{
& 3 \ar@{-}[d]  & & 3 \ar@{-}[d] \\
& 1 \ar@{-}[d]  & & 1 \ar@{-}[d] \\
 & 2 \ar@{-}[dr] &  &  2 \ar@{-}[dl] \\  
1 \ar@{-}[r] & 2  \ar@{-}[r] & 3 \ar@{-}[r] & 1  & \\
2 \ar@{-}[u] & & 2 \ar@{-}[u] & 3 \ar@{-}[u] & \\
3 \ar@{-}[u] \ar@{-}[r] & 2 \ar@{-}[r] & 1 \ar@{-}[u] \ar@{-}[r] & 2 \ar@{-}[u] & \\
& 2 \ar@{-}[ur] &  & 2 \ar@{-}[ul] & \\  
& 3 \ar@{-}[u]  & & 3 \ar@{-}[u] \\
& 1 \ar@{-}[u]  & & 1 \ar@{-}[u] \\
}
\]
\caption{The first case of Lemma~\ref{lem:bipartite-star-3-col-two} with its star $3$-colouring. At the top and bottom the paths may be made into even cycles by the identification of two vertices. Note there is a $313$ down the path of length $4$ which would cause a problem were we to have another one, but we already explained this would introduce an even $\H_{2i}$.}
\label{fig:bipartite-star-3-col-two-even}
\end{figure}

\begin{figure}
\[
\xymatrixrowsep{0.1in}
\xymatrixcolsep{0.2in}
\xymatrix{
& & 2 \ar@{-}[d]  & & 2 \ar@{-}[d] \\
& & 1 \ar@{-}[d]  & & 1 \ar@{-}[d] \\
&  & 3 \ar@{-}[dr] &  &  3 \ar@{-}[dl] \\  
2 \ar@{-}[r] & 1 \ar@{-}[r] & 3 \ar@{-}[r] & 2 \ar@{-}[r] & 3 \\
3 \ar@{-}[u] & & &  & 1 \ar@{-}[u] \\
2 \ar@{-}[u] & & &  & 3 \ar@{-}[u] \\
1 \ar@{-}[u] \ar@{-}[r] & 3 \ar@{-}[r] & 2  \ar@{-}[r] & 1 \ar@{-}[uuu] \ar@{-}[r] & 2 \ar@{-}[u] \\
& & 2 \ar@{-}[ur] & &  2 \ar@{-}[ul] \\  
& & 3 \ar@{-}[u]  & & 3 \ar@{-}[u] \\
& & 1 \ar@{-}[u]  & & 1 \ar@{-}[u] \\
}
\]
\caption{The case from Lemma~\ref{lem:bipartite-star-3-col-two} in which there is path from $u$ to $v$ of length $1$ with its star $3$-colouring. At the top and bottom the paths may be made into even cycles by the identification of two vertices.}\label{fig:bipartite-star-3-col-two-odd-path-one}
\end{figure}
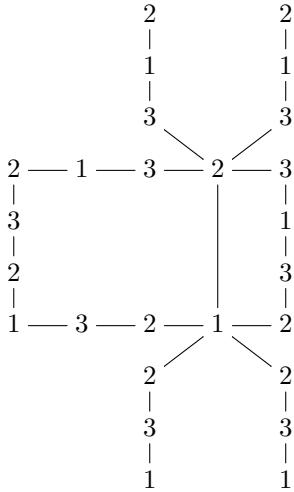
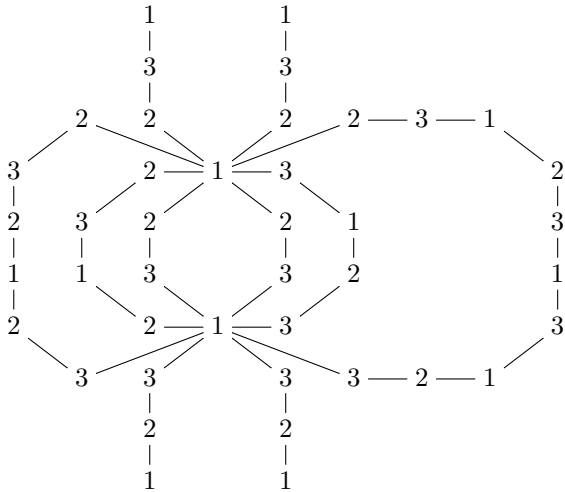
\begin{figure}
\[
\xymatrixrowsep{0.1in}
\xymatrixcolsep{0.2in}
\xymatrix{
& & 1 \ar@{-}[d]  & & 1 \ar@{-}[d] & \\
& & 3 \ar@{-}[d]  & & 3 \ar@{-}[d] & \\
 & 2 & 2 \ar@{-}[dr] & &  2 \ar@{-}[dl] & 2 \ar@{-}[r] & 3 \ar@{-}[r] & 1 \ar@{-}[dr] \\  
3 \ar@{-}[ur] \ar@{-}[d] & & 2  \ar@{-}[dl]  \ar@{-}[r] & 1 \ar@{-}[ull]  \ar@{-}[r] \ar@{-}[urr] & 3  \ar@{-}[dr] & & & & 2 \ar@{-}[d]\\
2 \ar@{-}[d] & 3 & 2 \ar@{-}[ru] & & 2 \ar@{-}[lu] & 1 \ar@{-}[d]  & & & 3 \ar@{-}[d]\\
1 \ar@{-}[d] & 1 \ar@{-}[u] & 3 \ar@{-}[u] & & 3 \ar@{-}[u]  & 2  & & & 1 \ar@{-}[d] \\
2 \ar@{-}[dr] &  & 2  \ar@{-}[ul] \ar@{-}[r]   & 1 \ar@{-}[ur] \ar@{-}[r] \ar@{-}[ul] \ar@{-}[dll] \ar@{-}[drr]   & 3 \ar@{-}[ur]  & & & & 3 \ar@{-}[dl]\\
& 3 & 3 \ar@{-}[ur] &  & 3 \ar@{-}[ul] &  3 \ar@{-}[r] & 2 \ar@{-}[r] & 1\\  
& & 2 \ar@{-}[u]  & & 2 \ar@{-}[u] & \\
& & 1 \ar@{-}[u]  & & 1 \ar@{-}[u] & \\
}
\]
\caption{The case of Lemma~\ref{lem:bipartite-star-3-col-two} in which, from $u$ to $v$, there are no paths of length $1$ and at most two paths of length $5$, with its star $3$-colouring.  There may be more paths of length $2$ between $u$ and $v$ but no more paths of length $5$. At the top and bottom the paths may be made into even cycles by the identification of two vertices.}
\label{fig:bipartite-star-3-col-two-odd-path-three}
\end{figure}
\begin{figure}
\[
\xymatrixrowsep{0.1in}
\xymatrixcolsep{0.2in}
\xymatrix{
& 1 \ar@{-}[d]  & & 1 \ar@{-}[d] & \\
& 3 \ar@{-}[d]  & & 3 \ar@{-}[d] & \\
 & 2 \ar@{-}[dr] & &  2 \ar@{-}[dl] & 2 \ar@{-}[r]  & 3 \ar@{-}[d]\\  
3   \ar@{-}[r]  \ar@{-}[d] & 2   \ar@{-}[r] & 1  \ar@{-}[r] \ar@{-}[urr] & 2 \ar@{-}[d]  & & 2\\
1 &  & 3 \ar@{-}[u] & 3  &  & 1 \ar@{-}[u]\\
2 \ar@{-}[u] &  & 1 \ar@{-}[u] & 1 \ar@{-}[u]  &  & 3 \ar@{-}[u]\\
1   \ar@{-}[r]  \ar@{-}[u] & 3  \ar@{-}[r]   & 2 \ar@{-}[u]  \ar@{-}[r]  \ar@{-}[drr]   & 3 \ar@{-}[u] & & 2 \ar@{-}[u]\\
& 3 \ar@{-}[ur] &  & 3 \ar@{-}[ul] & 3 \ar@{-}[r]  & 1 \ar@{-}[u]\\  
& 1 \ar@{-}[u]  & & 1 \ar@{-}[u] & \\
& 2 \ar@{-}[u]  & & 2 \ar@{-}[u] & \\
}
\]
\caption{The final case of Lemma~\ref{lem:bipartite-star-3-col-two} in which, from $u$ to $v$, there are no paths of length $1$ and at most one path of length $3$, with its star $3$-colouring.}
\label{fig:bipartite-star-3-col-two-odd-path-final}
\end{figure}
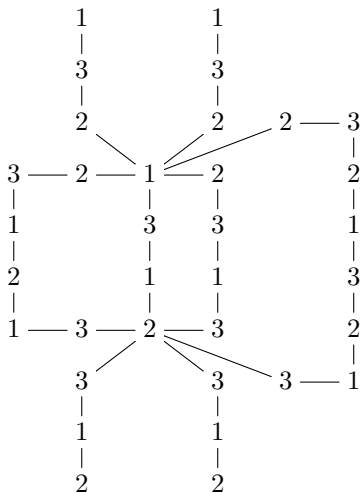

\noindent
We now prove the following lemma, which deals with the case where $G$ has exactly two $3^+$-vertices.

\begin{lemma}
Let $G$ be a bipartite $(\H_2,\H_4,\ldots)$-subgraph-free graph with exactly two $3^+$ vertices. If $G$ is $(\mathbb{A},\overline{2 P_4 + 3 P_6})$-subgraph-free, then $G$ is star $3$-colourable.
\label{lem:bipartite-star-3-col-two}
\end{lemma}

\begin{proof}
We may assume without loss of generality assume that $G$ is connected. Suppose $G$ is $(\mathbb{A},\overline{2 P_4 + 3 P_6})$-subgraph-free.
Let the two $3^+$ vertices be $u$ and $v$. There may be many paths from $u$ to $v$. Suppose one of these paths is of even length. Then, as $G$ is bipartite, they must all be of even length. 
There may further be even cycles attached to $u$ and even cycles attached to $v$. In order to avoid $\H_{2i}$, the number of these paths between $u$ and $v$ must be bounded, in particular, there can be no more than one path of length greater than or equal to $4$. See Figure~\ref{fig:bipartite-star-3-col-two-even} for the $3$-star colouring in this case.

Suppose, one of the paths between $u$ and $v$ is of odd length. Then they must all be of odd length. (There may further be even cycles attached to $u$ and even cycles attached to $v$.) Suppose there are paths between $u$ and $v$ of length both $1$ and $3$. Now, $u$ and $v$ must have third neighbours $u'$ and $v'$, respectively, that are not involved in these paths (e.g., by bipartiteness). Since $u'\neq v'$, again by bipartiteness, we have an $\mathbb{A}$ as a subgraph and this is a no-instance.

Now we consider three cases. Firstly, that there is a path of length $1$ and all other paths are of length at least $5$ (see Figure~\ref{fig:bipartite-star-3-col-two-odd-path-one}). Secondly, that all paths are of length at least $3$ and there are at most two paths of length $5$ (see Figure~\ref{fig:bipartite-star-3-col-two-odd-path-three}). Thirdly, that all paths are of length at least $3$ and there is at most one path of length $3$ (see Figure~\ref{fig:bipartite-star-3-col-two-odd-path-final}).
\end{proof}

\noindent
We are now ready to prove Theorem~\ref{t-bipartite}, which we restate below.

\medskip
\noindent
{\bf Theorem~\ref{t-bipartite} (restated)}
{\it A bipartite $(\H_2,\H_4,\H_6\ldots)$-subgraph-free~graph is star $3$-colourable if and only if it is $(\mathbb{A},\overline{2P_4 + 3P_6})$-subgraph-free.} 

\begin{proof}
Let $G$ be a bipartite $(\H_2,\H_4,\H_6\ldots)$-subgraph-free graph. If $G$ contains $(\mathbb{A}$ or $\overline{2P_4 + 3P_6}$ as a subgraph, we apply Lemma~\ref{l-not}. If $G$ is $(\mathbb{A},\overline{2P_4 + 3P_6})$-subgraph-free, then we can use Lemma~\ref{lem:bipartite-star-3-col-many} or Lemma~\ref{lem:bipartite-star-3-col-two} (note that the case in which there is only one or no $3^+$-vertex is covered a fortiori in the latter). 
\end{proof}

\subsection{The General Case: The Proof of Theorem~\ref{t-general}}

A {\it diamond} is a $C_4$ with a chord, that is, the graph with vertex set $\{a,b,c,d\}$ and edge set $\{ab,bc,cd,da,ac\}$.
Let \emph{diamond-with-pendant} be built from a diamond, with a pendant vertex added to a vertex of degree $2$, that is, the graph with vertex set $\{a,b,c,d,e\}$ and edge set $\{ab,bc,cd,da,ac,be\}$. 
We note that the diamond-with-pendant involves a $C_4$ with three $3^+$-vertices. See also Figure~\ref{f-diam}.

\begin{figure}[b]
	\centering
	\includegraphics[width=0.30\textwidth]{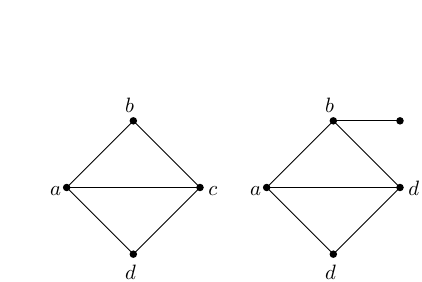}
	\caption{The diamond (left) and the diamond-with-pendant (right).}\label{f-diam}
\end{figure}

We start with proving the following lemma.

\begin{lemma}\label{lem:general-star-3-col-many}
Let $G$ be a connected graph $(\H_2,\H_4,\H_6\ldots)$-subgraph-free graph with at least five $3^+$-vertices. Then $G$ contains either a graph from $\{\mathbb{A},C_5\}$ as a subgraph or a $C_4$ with (at least) three $3^+$ vertices.
\end{lemma}

\begin{proof}
We will essentially follow the proof of Lemma~\ref{lem:bipartite-star-3-col-many} to show that, if we assume there are three $3^+$-vertices, we can never uncover more than one more, unless $G$ contains some $H_{2i}$ as a subgraph or contains $\mathbb{A}$ as a subgraph or contains some $C_4$ with (at least) three $3^+$ vertices. Of course, as we follow the proof of Lemma~\ref{lem:bipartite-star-3-col-many}, we can no longer use bipartiteness. This means we have an important new case at the start (later on we only used it to assume that certain cycles were even).

Suppose there are three consecutive $3^+$-vertices $u,v,w$ in a clique (the ``triangle''). Suppose any two of these have a common neighbour, then the graph is not star $3$-colourable by Lemma~\ref{lem:s3c-c4-simple}. Thus, they have private neighbours $u',v',w'$, respectively, that form with them the ``net''. Since $G$ is connected, there is a path from the net, \mbox{w.l.o.g.} through $u'$, to another $3^+$ vertex $t$. Choose such a $t$ closest to the triangle in the net. If $t$ has two neighbours outside of the net, then there is an $H_{2i}$ as a subgraph with central path either $t,\ldots,u$ (if distance from $t$ to $u$ even) or either of $t,\ldots,u,v$ and $t,\ldots,u,w$ (if distance from $t$ to $u$ odd). 
Thus, $t$ has a neighbour in the net and we have discovered a path from $u$ to (w.l.o.g.) $v$, avoiding $w$, with a $3^+$ vertex at distance either $1$ or $2$ from $v$. Let us choose a minimal such path (not permitting $v'=u'$ as we already ruled this our by assumption). This minimal path induces a cycle in $G$. If it is $C_4$ or $C_5$ we are in a no-instance (in both cases there is a -- potentially non-induced -- $C_5$). W.l.o.g., assume it is of length at least $6$.

Suppose $t$ is at distance $1$ from $v$ and we may \mbox{w.l.o.g.} assume it is $v'$ itself. If $v'$ has an edge to anywhere in the net, then we are in a no-instance (either we have a $C_5$ or diamond-with-pendant). Thus, some third neighbour $v''$ of $v'$ is outside the net. Now there is an $H_2$ with central path $v',v,u$. Suppose $t$ is at distance $2$ from $v$. We have found two neighbours of $t$ so let $t'$ be a third neighbour, If $t'$ is in the net then we are in a no-instance with a $C_5$; unless $t'=v$, in which case we were not in a minimal cycle (essentially $v'$ could be $t$). Thus, $t'$ is outside the net. Now there is an $H_2$ with central path $t,v',v$.

Now let us follow the proof as in the bipartite case (Lemma~\ref{lem:bipartite-star-3-col-many}) but no longer using bipartiteness. Suppose $u,v,w$ are consecutive. Then $u$ and $w$ must have another common neighbour to avoid an $H_2$. Now we violate Lemma~\ref{lem:s3c-c4-simple}. Thus, $u,v,w$ may not be consecutive.

Suppose now that $u$ and $w$ have two common neighbours $p$ and $q$. According to Lemma~\ref{lem:s3c-c4-simple}, both of these have degree $2$ and neither of these may have an edge to $v$, therefore $v$ has a distinct neighbour $z$. 
Suppose either $u$ or $v$ has degree greater than $3$, \mbox{w.l.o.g.} $u$. Then there is an $H_{2}$ with $u,p,w$ forming the central path. Thus, $u$ and $w$ are of degree $3$ and $v$ is a cut vertex.  If either of the paths from $u$ to $v$ or $v$ to $w$ is even then we have an $H_{2i}$. Thus we may assume both are odd. If $z$ has degree greater than $2$ then the paths from $v$ to $u$ and $v$ to $w$ must be even to avoid an $H_{2i}$ with central path (e.g.) $z,v,\ldots,u$. But then there is an $H_{2i}$ with central path $v,\ldots,u$. Thus, $z$ has degree $1$ or $2$. Indeed, we may follow a path $z=z_1$, $z_2$ etc. and never find a vertex of degree greater than $2$, unless it returns to $v$, as it would imply either an $H_{2i}$ with central path $z_j,\ldots,v$ (if $j$ even) or (e.g.) $z_j,\ldots,v,\ldots,u$ (if $j$ is odd). Thus, either the path ends or it returns to $v$ as a cycle. And there may be more than one such path or cycle. In this case, $G$ has three $3^+$-vertices. 

Suppose now that $u$ and $w$ have one common neighbour $p$ and two distinct neighbours $u'$ and $w'$, respectively. Since $u,v,w$ are not consecutive, there is an $H_2$ using $u,p,w$ as the central path.

Suppose now that $u$ and $w$ have neighbours $u',u''$ and $w',w''$, respectively, so that these vertices are pairwise distinct. If the total length of the path on $u,v,w$ is even then we have some $H_{2i}$. W.l.o.g., assume the path from $u$ to $v$ is even, whereupon the path from $v$ to $w$ must be odd. There must be an edge from $v$ to one of $u',u''$ (\mbox{w.l.o.g.} $u'$) to avoid an $H_{2i}$. By minimality of the path $u,v,w$, the distance from $u$ to $v$ along the path we chose must be $2$ with middle vertex $s$. According to Lemma~\ref{lem:s3c-c4-simple}, both $s$ and $u'$ have degree $2$ in $G$. If $u$ has degree greater than $3$ there is an $H_2$ with middle path $u,s,v$ (note that there cannot be an edge from $u$ to the neighbour $d$ of $v$ along the path to $w$, by our original minimal choice of the path along $u,v,w$).  Suppose $u''$ has no neighbours outside of $\{u,u',s,v\}$ (noting that there might be an edge from $u''$ to $v$). Then $v$ is a cut vertex. 
Note that none of the paths emanating from $w$ and not going through $v$ can have vertices of degree greater than $2$, unless they are part of a squash racket whose distinguished vertex is identified with $w$, as we would introduce an $H_{2i}$. The squash racket case can be dealt with just like a cycle so we do not concern ourselves with it specially. Note however that there can be only one squash racket attached at $w$ and if there is a squash racket there can be no additional cycles or paths. If the edge $u''$ to $v$ exists then there can be no new neighbours of $u''$ without introducing an $\mathbb{A}$ as a subgraph.
In this case, $G$ has three $3^+$-vertices.

Suppose now that $u''$ has a neighbour outside $\{u,u',s,v\}$ which we call $t$ (noting that in this case there can be no edge from $u''$ to $v$ since then there would be three consecutive $3^+$-vertices). Let us  continue expanding from $t$ along some path away from $u''$ which we enumerate $t=t_1,t_2,\ldots$. Suppose no such path like this reaches $w$ (or $v$) avoiding $u$. Note that such a path cannot now arrive at either $v$ or $s$ since we have uncovered all of their neighbours (if $v$ had more neighbours then there would be an $H_2$ with central path $u,s,v$). If one of the vertices $t_j$ has degree greater than $2$ then we have a $H_{2i}$ involving a central path either from $t_j$ to $v$ via $u$ or from $t_j$ to $w$ via $u$. Thus we have nothing more than a path here which must terminate at some $t_{k}$. In this case, $G$ has three $3^+$-vertices.

Suppose we have a path $t_1,t_2,\ldots$ that eventually joins the path between $v$ and $w$ exclusive of $v$ and $w$. Then we contradict minimality of the path from $u$ to $w$ via $v$. Thus we have a path $t_1,t_2,\ldots$ that eventually joins $w$. (Possibly it did so via, w.l.o.g., $w''$.) This path must be of odd length since otherwise it would produce an $H_{2i}$. 
If any vertex on this path has degree greater than $2$ then it induces some $H_{2i}$ on a path to either $u$ or $w$. Except in the case that it is distance $2$ from $w$, an odd distance from $u$, and shares two neighbours with $w$ (and these are $w$'s only neighbours).
In this case, $G$ has four $3^+$-vertices.

Therefore, there is nothing but this path, except that we did not isolate the degree of $w$. Indeed, $w$ may have arbitrarily high degree. Any path emanating from $w$ and not eventually reaching $v$ (we have accounted for all such paths) may not have a vertex $c$ of degree greater than $2$, since otherwise this vertex would produce an $H_{2i}$ with middle path to either $w$ or $v$. Except in the case that $c$ is at distance $2$ from $w$ and shares two neighbours the two neighbours $w'$ and $w''$ with $w$, in which case $c$ may have a path emanating from it in the form of $c=c_1$, $c_2$, etc. This case, in which there is a squash racket with distinguished vertex identified with $w$, will be superseded by the next. Note however that there can be only one squash racket attached at $w$ and if there is a squash racket there can be no additional cycles or paths. It follows then that some path may leave $w$ and terminate in a vertex of degree $1$, or return to $w$ producing a cycle. In this case, $G$ has three $3^+$-vertices.
\end{proof}

\noindent
We now prove a more general lemma. For an integer $d\geq 0$, let ${\cal G}_d$ be the class of graphs that consists of graphs, in which each connected component has at most $d$ $3^+$ vertices.

\begin{lemma}\label{lem:general-star-3-col-few}
For every $d\geq 1$, a graph from ${\cal G}_d$ is star $3$-colourable if and only if it is ${\cal F}'$-subgraph-free for some finite set of graphs ${\cal F}'$.
\end{lemma}

\begin{proof}
Let us give the proof for $d=2$ where we consider a connected graph $G$ from ${\cal G}_2$ that has exactly two $3^+$-vertices before we make the obvious generalisation. Let $u$ and $v$ be the $3^+$-vertices of $G$. Let $P_x$ be the set of paths from $x \in \{u,v\}$ that terminate with a vertex of degree $1$. Let $C_x$ be the set of paths from $x \in \{u,v\}$ that return to $x$ making a cycle. Let $P_{u,v}$ be the set of paths from $u$ to $v$. These account for all vertices of $G$. Recall that star colouring has the property that it can see up to distance $3$ from a certain vertex in terms of the effect of colours. Consider the paths in $P_{u,v}$ that are of length $\geq 3+4+3=10$. In any star $3$-colouring of a path of length $\geq 10$, the middle four vertices must contain a consecutive run of three distinct colours. It follows that a path of length $10$ covers all longer paths of length $1 \bmod 3$. Then, paths of length $11$ and $12$ cover all longer paths of length $2$ and $0 \bmod 3$, respectively. Thus, we may w.l.o.g. consider that the paths in $P_{u,v}$ are of length at most $12$. Now, suppose there are more than $N=3^6$ of some particular length $m$ of path. Then in any star $3$-colouring of these, some sequence of three colours at the beginning and at the end must appear more than once. Since it appears more than once it could be repeated if there were additional paths of length $m$ over and above the $N$ without violating the condition of star $3$-colouring. In this fashion, we may assume that there are at most $3^6$ copies of a path of any length. But now we reduced our graph to one of bounded size $\leq 5 \cdot 12 \cdot 12 \cdot 3^6$ (this is: \#sets, \#path-lengths, \#vertices, \#end-types, respectively) and we can check using brute force whether it is star $3$-colourable. If it is not star $3$-colourable then may add it to our list of forbidden subgraphs.

Now, when $d$ is larger, we consider $P_x, C_x, P_{x,y}$ for all vertices $x$ and $y$ giving a quadratic number in $d$ of sets to consider. The same bound on the length of paths applies as in the case $d=2$. The argument again reduces any graph to a graph of size at most $f(d)$ for some function $f$ which is $O(d^2)$ and the result follows as in the case $d=2$.
\end{proof}

\noindent
We are now ready to prove Theorem~\ref{t-general}, which we restate below.

\medskip
\noindent
{\bf Theorem~\ref{t-general} (restated).}
{\it A $(\H_2,\H_4,\H_6\ldots)$-subgraph-free~graph is star $3$-colourable if and only if it is ${\cal F}$-subgraph-free for some \emph{finite} set of graphs ${\cal F}$.} 

\begin{proof}
Let $G$ be a $(\H_2,\H_4,\H_6\ldots)$-subgraph-free~graph. We may assume without loss of generality that $G$ is connected. First suppose $G$ has at least five $3^+$-vertices. Then, by Lemma~\ref{lem:general-star-3-col-many}, we find that $G$ contains a subgraph isomorphic to $\mathbb{A}$ or $C_5$, or $G$ contains a $C_4$ with at least three $3^+$ vertices. Recall from Lemma~\ref{l-not} that ${\mathbb A}$ is not star $3$-colourable. It is readily seen that $C_5$ is not star 3-colourable. By Lemma~\ref{lem:s3c-c4-simple}, any graph that contains a $C_4$ with at least three $3^+$ vertices is not star $3$-colourable either. If $G$ has fewer than five $3^+$-vertices, then we apply
Lemma~\ref{lem:general-star-3-col-few}.
\end{proof}

\section{Acyclic ${\mathbf{3}}$-Colouring}\label{a-acyclic}

The {\sc Acyclic $3$-Colouring} problem is to decide if a graph $G$ has an {\it acyclic $3$-colouring}, which is a mapping $f:V(G)\to \{1,2,3\}$ such that for every $i$, the set $U_i$ of vertices of $G$ mapped to $i$ is independent (so, $f$ is a {\it $3$-colouring}) and $U_1\cup U_2$, $U_1\cup U_3$, $U_2\cup U_3$ all induce forests. Note that every star colouring is acyclic, while the reverse statement does not hold.

We note that {\sc Acyclic $3$-Colouring} is a C1-problem that does not satisfy C3. This is because the problem is polynomial-time solvable for graphs of bounded treewidth due to Courcelle's Theorem~\cite{Co90}, so C1 is satisfied.
Moreover, it follows from Proposition~\ref{p-sc3} that for all $p\geq 3$, the $p$-subdivision of a graph $G$ is star $3$-colourable and thus acyclic $3$-colourable, so C3 is not satisfied. As mentioned in Section~\ref{s-con}, the complexity of {\sc Acyclic $3$-Colouring} is still open for subcubic graphs, so it is not known if the problem satisfies C2.

We now prove a result using the same construction as in the corresponding result for {\sc Star $3$-Colouring} (Theorem~\ref{t-s3c}).

\begin{proposition}\label{p-acyclic}
{\sc Acyclic $3$-Colouring} is \NP-complete for planar bipartite graphs in which one partition class has size~$2$.
\end{proposition}

\begin{proof}
Let $G$ and $G'$ be the graphs as defined in the proof of Theorem~\ref{t-s3c}.
We claim that $G$ has a $3$-colouring if and only if $G'$ has an acyclic $3$-colouring.
First suppose $G$ has a $3$-colouring. From the proof of Theorem~\ref{t-s3c} we find that $G'$ has a star $3$-colouring, which is an acyclic $3$-colouring by definition. Now suppose $G'$ has an acyclic $3$-colouring $c$. Then no two adjacent vertices $u$ and $v$ in $G$ can be coloured alike by $c$. Hence, the restriction of $c$ to $G$ is a $3$-colouring.
\end{proof}

Our next result immediately follows from previous results.

\begin{theorem}\label{thm:main-acyclic}
{\sc Acyclic $3$-Colouring} is polynomial-time solvable for $(\H_{\ell},\H_{\ell+1},\ldots)$-subgraph-free graphs $(\ell\geq 1)$ and $(\H_i,\H_{2i},\H_{3i},\ldots)$-subgraph-free graphs $(i\geq 1)$, but \NP-complete for  $(\H_1,\H_3,\H_5\ldots)$-subgraph-free~graphs.\\[-5mm] 
\end{theorem}

\begin{proof}
The polynomial results follow from combining the above observation that the problem satisfies C1 with Propositions~\ref{p-h} and~\ref{p-h2}, respectively. The \NP-completeness result follows from Proposition~\ref{p-acyclic}.
\end{proof}


\begin{thebibliography}{10}

\bibitem{ACKKR04}
Michael~O. Albertson, Glenn~G. Chappell, Henry~A. Kierstead, Andr{\'{e}}
  K{\"{u}}ndgen, and Radhika Ramamurthi.
\newblock Coloring with no $2$-colored $P_4$'s.
\newblock {\em Electronic Journal of Combinatorics}, 11(1), 2004.

\bibitem{ABKL07}
Vladimir~E. Alekseev, Rodica Boliac, Dmitry~V. Korobitsyn, and Vadim~V. Lozin.
\newblock N{P}-hard graph problems and boundary classes of graphs.
\newblock {\em Theoretical Computer Science}, 389:219--236, 2007.

\bibitem{AK90}
Vladimir~E. Alekseev and Dmitry~V. Korobitsyn.
\newblock Complexity of some problems on hereditary graph classes.
\newblock {\em Diskretnaya Matematika}, 2:90--96, 1990.

\bibitem{AK92}
Vladimir~E. Alekseev and Dmitry~V.   Korobitsyn. Complexity of some problems on hereditary
  graph classes. {\em Diskretnaya Matematika}, 4:34--40 (1992)

\bibitem{AP89}
Stefan Arnborg and Andrzej Proskurowski.
\newblock Linear time algorithms for {N}{P}-hard problems restricted to partial
  $k$-trees.
\newblock {\em Discrete Applied Mathematics}, 23:11--24, 1989.

\bibitem{Bi91}
Daniel Bienstock. On the complexity of testing for odd holes and induced odd paths. {\em Discrete Mathematics}, 90:85--92, 1991 (see also Corrigendum, {\em Discrete Mathematics} 102:109, 1992).

\bibitem{BRST91}
Daniel Bienstock, Neil Robertson, Paul D. Seymour, and Robin Thomas. Quickly excluding a
  forest. {\em Journal of Combinatoral Theory, Series {B}}, 52:274--283, 1991.
  
\bibitem{BHKKOOZ20}
Hans L. Bodlaender, Tesshu Hanaka, Yasuaki Kobayashi, Yusuke Kobayashi, Yoshio Okamoto, Yota Otachi, Tom C. van der Zanden.
Subgraph Isomorphism on Graph Classes that Exclude a Substructure. {\em Algorithmica}, 82:3566--3587, 2020.

\bibitem{BJMOPPSV}
Hans L. Bodlaender, Matthew Johnson, Barnaby Martin, Jelle J. Oostveen, Sukanya Pandey, Dani\"el Paulusma, Siani Smith and Erik Jan van Leeuwen.
Complexity framework for forbidden subgraphs IV: The Steiner Forest problem. {\em CoRR},  abs/2305.01613, 2023.
 
\bibitem{BL02}
Rodica Boliac and Vadim~V. Lozin.
\newblock On the clique-width of graphs in hereditary classes.
\newblock {\em Proc. ISAAC 2022, LNCS}, 2518:44--54, 2002.

\bibitem{BGMOPS22}
Christoph Brause, Petr A. Golovach, Barnaby Martin, Pascal Ochem, Daniël Paulusma, and Siani Smith.
Acyclic, Star, and Injective Colouring: Bounding the diameter. {\em Electronic Journal of Combinatorics}, 29:P2.43, 2022.

\bibitem{CHRSZ19}
Maria Chudnovsky, Shenwei Huang, Pawel Rzazewski, Sophie Spirkl, and Mingxian
  Zhong.
\newblock Complexity of ${C}_k$-coloring in hereditary classes of graphs.
\newblock {\em Proc. ESA 2019, LIPIcs}, 144:31:1--31:15, 2019.

\bibitem{Co90}
Bruno Courcelle.
\newblock The monadic second-order logic of graphs. {I}. {R}ecognizable sets of
  finite graphs.
\newblock {\em Information and Computation}, 85:12--75, 1990.

\bibitem{Da80}
David P. Dailey. Uniqueness of colorability and colorability of planar $4$-regular graphs are \NP-complete. {\em Discrete Mathematics}, 30:289--293, 1980.

\bibitem{DP89}
Rina Dechter and Judea Pearl.
\newblock Tree clustering for constraint networks.
\newblock {\em Artificial Intelligence}, 38:353--366, 1989.

\bibitem{GraphTheory}
Reinhard Diestel.
\newblock {\em Graph theory}, volume 173 of {\em Graduate Texts in
  Mathematics}.
\newblock Springer-Verlag, New York, 1997.

\bibitem{EHK98}
Thomas Emden-Weinert, Stefan Hougardy and Bernd Kreuter. Uniquely colourable graphs and the hardness of colouring graphs of large girth. {\em Combinatorics, Probability \& Computing}, 7:375--386, 1998.

\bibitem{GHN00}
Anna Galluccio, Pavol Hell, and Jaroslav Ne\v{s}et\v{r}il.
\newblock The complexity of ${H}$-colouring of bounded degree graphs.
\newblock {\em Discrete Mathematics}, 222:101--109, 2000.

\bibitem{GJS76}
Michael~R. Garey, David~S. Johnson, and Larry~J. Stockmeyer.
\newblock Some simplified {N}{P}-complete graph problems.
\newblock {\em Theoretical Computer Science}, 1:237--267, 1976.

\bibitem{GJT76}
Michael~R. Garey, David~S. Johnson, and Robert~Endre Tarjan.
\newblock The {P}lanar {H}amiltonian {C}ircuit problem is {N}{P}-complete.
\newblock {\em {SIAM} Journal on Computing}, 5:704--714, 1976.

\bibitem{GR09}
Alfred Geroldinger and Imre Z. Ruzsa.
Combinatorial Number Theory and Additive Group Theory.
Birkh\"auser, 2009.

\bibitem{GP14}
Petr~A. Golovach and Dani{\"{e}}l Paulusma.
\newblock List coloring in the absence of two subgraphs.
\newblock {\em Discrete Applied Mathematics}, 166:123--130, 2014.

\bibitem{GPR15}
Petr~A. Golovach, Dani{\"{e}}l Paulusma, and Bernard Ries.
\newblock Coloring graphs characterized by a forbidden subgraph.
\newblock {\em Discrete Applied Mathematics}, 180:101--110, 2015.

\bibitem{HN90}
Pavol Hell and Jaroslav Nesetril.
\newblock On the complexity of ${H}$-coloring.
\newblock {\em Journal of Combinatorial Theory, Series {B}}, 48:92--110, 1990.

\bibitem{JMOPPSV}
Matthew Johnson, Barnaby Martin, Jelle~J. Oostveen, Sukanya Pandey, Dani\"el
  Paulusma, Siani Smith, and Erik~Jan van Leeuwen.
\newblock Complexity framework for forbidden subgraphs {I}: The framework.
\newblock {\em CoRR}, 2211.12887, 2022.

\bibitem{JMPPSV23}
Matthew Johnson,  Barnaby Martin, Sukanya Pandey, Dani\"el Paulusma, Siani Smith, and Erik Jan van Leeuwen,
Complexity framework for forbidden subgraphs III: When problems are
  tractable on subcubic graphs.
{\em  Proc. MFCS 2023, LIPIcs}, 272:57:1-57:15, 2023.

\bibitem{Ka12}
Marcin Kami\'nski. Max-{C}ut and containment relations in graphs. 
{\em Theoretical Computer Science},  438:89--95, 2012.

\bibitem{KLMT11}
Nicholas Korpelainen, Vadim~V. Lozin, Dmitriy~S. Malyshev, and Alexander
  Tiskin.
\newblock Boundary properties of graphs for algorithmic graph problems.
\newblock {\em Theoretical Computer Science}, 412:3545--3554, 2011.

\bibitem{LLMT09}
Benjamin L{\'{e}}v{\^{e}}que, David~Y. Lin, Fr{\'{e}}d{\'{e}}ric Maffray, and
  Nicolas Trotignon.
\newblock Detecting induced subgraphs.
\newblock {\em Discrete Applied Mathematics}, 157:3540--3551, 2009.

\bibitem{LR22}
Vadim V. Lozin and Igor Razgon, Tree-width dichotomy. {\em European Journal of
  Combinatorics}, 103:103517, 2022.

\bibitem{MPSL23}
Barnaby Martin, Dani{\"{e}}l Paulusma, Siani Smith, and Erik~Jan van Leeuwen.
\newblock Few induced disjoint paths for ${H}$-free graphs.
\newblock {\em Theoretical Computer Science}, 939:82--193, 2023.

\bibitem{RS84}
Neil Robertson and Paul D. Seymour. Graph minors. {I}{I}{I}. {P}lanar tree-width.
{\em  Journal of Combinatorial Theory, Series {B}}, 36:49--64, 1984.

\bibitem{RS86}
Neil Robertson and Paul D. Seymour. Graph minors. {V}. {E}xcluding a planar graph.
 {\em Journal of Combinatorial Theory, Series {B}},  41:92--114, 1986.

\bibitem{RS95}
Neil Robertson and Paul~D. Seymour.
\newblock Graph minors. {X}{I}{I}{I}. {T}he {D}isjoint {P}aths problem.
\newblock {\em Journal of Combinatorial Theory, Series {B}}, 63:65--110, 1995.

\bibitem{Sc94}
Petra Scheffler.
\newblock A practical linear time algorithm for disjoint paths in graphs with  bounded tree-width.
\newblock Preprint 396, Technische Universit\"at Berlin, Institut f\"ur  Mathematik, 1994.

\bibitem{SA24}
M. A. Shalu, Cyriac Antony.
Hardness transitions and uniqueness of acyclic colouring. {\em Discrete Applied Mathematics}. 345:77--98, 2024.

\bibitem{SA22}
M. A. Shalu, Cyriac Antony.
Star colouring of bounded degree graphs and regular graphs. {\em Discrete Mathematics}, 345:112850, 2022.

\bibitem{Sh80}
Yossi Shiloach.
\newblock A polynomial solution to the undirected {Tw}o {P}aths problem.
\newblock {\em Journal of the {ACM}}, 27:445--456, 1980.

\end{thebibliography}
\end{document}